\documentclass[journal]{IEEEtran}
\makeatletter
\let\NAT@parse\undefined
\makeatother
\pdfoutput=1

\usepackage{graphics} 
\usepackage{epsfig} 
\usepackage{mathptmx} 
\usepackage{times} 
\usepackage{amsmath} 
\usepackage{amssymb}  
\usepackage{xcolor}  
\usepackage{caption} 
\usepackage{amsthm}
\usepackage{subcaption}
\usepackage{graphicx}
\usepackage{multirow}
\usepackage{booktabs}
\usepackage{MnSymbol}

\usepackage[hyphens]{url} %
\newcommand{\sli}[1]{\textcolor{black}{#1}} 
\newcommand{\slired}[1]{\textcolor{black}{#1}} 
\newcommand{\slirebuttal}[1]{\textcolor{black}{#1}} 

\usepackage[numbers,compress]{natbib}

\newcommand{\squeeze}{-1.5mm}

\usepackage{enumitem}
\newtheoremstyle{mystyle}
  {}
  {}
  {}
  {}
  {\bfseries}
  {.}
  {.5em}
  {}
  
\newtheorem{definition}{Definition}
\newtheorem{assumption}{Assumption}
\newtheorem{proposition}{Proposition}
\newtheorem{lemma}{Lemma}
\newtheorem{theorem}{Theorem}
\newtheorem{corollary}{Corollary}
\theoremstyle{mystyle}{
    \newtheorem*{remark}{Remark}
}

\newtheorem{controller}{Controller}
\newtheorem{observation}{Observation}

\makeatletter
\def\subsubsection{\@startsection{subsubsection}
                                 {3}
                                 {\z@}
                                 {0ex plus 0.1ex minus 0.1ex}
                                 {0ex}
                                 {\normalfont\normalsize\itshape}}
\makeatother
                                 
\title{
Hybrid System Stability Analysis of Multi-Lane Mixed-Autonomy Traffic}

\author{Sirui Li, Roy Dong, Cathy Wu
\thanks{This work was supported by a gift from Mathworks and MIT’s Research Support Committee.}
\thanks{Sirui Li is with the Institute for Data, Systems, and Society, Massachusetts Institute of Technology,
Cambridge, MA, 02139, USA.
    {\tt\small siruil@mit.edu}}%
\thanks{Roy Dong is with the Industrial \& Enterprise Systems Engineering department at the University of Illinois at Urbana-Champaign, Urbana, IL, 61801, USA.
    {\tt\small roydong@illinois.edu}}
\thanks{Cathy Wu is with the Laboratory for Information \& Decision Systems; the Institute for Data, Systems, and Society; and the Department of Civil and Environmental Engineering, Massachusetts Institute of Technology,
        Cambridge, MA, 02139, USA.
        {\tt\small cathywu@mit.edu}}%
}

\begin{document}

\maketitle

\begin{abstract}
Autonomous vehicles (AVs) hold vast potential to enhance transportation systems by reducing congestion, improving safety, and lowering emissions. AV controls lead to emergent traffic phenomena; one such intriguing phenomenon is traffic breaks (rolling roadblocks), where a single AV efficiently stabilizes multiple lanes through frequent lane switching, similar to the highway patrolling officers weaving across multiple lanes during difficult traffic conditions. While previous theoretical studies focus on single-lane mixed-autonomy systems, this work proposes a stability analysis framework for multi-lane systems under AV controls. Casting this problem into the hybrid system paradigm, the proposed analysis integrates continuous vehicle dynamics and discrete jumps from AV lane-switches. Through examining the influence of the lane-switch frequency on the system's stability, the analysis offers a principled explanation to the traffic break phenomena, and further discovers opportunities for less-intrusive traffic smoothing by employing less frequent lane-switching. The analysis further facilitates the design of traffic-aware AV lane-switch strategies to enhance system stability. Numerical analysis reveals a strong alignment between the theory and simulation, validating the effectiveness of the proposed stability framework in analyzing multi-lane mixed-autonomy traffic systems.
\end{abstract}
\begin{IEEEkeywords}
Intelligent Transportation Systems, Autonomous Agents, Hybrid Systems, Hybrid Logical/Dynamical Planning and Verification
\vspace*{-0.3cm}
\end{IEEEkeywords}


\section{INTRODUCTION}
\label{sec:introduction}
Efficient and eco-friendly transportation systems are crucial for economic vitality, public health and safety, and environmental sustainability~\cite{national2019traffic, epaghgemission, winston2015transportation}. With transportation contributing to 29\% of the U.S. GHG emission in 2021~\cite{national2019traffic}, studies indicate that mitigating congestion could potentially lead to a nearly 20\% decrease in CO$_2$ emissions. Distinct traffic phenomena emerge when congestion mitigation strategies are put into practice. For example, under difficult traffic conditions such as car accidents, high congestion, or road construction, highway patrol officers employ a strategy known as a ``traffic break" (rolling roadblock)~\cite{fhwa2019,californiahighway} to clear traffic by weaving across multiple lanes. Despite being a common practice, the impact of traffic break on the system's stability is understudied; limited preliminary empirical study on the effect on safety and mobility appear in the literature~\cite{saha2021application,nafakh2022safety}.

\begin{figure}[!t]
    \includegraphics[width=0.49\textwidth]{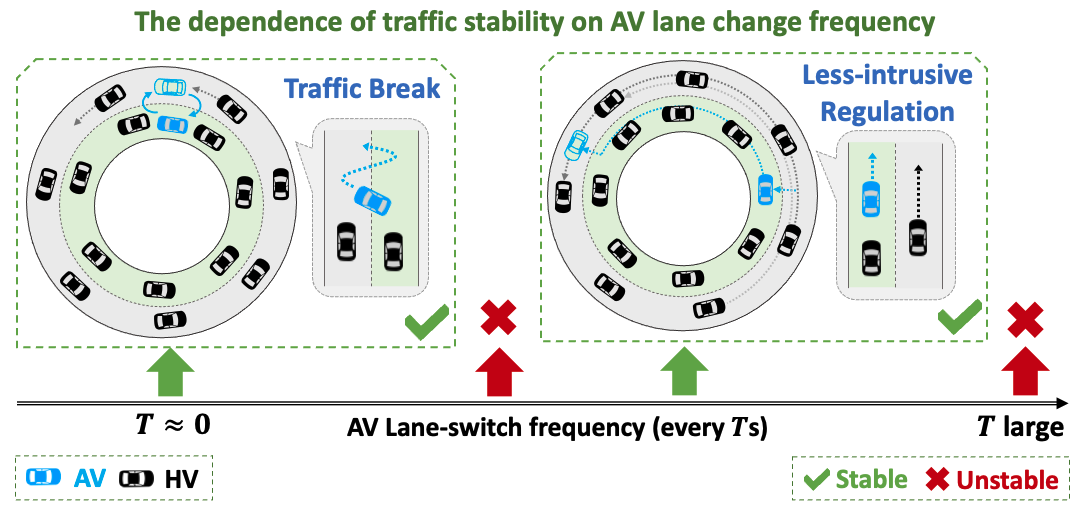}
    \caption{\textbf{Overview.} Consider a multi-lane mixed-autonomy traffic system with human-driven vehicles on each lane (HV, in black), and an autonomous vehicle (AV, in blue) switching between lanes. This work proposes a theoretical framework to analyze the stability of the multi-lane system, uncovering emergent phenomena based on different AV lane-switch frequencies: the AV can stabilize the traffic by implementing traffic breaks with $T \approx 0$ or less-intrusive regulation with a moderate $T$ (the green check-marks); in contrast, traffic becomes unstable when $T$ is small or too large (the red crosses between and outside the two stable regimes).}
    \vspace*{-0.3cm}
    \label{fig:illustration}
\end{figure}

Emerging technologies such as autonomous vehicles (AVs) have the potential to revolutionize transportation by mitigating traffic congestion and alleviating environmental impacts. A field experiment demonstrates that a single AV can stabilize a single-lane circular track with 22 vehicles~\cite{stern2018dissipation}. Simulations further display the emergent phenomena of AVs trained through deep reinforcement learning to stabilize the mixed-autonomy systems under various traffic network topologies~\cite{wu2021flow, yan2022unified}, including the ``traffic break'' where a single AV simultaneously stabilizes two-lane circular tracks by rapidly switching between the lanes.
AVs offer more precise and advanced controls, and are potentially more scalable compared to highway patrolling officers, who requires specialized training to execute traffic-smoothing techniques. This presents an opportunity to systematically deploy AVs to enhance traffic system stability through their precise control of these special techniques.

While ample studies provide stability analysis of a single-lane mixed autonomy traffic~\cite{cui2017stabilizing, zheng2020smoothing, zheng2015stability, zhu2018analysis, wang2020controllability, mousavi2022synthesis, swaroop1994string, bose2003analysis, rogge2008vehicle, wu2018stabilizing, giammarino2020traffic, liu2022structural}, the stability of an AV-controlled multi-lane systems have been rarely considered. The goal of this work is hence to present a theoretical framework for stability analysis of multi-lane mixed autonomy, with a focus on a single AV simultaneously stabilizing two-lane ring-roads by switching between the two lanes. Casting the problem into the hybrid system framework, this work proposes a variance-based Lyapunov analysis to quantify the stability during both continuous operation (when the AV is in a given lane) and discrete jumps (when the AV switches lanes).

The analysis provides a principled theoretical understanding to the emergent traffic phenomena. For instance, when the AV frequent switches lanes, it produces a ``phantom car" effect akin to the traffic break phenomenon, where the AV appears to duplicate itself to maintain stability in both lanes, as illustrated in Fig.~\ref{fig:illustration}. We further analyze the influence of reducing the switch frequency on the system's stability, considering that a reduced frequency leads to more practical deployment with less traffic disruption and enhanced passenger comfort. The stability analysis in this work uncovers a sweet spot for the switch frequency, which strikes a balance between duration that are neither too short (resulting in insufficient control duration) nor too long (leading to uncontrolled instability blowup) beyond the phantom car regime. The analysis further contributes to designing traffic-aware AV lane-switch strategies to enhance the stability of the multi-lane hybrid system.

In summary, our contributions are
\begin{itemize}
    \item We propose a theoretical framework for the stability analysis of multi-lane mixed-autonomy traffic systems.
    \item The analysis recovers emergent phenomena, such as traffic breaks, under various AV lane-switch frequencies.
    \item The analysis enables the design of traffic-aware AV lane-switch strategies to further improve traffic stability.
    \item We conduct extensive numerical analysis to validate that our theoretical analysis closely aligns with simulations.
\end{itemize}

\sli{The structure of the paper is as follows: Sec.~\ref{sec:related_work} discusses related works. Sec.~\ref{sec:prelim} introduces the multi-lane traffic model, which combines continuous vehicle dynamics \slirebuttal{with} AV lane-switches modeled as discrete jumps. \slirebuttal{Sec.~\ref{sec:analysis} presents} a Lyapunov stability analysis of the system under the hybrid-system framework\slirebuttal{, with detailed derivation in Sec.~\ref{sec:variance_derivation}}. The analysis is applied in Sec.~\ref{sec:controller_fixed_duration} to explain the emergent traffic phenomena depicted in Fig.~\ref{fig:illustration}, and in Sec.~\ref{sec:controller} to design controllers to improve system's stability. Finally, Sec.~\ref{sec:experiment} validates the stability analysis through extensive numerical experiments.}

\section{RELATED WORK}
\label{sec:related_work}
\subsection{Traffic Stabilization with Autonomous Vehicles}
Controlling AVs to stabilize mixed-autonomy systems has garnered considerable attention from both empirical and theoretical perspectives. Empirical studies include field experiments~\cite{stern2018dissipation} and deep reinforcement learning simulations where controlling a small number of AVs can stabilize traffic~\cite{wu2021flow}, reduce fuel consumptions~\cite{jayawardana2022learning}, and exhibit emergent behaviors such as traffic lights~\cite{yan2022unified}. Theoretical studies have majorly been conducted on the linearized Optimal Velocity Model (OVM)~\cite{bando1995dynamical} for the single-lane ring-road traffic setting, under asymptotic stability~\cite{cui2017stabilizing, zheng2020smoothing, zheng2015stability, zhu2018analysis, wang2020controllability, mousavi2022synthesis}, or string stability~\cite{swaroop1994string, bose2003analysis, rogge2008vehicle, wu2018stabilizing, giammarino2020traffic, liu2022structural}. Recent theoretical studies extend to disturbance and uncertainty~\slirebuttal{\cite{mousavi2022synthesis, jin2016optimal, liu2023reachability}}, sector nonlinearity~\cite{gisolo2022nonlinear}, and deploying AV policies under coarse-grained guidance~\cite{li2023stabilization}. However, previous theoretical works focus on the stability of a single-lane traffic setting, with very few works consider multi-lane traffic systems. This work extends the previous analysis to multi-lane mixed-autonomy systems, investigating the capacity of a single AV to simultaneously stabilize multiple lanes. 

\vspace{-0.2cm}
\subsection{Switched and Hybrid Systems}
There exists a rich literature on switched and hybrid systems~\cite{liberzon2003switching, shorten2007stability, zhu2015optimal}, which are prevalent in real-world systems~\cite{frazzoli2000robust, gans2007stable, schill2018robust} and are characterized by a set of continuous modes with discrete transitions among the modes. Associated stability analysis techniques involve common Lyapunov function~\cite{shorten1998stability, liberzon2004common, williams2013constrained}, multiple Lyapunov function~\cite{branicky1998multiple, long2017multiple}, and slow-switching conditions based on the dwell-time~\cite{yuan2014hybrid, zhao2014switching}. Inspired by these stability analysis techniques, this work introduces a variance-based Lyapunov analysis \slirebuttal{that upper bounds} stability changes in \slirebuttal{both} continuous dynamics (controlled and uncontrolled lanes) and transitions via discrete jumps (AV lane-switches).

\vspace{-0.2cm}
\subsection{Lane Changing Maneuver}
\textbf{Lane Changing Models:} An essential component of the proposed analysis is the AV lane-switch; this is related to previous literature on lane changing models~\cite{moridpour2010lane, zheng2014recent}, which consider different levels of complexity and scales (macroscopic~\cite{tang2009macroscopic, jin2010macroscopic} vs. microscopic~\cite{rahman2013review, halati1997corsim, ahmed1996models}) and different lane changing rules (discretionary~\cite{knoop2014calibration} v.s. mandatory~\cite{zhang2019game}). Certain studies also explore the control of vehicles during lane changing maneuvers through Model Predictive Control~\cite{kamal2015efficient, dixit2019trajectory}, Monte Carlo Tree Search~\cite{lenz2016tactical}\sli{, or Belief Tree Search under the POMDP framework~\cite{lauri2022partially, ding2021epsilon, ulbrich2013probabilistic}.} 

However, these previous works primarily center on modeling human drivers' lane changing behaviors for realistic traffic simulations, or enhancing the safety and efficiency of the ego lane changing vehicle. In contrast, this work focuses on AV lane-switch with the intent of enhancing the stability of the entire system. \sli{While this work models the AV lane-switch as discrete jumps (following prior works~\cite{piu2022stability, wu2017multi}) to simplify the theoretical analysis, we acknowledge that real-world lane changes involve continuous dynamics consisting of longitudinal and lateral maneuvers, which can be modeled by Dubins’ path~\cite{dubins1957curves} or Kinematic Bicycle Model~\cite{kong2015kinematic}.}

\textbf{Lane Changing Effects:} Numerous previous studies \sli{examine simulation and real-world traffic data} to show that lane changing executed by HVs can cause shockwave formulation and traffic slowdown~\cite{daganzo2002behavioral,jin2010kinematic,oh2015impact}. Similarly, in this work, AV lane switching can cause traffic bottlenecks as it may disrupt the traffic flow. The proposed analysis models the AV lane-switch as discrete jumps to assess its impact on the system stability.

\vspace{-0.2cm}
\subsection{Multi-lane Stability Analysis}
To the best of our knowledge, theoretical analysis on the multi-lane mixed-autonomy systems has not been considered. Previous relevant studies primarily focus on multi-lane systems with solely HVs, including macroscopic analyses~\cite{michalopoulos1984multilane, daganzo2002behavioral} that examine the cumulative effects from lane changing by HVs, and a microscopic study~\cite{wu2017multi} that reduces multi-lanes into a stochastic single-lane model \sli{in which vehicles appear and disappear based on human-driving data, and applies Markov Chain analysis to the reduced system. While the prior works focus on modeling of the multi-lane HV systems, this work focuses on the control aspect by introducing an AV to switch between lanes to stabilize the HVs, and derives Lyapunov stability analysis of the multi-lane mixed-autonomy system.}

The most related prior study~\cite{piu2022stability} examines a microscopic lane changing model for HVs on multi-lane ring-roads, where they consider simple lane change rules that account for the target speed and safety distances. The study identifies the equilibrium state of a multi-lane system, as well as regions around the equilibrium such that the HVs have no incentive to change lanes. We draw inspiration from the analysis of this work. \sli{Instead of finding conditions such that lane change never occurs (based on each HV's incentive), we allow and design AV controllers to switch lanes (frequently) as a means to stabilize the system through AV lane-switching.}

\section{MULTI-LANE MODEL}
\label{sec:prelim}
This work considers a two-lane ring road system with one AV $i=n$ and $n-1$ HVs $i=1,...,n-1$ on each lane, for a total of $2(n-1)$ HVs. Each lane has a circumference $C$ (see Fig.\ref{fig:illustration}). We assume the HVs stay on the same lane (no lane changing), while the AV can switch between lanes at any time. The two lanes are denoted as $\{L, R\}$ to represent the Left and Right lanes with respect to the traffic flow. We formulate the continuous vehicle dynamics for each lane and the discrete jump from AV lane-switch as follows.

\vspace{-0.2cm}
\subsection{Continuous vehicle dynamics}
\label{sec:prelim_singlelane}
At each time $t$, \slirebuttal{the AV is present on one of the two lanes $l \in \{L, R\}$. We denote the lane with AV presence ($n_l = n$ total number of vehicles) as the Controlled Lane, and the other lane ($n_l = n-1$ total number of vehicles) as the Uncontrolled lane. Equivalently, we say a lane is controlled (uncontrolled) by the presence (absence) of the AV on the lane.} For each lane $l$, we denote the position of the $i^{th}$ vehicle by $p_i^l(t)$, the velocity by $v^l(t) = \dot{p}^l(t)$, and the acceleration by $a^l(t) = \dot{v}^l(t)$, \slirebuttal{The vehicles are ordered such that  $mod(p^l_1(t), C) \geq ... \geq mod(p_{n_l}(t), C)$, with the headway $s_i^l(t) = mod(p_{i-1}^l(t) - p_i^l(t), C)$, where the modulo operation accounts for the circular redundancy of the vehicle's position on the ring-road}. The car following model (CFM) for each HV takes the nonlinear form
\vspace{\squeeze}
\begin{equation}
    \dot{v}^l_i(t) = F(s^l_i(t), \dot{s}^l_i(t), v^l_i(t)),
\vspace{\squeeze}
\end{equation}
where the uniform flow equilibrium is obtained at headway $s^{*}_{n_l}$ and velocity $v^{*}_{n_l}$ such that $F(s^{*}_{n_l}, 0, v^{*}_{n_l}) = 0$. Notably, $s^{*}_{n_l}$ and $v^{l*}$ vary with the number of vehicles $n_l$ on the lane. 

Following~\cite{zheng2020smoothing}, we denote the error state as $\tilde{s}^l_i(t) = s^l_i(t) - s^{*}_{n_l}$ and $\tilde{v}^l_i(t) = v^l_i(t) - v^{*}_{n_l}$. \slirebuttal{We assume that $F$ is differentiable around the equilibrium $(s^{*}_{n_l}, v^{*}_{n_l})$, and} consider linearization of the HV's CFM around the equilibrium  
\begin{equation}
    \begin{cases}
    \dot{\tilde{s}}^l_i(t) & = \tilde{v}^l_{i-1}(t) - \tilde{v}^l_i(t)\\
    \dot{\tilde{v}}^l_i(t) & = a_1^{l} \tilde{s}^l_i(t) - a_2^{l} \tilde{v}^l_i(t) + a_3^{l} \tilde{v}^l_{i-1}(t)
    \end{cases}
    \label{eq:hv_ode}
\end{equation}
where $a_1^{l} = \frac{\partial F}{\partial s}, a_2^{l} = \frac{\partial F}{\partial \dot{s}} - \frac{\partial F}{\partial v}, a_3^{l} = \frac{\partial F}{\partial \dot{s}}$ evaluated at $(s^{*}_{n_l}, v^{*}_{n_l})$. 

This work \slirebuttal{further follows~\cite{zheng2020smoothing}} to assume the HVs are following the Optimal Velocity Model (OVM)~\slirebuttal{\cite{orosz2010traffic, bando1995dynamical}}
\begin{equation}
    F(s^l_i(t), \dot{s}^l_i(t), v^l_i(t)) = \alpha(V(s^l_i(t)) - v^l_i(t)) + \beta \dot{s}^l_i(t),
    \label{eq:ovm}
\end{equation}
where $\alpha > 0, \beta > 0$, and the optimal velocity $V(s)$ is
\vspace{\squeeze}
\begin{equation}
    V(s) = \begin{cases}
    0, & s \leq s_{st}\\
    f_v(s), & s_{st} < s < s_{go}\\
    v_{max}, & s \geq s_{go}
    \end{cases}    ,
    \label{eq:vclipping}
\vspace{\squeeze}
\end{equation}
\slirebuttal{where the spacing thresholds $s_{st}$ and $s_{go}$ correspond to $0$ and $v_{max}$ in the optimal velocity function $V(s)$. In between, $f_v(s)$ characterizes the transition in target velocity from $0$ and $v_{max}$, where $f_v(s)$ is generally continuous and monotone increasing. Many choices of $f_v(s)$ exists~\cite{lazar2016review}, and a typical form is~\cite{zheng2020smoothing}}
\vspace{\squeeze}
\begin{equation}
    f_v(s) = \frac{v_{max}}{2}\Big(1 - \cos \big(\pi \frac{s - s_{st}}{s_{go} - s_{st}}\big)\Big).
    \label{eq:ovm_f}
\vspace{\squeeze}
\end{equation}
\slirebuttal{The OVM equilibrium for each ring road follow} $s^*_{n_l} = C / n_l$ and $v^{*}_{n_l} = V(s^{*}_{n_l})$\slirebuttal{; in particular, the equilibrium velocity $v^{*}_{n_l}$ is the target optimal velocity $V(s^{*}_{n_l})$ evaluated at the equilibrium headway $s^{*}_{n_l}$. The OVM linearization satisfies} $a_1^{l} = \alpha \dot{V}(s^*_{n_l}),  a_2^{l} = \alpha + \beta, a_3^{l} = \beta$. \slirebuttal{We refer the readers to Table~\ref{tab:notation_var} in Sec.~\ref{sec:experiment} for a list of the OVM parameters and their descriptions, and~\cite{zheng2020smoothing, li2023stabilization} for visualizations of Eq.~\eqref{eq:vclipping}.}
\\

\noindent \textbf{Controlled Lane.}
When a lane $l$ is under AV control, we directly control the acceleration of the AV with a input $u^l(t)$. The CFM of the AV is hence denoted as
\vspace{\squeeze}
\begin{equation}
    \begin{cases}
    \slirebuttal{\dot{\tilde{s}}^l_{n}(t)}& = \slirebuttal{\tilde{v}^l_{n-1}(t)} - \slirebuttal{\tilde{v}^l_{n}(t)}\\
    \slirebuttal{\dot{\tilde{v}}^l_{n}(t)}& = u^l(t)
    \end{cases}    ,
    \label{eq:av_ode}
\vspace{\squeeze}
\end{equation}
Denote the \textit{state vector} at any time $t$ as $z_c^l(t) = [s_1^l(t), v_1^l(t), ..., s_n^l(t), v_n^l(t)]$, and the corresponding \textit{error state vector} as $x^l_c(t) = [\tilde{s}^l_{1}(t), \tilde{v}^l_{1}(t), ..., \tilde{s}^l_{n}(t), \tilde{v}^l_{n}(t)]^\intercal = z_c^l(t) - z^{*}_n$, where $z^{*}_n = [s^{*}_n, v^{*}_n, ... s^{*}_n, v^{*}_n] \in \mathbb{R}^{1\times 2n}$ is the equilibrium state. The error dynamics of the linearized controlled system is
\vspace{\squeeze}
\begin{equation}
    \dot{x}^l_c(t) = A_c x^l_c(t) + B_c u^l(t)
    \label{eq:control_matrix_system}
\vspace{\squeeze}
\end{equation}
with 
\vspace{\squeeze}
\begin{equation}
    A_c = \begin{bmatrix} 
            D_1 & 0 & ... & ... & 0 & D_2 \\
        	D_2 & D_1 & 0 & ... & ... & 0\\
        	0 & D_2 & D_1 & 0 & ... & 0 \\
        	\vdots & \ddots & \ddots & \ddots & \ddots & \vdots\\
        	0 & ... & 0 & D_2 & D_1 & 0 \\
            0 & ... &  ... & 0 &  C_2 & C_1 \\

    	\end{bmatrix},
    B_c = \begin{bmatrix} 
        	B_2 \\
        	B_2 \\ 
        	B_2 \\
        	\vdots \\
        	B_1
    	\end{bmatrix} \in \mathbb{R}^{2n \times 2n},
\end{equation}
\begin{equation}
\begin{aligned}
    \text{where} \quad \quad D_1 = \begin{bmatrix}
            0 & -1 \\
            a_1 & - a_2
          \end{bmatrix} &,
    D_2 = \begin{bmatrix}
            0 & 1 \\
            0 & a_3
          \end{bmatrix}, \\[5pt]
    C_1 = \begin{bmatrix}
            0 & -1 \\
            0 & 0
          \end{bmatrix},
    C_2 = \begin{bmatrix}
            0 & 1 \\
            0 & 0
          \end{bmatrix} & ,
    B_1 = \begin{bmatrix}
            0 \\
            1
          \end{bmatrix}, 
    B_2 = \begin{bmatrix}
            0 \\
            0
          \end{bmatrix}.
    \label{eq:BCD_notation}
\end{aligned}
\end{equation}
\slirebuttal{The AV adopts a full state feedback control $u^l(t) = - K x^l_c(t)$. We assume the control gain matrix $K$ can stabilize the single-lane linearized OVM dynamics, i.e. the resulting system $\dot{x}^l_c(t) = (A_c - B_c K) x^l_c(t)$ is stable. The prior work~\cite{zheng2020smoothing} proves that such a matrix $K$ exists.} \sli{From the Converse Lyapunov Theorem~\cite{khalil2001nonlinear}, there exists a valid Lyapunov function $V(x^l_c(t)) = x^l_c(t)^\intercal P x^l_c(t) > 0$ with $\dot{V}(x^l_c(t)) = -x^l_c(t)^\intercal Q x^l_c(t) < 0$ and $-Q = (A_c - B_c K)P + P(A_c - B_c K)^\intercal$, where \slirebuttal{$P, Q \in \mathbb{R}^{n\times n} > 0$}.}\\

\noindent \textbf{Uncontrolled Lane.} \slirebuttal{Similarly}, denote the \textit{state vector} at time $t$ as $z^l_u(t) = [s_1^l(t), v_1^l(t), ..., s_{n-1}^l(t), v_{n-1}^l(t)]$, and the corresponding \textit{error state vector} as $x^l_u(t) = [\tilde{s}^l_{1}(t), \tilde{v}^l_{1}(t), ..., \tilde{s}^l_{n-1}(t), \tilde{v}^l_{n-1}(t)]^\intercal = z^l_u(t) - z^{*}_{n-1}$, where $z^{*}_{n-1} = [s^{*}_{n-1}, v^{*}_{n-1}, ... s^{*}_{n-1}, v^{*}_{n-1}] \in \mathbb{R}^{1\times 2(n-1)}$ is the equilibrium state. The error dynamics of the linearized uncontrolled system is
\vspace{\squeeze}
\begin{equation}
    \dot{x}^l_u(t) = A_u x^l_u(t) 
    \label{eq:uncontrol_matrix_system}
\vspace{\squeeze}
\end{equation}
where $A_u \in \mathbb{R}^{2(n-1)\times 2(n-1)}$ is similar to $A_c$ but with no AV control, and can be obtained from $A_c$ by removing the $[0; ...; 0; C_2; C_1]$ rows and adjusting the first row so that the $D_2$ matrix represents the corresponding leading vehicle (last HV) for the first HV. \slired{A previous work~\cite{cui2017stabilizing} uses string stability to show that the linearized, uncontrolled single-lane OVM system is unstable if $\alpha + 2\beta < 2\dot{V}(s^*_{n_l})$, or equivalently, $\alpha + 2\beta - v_{max}\frac{\pi}{s_{go}-s_{st}}\sin(\pi\frac{C/n_l - s_{st}}{s_{go} - s_{st}}) < 0$ for $s^*_{n_l} = C / n_l \in [s_{st}, s_{go}]$. By default, we assume OVM system parameters are chosen such that the uncontrolled single-lane system is asymptotically unstable. Notably, the stability of the original, nonlinear single-lane OVM system may improve at the beginning of a trajectory before eventually becoming unstable (see Sec.~\ref{sec:analysis_uncontrol_var} for details).}

\begin{figure}[!t]
    \centering
    \includegraphics[width=0.49\textwidth]{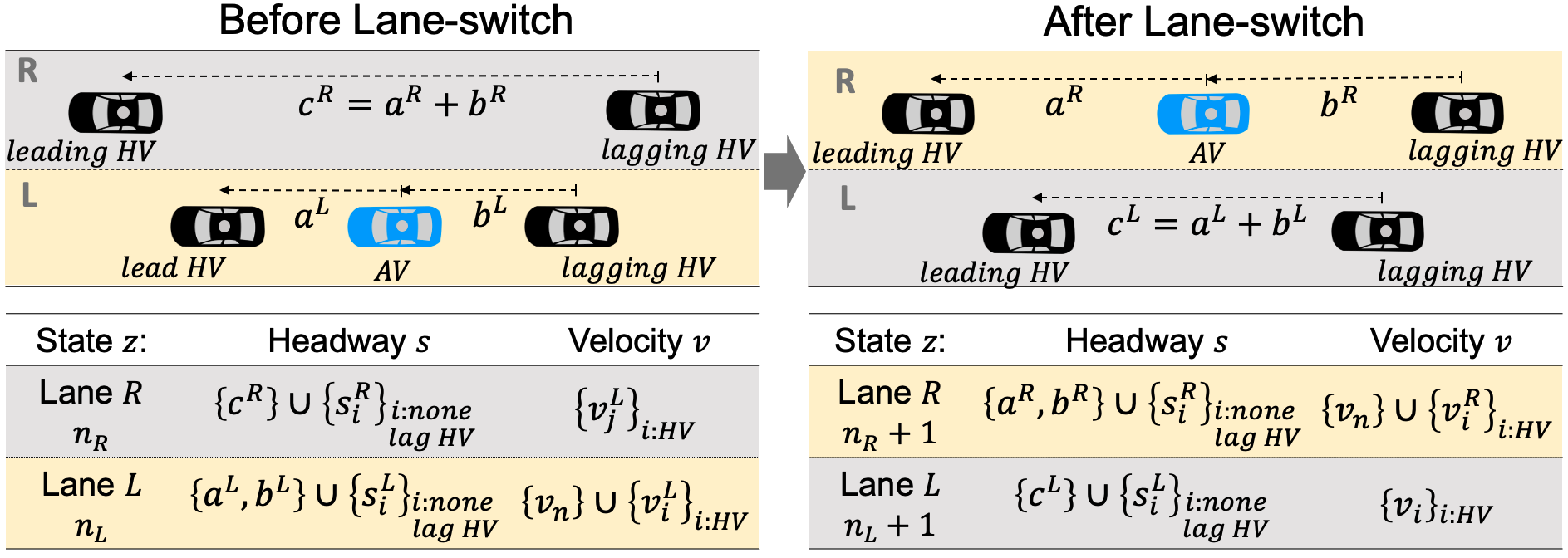}
    \caption{\textbf{Notation for AV Lane switching.} \sli{The table presents the detailed state notation ($z^l$) for the headways and velocities of all vehicles before and after the lane-switch. The AV (in blue) switches from lane $L$ to the same position on lane $R$, maintaining the same velocity $v_n$ and changes the total number of vehicles on each lane from \slirebuttal{$n_L = n, n_R = n-1$ to $n_L = n-1, n_R = n$}. The instantaneous impact of an AV lane-switch takes place in the local neighborhood involving the AV and its two neighboring HVs, where the headways change from $a^L, b^L$ and $c^R = a^R + b^R$ to $c^L = a^L + b^L$ and $a^R, b^R$. The headways $\{s^L_i\}_{\substack{i: none\\ lag\;HV}}, \{s^R_i\}_{\substack{i: none\\ lag\;HV}}$ for all none-lagging HVs, and the velocities $\{v^L_i\}_{i: HV}, \{v^R_I\}_{i: HV}$ for all HVs  stay the same.\vspace{-0.2cm}}}
    \label{fig:laneswitch_notation}
\end{figure}

\begin{figure*}[!t]
    \centering
    \includegraphics[width=0.95\textwidth]{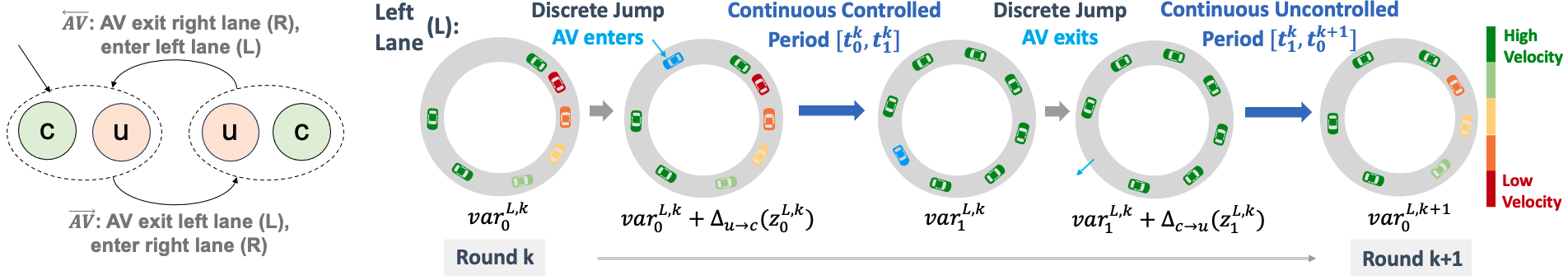}
    \caption{\textbf{System-level stability analysis for hybrid multi-lane systems.} \textbf{Left:} We analyze the multi-lane system under the hybrid system framework. \sli{The two modes are depicted by dashed oval, where the green and red circles inside each oval represent controlled and uncontrolled lanes. Transitions between modes correspond to the AV switching lanes. The AV starts by controlling lane $L$ in the two-lane hybrid system, as denoted by the solid pointed line.} \textbf{Right:} The proposed stability analysis for the traffic system assesses the variance change of each lane $l$ during each round $k$, where a round consists of 1) the AV enters a lane (discrete jump), 2) the AV continuously control the lane for a period, 3) the AV exits the lane (discrete jump), and 4) the lane stays uncontrolled for a period. The system becomes more stable if the initial variance $var^{l, k}_0$ is less than the final variance $var^{l, k}_2$ of the round. We illustrate the behavior of lane $L$, which becomes more stable, during the round. The color of each HV represents its velocity, and the AV is depicted in blue.~\slirebuttal{The color bar ranges from low velocity (congestion) to high velocity ($\geq$ equilibrium) to highlight the congestion mitigation and formation due to the presence and absence of AV.}}
    \label{fig:round_notation}
\end{figure*}

\vspace{-0.2cm}
\subsection{AV lane-switch maneuver}
\label{sec:prelim_multilane}
Without loss of generality, suppose the AV is on lane $L$ at time $t$, whose position $p^L_{n}(t) = p_n$ and velocity $v^L_{n}(t) = v_n$. We allow the AV to switches lane at any time. Following previous works~\cite{piu2022stability}, we model the lane switch as an instantaneous discrete jump, where the AV switches to the same position $p_n$ on lane $R$ and maintain the same velocity $v_n$. \slirebuttal{For the tractability of theoretical analysis, this work represents each vehicle by a single position $p_i^l(t) \in \mathbb{R}$ with zero length~\cite{zheng2020smoothing, li2023stabilization, liu2023reachability} (gray dot at each vehicle's center in Fig.~\ref{fig:laneswitch_notation}). We leave as a future work to incorporate positive vehicle length in the modeling of the system dynamics. Notably, although we do not explicitly prohibit the AV to switch lanes when it is close to the HVs on the enter lane, our theoretical analysis reveals the inadequacy of such a lane-switch strategy as it leads to system instability. In Sec.~\ref{sec:controller}, we further design Traffic-Aware controllers to ensure adequate vehicle headways when AV switches lanes, which can significantly improves the system’s stability.}

As shown \slirebuttal{in} Fig.~\ref{fig:laneswitch_notation}, we present a few notation for the headway and velocity components of the state $z^l(t)$ at the discrete jump to facilitate the theoretical analysis. One instantaneous effect of AV lane-switch is changing the headways between the two adjacent HVs to the AV (on both exit and enter lanes). We denote the headway between the AV and the leading and lagging HVs as $a^L, b^L$ (exit lane, which is combined to a single headway of $c^L = a^L + b^L$ after the switch) and $a^R, b^R$ (enter lane, split from a single headway of $a^R + b^R$ before the switch). The AV lane-switch further affects the lagging HVs' acceleration (on both lanes) by changing the leading vehicles' velocities due to a change in $\dot{\tilde{s}}_i = \tilde{v}_{i-1} - \tilde{v}_i$. We further denote the headway of all none-lagging HVs as $\{s^L_i\}_{\substack{i: none\\lag\;HV}}$ (exit lane) and $\{s^R_i\}_{\substack{i: none\\lag\;HV}}$ (enter lane), and denote the velocity of all HVs as $\{v^L_i\}_{i: HV}$ and $\{v^R_i\}_{i: HV}$ on each lane. These quantities remain unchanged immediately after the AV switches lanes.

\section{STABILITY ANALYSIS UNDER HYBRID SYSTEM FRAMEWORK}
\label{sec:analysis}
This section presents a variance-based stability analysis by casting the two-lane mixed-autonomy system in Sec.~\ref{sec:prelim} into the hybrid system framework. Intuitively speaking, to regulate the traffic and achieve stability on both lanes through the control of a single AV, the AV should balance the \slirebuttal{duration of having} each lane \slirebuttal{controlled and uncontrolled} (which generally increases and decreases the lane's stability, respectively), by properly switching between them, while considering the effect of the AV lane-switches, which may introduce additionally shockwaves. \sli{We formalize the intuition by presenting a Lyapunov analysis framework that uses variance as a stability metric, to account for the tug-of-war between various factors}. 

The analysis is later applied to analyze specific controllers (Sec.~\ref{sec:controller_fixed_duration} and~\ref{sec:controller}). Rather than a uniform outcome, the analysis reveals that different configurations of the traffic state reveal different results for which the tug-of-war is won or lost. This not only explains the iconic emergent phenomenon of the traffic break, but also uncovers behaviors resulting from slowed down versions of the traffic break, characterized by less frequent lane switches. The analysis can further improve the controller design, by mitigating the impact of AV lane switches on the system's stability.

\vspace{-0.3cm}
\subsection{Hybrid System Modeling of the Multi-lane System}
As depicted in Fig.~\ref{fig:round_notation} (left), the two-lane traffic system can be naturally modeled as a hybrid system with two modes, representing whether the AV controls lane $L$ or $R$ \slirebuttal{(the dashed oval). The transition between the two modes represents the AV lane-switch. We formally model the system as a hybrid automaton~\cite{goebel2009hybrid} with the following components:
\begin{itemize}
    \item \textbf{Modes:} $Q = \{m_L, m_R\}$ where the AV controls lane L in mode $m_L$ and lane R in mode $m_R$.
    \item \textbf{Domain Mapping:} $Domain(m_L) = \mathbb{R}^{2n}\times \mathbb{R}^{2n-2}$, $Domain(m_R) = \mathbb{R}^{2n-2}\times\mathbb{R}^{2n}$, representing the continuous states $(x_c^L(t), x_u^R(t))$ and $(x_u^L(t), x_c^R(t))$.
    \item \textbf{Flow Map} $f: Q \times \mathbb{R}^{4n-2} \rightarrow \mathbb{R}^{4n-2}$: For each mode, the state for the controlled lane $x_c^l(t)$ follows the continuous dynamics in Eq.~\eqref{eq:control_matrix_system} and the uncontrolled lane $x_u^l(t)$ follows the continuous dynamics in Eq.~\eqref{eq:uncontrol_matrix_system}.
    \item \textbf{Edges} $E = \{(m_L, m_R), (m_R, m_L)\}$: The possible transitions between modes corresponding to the AV lane switch.
    \item \textbf{Reset Map} $R: E\times \mathbb{R}^{4n-2} \rightarrow \mathbb{R}^{4n-2}$: Following the lane-switch maneuver as described in Sec.~\ref{sec:prelim_multilane}, the reset map $R = R_{permute} \circ R_{insert-remove}$ is a composite function, where $R_{insert-remove}$ inserts the AV to the corresponding location in the states vector of the enter lane and removes it from the states vector of the exiting lane, $R_{permute}$ permutes the vehicle indices in the enter lane so that the AV appears in the last index of the state vector. Eq.~\eqref{eq:reset_map} illustrates the reset map when the AV switches from lane $L$ to $R$.
    \item \textbf{Guard Condition}: This paper controls the AV lane-switch behavior, which equivalently controls the guard condition that describes the criteria for transition between modes. The Fixed Duration Control in Sec.~\ref{sec:controller_fixed_duration} enforces mode transition every $Ts$, while the Traffic-Aware Control in Sec.~\ref{sec:controller} enables mode transition when the state of the enter lane satisfies certain regularity conditions.
\end{itemize}}

Without loss of generality, we suppose the AV starts by controlling lane $L$ (the solid pointed line). Under every two transitions, the system returns to the same mode, although possibly with different system states (headways and velocities). We denote this two-transition period as a \textit{round} $k$, which spans a time period $[t_0^k, t_0^{k+1}] = [t_0^{k}, t_1^{k}] \cup [t_1^{k}, t_0^{k+1}]$. A trajectory consists of multiple such rounds $k \in \mathbb{N}$. Fig.~\ref{fig:round_notation} (right) illustrates the behavior of lane $L$ during the round $k$: the AV enters lane $L$ (discrete jump), controls it for $[t_0^{k}, t_1^{k}]$ (continuous controlled period), exits lane $L$ at time $t_1^k$ (discrete jump), and lane $L$ stays uncontrolled for $[t_1^k, t_0^{k+1}]$ (continuous uncontrolled period). 

\begin{table*}[]
\centering
\begin{equation}
    \begin{aligned}
    (x_c^L(t), x_u^R(t)) = \quad\quad\;\; ([s_1^L(t), v_1^L(t), ..., \underline{s_n^L(t), v_n^R(t)}]& , [s_1^R(t), v_1^R(t), ..., s_{n-1}^R(t), v_{n-1}^R(t)]) \\
    & \downmapsto \; \boxed{R_{insert-remove}, \text{ where the AV position satisfies } mod(p_{k-1}^R(t), C) \geq mod(p_n^L(t), C) \geq mod(p_k^{R}(t), C)}\\
   ([s_1^L(t), ..., s_{n-1}^R(t), v_{n-1}^R(t)] & , [s_1^R(t), ..., v_{k-1}^R(t), \underline{s_n^L(t), v_n^R(t)}, s_{k}^R(t), ..., v_{n-1}^R(t)]) \\
    & \downmapsto \; \boxed{R_{permute}, \text{ based on } \sigma: [1, ..., k-1, k, k+1, ..., n] \mapsto [n-k+1, ..., n-1, n, 1, ..., n-k]} \quad   \\
    (x_u^L(t), x_c^R(t)) = \:\:\: ([s_1^L(t), v_1^L(t), ..., s_{n-1}^R(t), v_{n-1}^R(t)] & , [s_{k}^R(t), ..., v_{n-1}^R(t), s_1^R(t), ..., v_{k-1}^R(t), \underline{s_n^L(t), v_n^R(t)}]) 
    \end{aligned}
    \label{eq:reset_map}
\end{equation}
\vspace{-0.15cm}
\end{table*}

\slirebuttal{
\noindent \textbf{Stability analysis overview.} In the following, we present a Lyapunov stability framework for the two-lane hybrid system, using a variance-based stability metric in Eq.~\eqref{eq:variance_twolane} as the Lyapunov function. Summarized in Sec.~\ref{sec:analysis_summary} and derived in Sec.~\ref{sec:variance_derivation}, the analysis examines the evolution of the stability metric within each round, which involves a dynamic interplay among multiple competing factors: the continuous controlled periods (improves stability), the continuous uncontrolled periods, and the discrete jumps (both could impair stability). }

\begin{table*}[]
\centering
\caption{Notation for states and variances at key time points during a round $k$ \slirebuttal{spanning} a time period $[t_0^k, t_1^k] \cup [t_1^k, t_0^{k+1}]$. We consider the states and variances \slirebuttal{of the specific time points} before (top) and after (bottom) the AV switches lanes. \slirebuttal{For example, we denote the state at the time point $t_0^k$ right before the AV switches lanes (enter lane $L$) as $z^L_u(t_0^k) = z_0^{L, k}$, and the corresponding variance as $var^L(t_0^k) = var_0^{L, k}$ obtained by evaluating the variance function $var^l(t)$ in Eq.~\eqref{eq:variance} at the specific time $t_0^k$.}}
\scalebox{0.95}{
\begin{tabular}{lcccc}
\hline \\[-0.85em]
\multirow{2}{*}{$[t_0^k, t_1^k] \cup [t_1^k, t_0^{k+1}]$: Round $k$}  & \multicolumn{2}{c}{Lane $L$}   & \multicolumn{2}{c}{Lane $R$}  \\[-0.85em] \\ 
\cline{2-5}\\[-0.85em] & State   & Variance    & State        & Variance    \\[-0.85em]  \\ \hline \\[-0.85em]
\begin{tabular}[c]{@{}l@{}}$t_0^k:$ Initialization  / \\ \hspace{0.4cm} End of prev. cont. period\\ \hspace{0.4cm} (Before lane-switch)\end{tabular}       & $z_0^{L, k}$         & $var_0^{L, k}$   & $z_0^{R, k}$     & $var_0^{R, k}$       \\ \\[-0.85em]
\begin{tabular}[c]{@{}l@{}}$t_0^k:$ AV enters lane $L$\\\hspace{0.4cm} (After lane-switch)\end{tabular}    & $\tilde{z}_0^{L, k}$ & $\tilde{var}_0^{L, k} = var_0^{L, k} + \Delta_{u\rightarrow c}(z_0^{L, k})$ & $\tilde{z}_0^{R, k}$               & $\tilde{var}_0^{R, k} = var_0^{R, k} + \Delta_{c\rightarrow u}(z_0^{R, k})$  \\ \\[-0.85em]
\multicolumn{5}{c}{\hrulefill \;\; $[t_0^k, t_1^k]$: AV controls lane $L$ \;\;\hrulefill } \\ \\[-0.85em]
\begin{tabular}[c]{@{}l@{}}$t_1^k:$ End of cont. period \\\hspace{0.4cm} (Before lane-switch)\end{tabular} & $z_1^{L, k}$         & $var_1^{L, k}$        & $z_1^{R, k}$     & $var_1^{R, k}$       \\ \\[-0.85em]
\begin{tabular}[c]{@{}l@{}}$t_1^k:$ AV enters lane $R$\\\hspace{0.4cm} (After lane-switch)\end{tabular}    & $\tilde{z}_1^{L, k}$ & $\tilde{var}_1^{L, k} = var_1^{L, k} + \Delta_{c\rightarrow u}(z_1^{L, k})$ & $\tilde{z}_1^{R, k}$               & $\tilde{var}_1^{R, k} = var_1^{R, k} + \Delta_{u\rightarrow c}(z_1^{R, k})$ \\ \\[-0.85em]
\multicolumn{5}{c}{\hrulefill \;\; $[t_1^k, t_0^{k+1}]$: AV controls lane $R$ \;\;\hrulefill}    \\  \\[-0.85em]
$t_0^{k+1}$ End of cont. period    & $z_0^{L, k+1}$       & $var_0^{L, k+1}$     & $z_0^{R, k+1}$ & $var_0^{R, k+1}$    \\[-0.85em]   \\ \hline
\end{tabular}\vspace{-0.2cm}
}
\label{tab:notation_var}
\end{table*}
\vspace{-0.3cm}
\subsection{Variance-based Stability Metric}
\label{sec:analysis_variance}
In a stable system, all vehicles on each lane maintain uniform headway and velocity, resulting in low variance; conversely, stop-and-go waves in an unstable system disrupt the uniformity, causing fluctuations in the headways and velocities as the vehicles' decelerate or accelerate, resulting in high variance. We formalize this intuition by designing a varianced-based stability metric as follows. \slirebuttal{Specifically, we consider the state of each vehicle as individual samples, and use the sample variance of the headways and velocities of the vehicles on each lane as the stability metric for the respective lane.}  

\begin{definition}[Stability Metric]
\label{def:stability_metric}
For each lane $l \in \{L, R\}$, we define the variance-based stability metric as 
\begin{equation}
    \begin{aligned}
        var^l(t) &= variance(\{s^l_i(t)\}_{i}) + variance(\{v^l_i(t)\}_{i})
    \end{aligned}
    \label{eq:variance}
\end{equation}
where the variances are taken over all vehicles on the lane; the \slirebuttal{sample} variance operator applying on a set $\{\delta_i\}_{i=1}^{m}$ of size $m$ is $variance(\{\delta_i\}_{i=1}^{m}) = \frac{1}{m}\sum\limits_{i=1}^{m} \delta_i^2 - \slirebuttal{\Big(}\frac{1}{m}\sum\limits_{i=1}^{m} \delta_i\slirebuttal{\Big)}^2$. \\
\slirebuttal{The stability metric for the two-lane system is then defined as 
\begin{equation}
V_{m}(t) = var^L(t) + var^R(t), \text{ for each mode } m\in \{m^L, m^R\}, 
\label{eq:variance_twolane}
\end{equation}
which serve as the Lyapunov function for the Hybrid system.
}
\end{definition}
\vspace{-0.6cm}
\slirebuttal{\begin{remark}
The variance-based stability metric can be seen as ellipsoidal sets centered on the average state that dynamically changes within each mode and between mode transitions. We use this metric because (i) it has ideal convergence property, as the system reaches equilibrium when the variances are zero; (ii) it is challenging to establish meaningful stability upper bounds for AV lane-switches directly with the error state norm $\|x^l(t)\|_2^2 = \|z^l(t) - z^*_{n_l}\|_2^2$, which has been used as Lyapunov functions for the single-lane system~\cite{li2023stabilization}: as this metric is centered on the equilibrium state $z^*_{n_l}$ that is constant within each mode but changes at each mode transition, it complicates the closed-form derivation of the changes in $\|x^l(t)\|_2^2$ when the AV switches lanes; (iii) in contrast, the variance-based stability metric automatically adapts to equilibrium changes by centering on the average state, resulting in interpretable closed-form expressions for the variances change when AV switches lanes (Sec.~\ref{sec:analysis_jump}). This allows us to  explain emergent traffic phenomena (Sec.~\ref{sec:controller_fixed_duration}) and design controllers (Sec.~\ref{sec:controller}).
\end{remark}}

Next, we examine the stability of each lane using the variance-based metric; in essence, the overall two-lane system is stable when both lanes reach low-variance states.
\vspace{-0.2cm}
\subsection{Multi-lane Stability Analysis}
\label{sec:analysis_summary}
Table~\ref{tab:notation_var} \slirebuttal{presents} the notation for states and variances\slirebuttal{, where we introduce additional superscripts and subscripts to refer to the states and variances evaluated at specific key} time points during a round $k$.
The variance notation for lane $L$ aligns with those displayed in Fig.~\ref{fig:round_notation}, \sli{where we denote the variances (i) at the beginning of the round $k$ as $var_0^{L, k}$, (ii) after the AV enters lane $L$ as $var_0^{L, k} + \Delta_{u\rightarrow c}(z_0^{L,k})$, \slirebuttal{where $z_0^{L,k}$ is the state of lane $L$ right before the AV enters, and $\Delta_{u\rightarrow c}(z_0^{L,k})$ denotes the change in lane $L$'s variance after the AV enters}, (iii) after the AV controls lane $L$ for a period of $[t_0^k, t_1^k]$ as $var_1^{L,k}$, (iv) after the AV exits the lane $L$ as $var_1^{L,k} + \Delta_{c\rightarrow u}(z_1^{L,k})$, \slirebuttal{where $z_1^{L, k}$ and $\Delta_{c\rightarrow u}(z_1^{L,k})$ are similarly defined when the AV exits}, and (v) after the lane stays uncontrolled for a period of $[t_1^k, t_0^{k+1}]$ as $var_0^{L, k+1}$.} The variance metric allows us to measure the stability of each lane $l \in \{L, R\}$, \sli{by unwrapping the hybrid system dynamics (i) - (v)}. This leads to a tug-of-war among: 
\begin{enumerate}
\item the continuous controlled period\slirebuttal{s}, which generally reduce the variance. For instance, generally, $var_0^{L, k} + \Delta_{u\rightarrow c}(z_0^{L, k}) < var_1^{L, k}$, whose magnitude depends on the duration $t_1^k - t_0^k$;
\item the continuous uncontrolled period\slirebuttal{s}, which generally increase the variance. For instance, generally, $var_1^{L, k} + \Delta_{c\rightarrow u}(z_1^{L, k}) > var_0^{L, k+1}$, whose magnitude depends on the duration $t_0^{k+1} - t_1^k$;
\item the discrete jumps, which generally increase the variance $\Delta_{u\rightarrow c}, \Delta_{c\rightarrow u} > 0$) through altering the headways (and subsequently, velocities) of the leading and \slirebuttal{lagging} HVs when \slirebuttal{the} AV switches lanes.
\end{enumerate}
If the final variance is smaller than the initial variance ($var_0^{l, k+1} < var_0^{l, k}$), the variance of the lane decreases within the round $k$, after which the lane starts a new round $k+1$ with a lower initial variance $var_0^{l, k+1}$. \slirebuttal{Moreover,} the variances of both lanes are interconnected\slirebuttal{:} if the AV primarily focuses on controlling one lane, it would lead to a high variance on the other lane \slirebuttal{due to an extended uncontrolled period}. Decreasing variances for both lanes in successive rounds promote\slirebuttal{s} the system to become stable, whereas continual increasing variance for some lane leads to instability of the lane and the overall two-lane system. 

\vspace*{0.1cm}
\noindent \textbf{General Theorem.} The following Theorem~\ref{thm:general_round} formally characterizes the variance changes in a two-lane mixed-autonomy system during each round $k$, \sli{through upper bounding the variance changes after each continuous period and discrete jump}. The theorem can be applied to different continuous dynamics and discrete jumps, and further to a broader class of hybrid systems with the same two-mode structure as in Fig.~\ref{fig:round_notation} (left). We later apply the general theorem to analyze the two-lane traffic system with the specific discrete jump specification. 
\begin{theorem}
\label{thm:general_round} For each lane $l \in \{L, R\}$, consider a round $k$ within a time period $[t_0^k, t_0^{k+1}] = [t^k_0, t^k_1] \cup [t_1^k, t_0^{k+1}]$ as described in Fig.~\ref{fig:round_notation} and Table~\ref{tab:notation_var}, where each lane has an initial variance $var_0^{l, k}$ at $t^k_0$. Suppose
\begin{enumerate}
\item \slirebuttal{$f_c$ and $f_u$ are differentiable functions such that} the variances at the continuous controlled and uncontrolled periods are upper bounded by $var^l(t) \leq f_c(var_0^l, t)$ and $var^l(t) \leq f_u(var_0^l, t)$ where $t \geq 0$ represents the controlled or uncontrolled duration and $var_0^l$ represents the initial variance at a duration $t = 0$. 
\item the discrete jumps (from AV entering and exiting the lane) change the variances by ${\Delta}_{u\rightarrow c}(z^l)$ and ${\Delta}_{c\rightarrow u}(z^l)$,  where $z^l$ denote the state of the lane before the jump. 
\end{enumerate}
Then, the variance of lane $L$ after the round $k$ is upper bounded by $var_0^{L,k+1}$, where we have
\begin{equation}
    \begin{aligned}
        var_1^{L, k} & \leq f_{c}(var_0^{L, k} + \Delta_{u\rightarrow c}(z_0^{L, k}), t^{L,k}_1-t^{L,k}_0), \\
        var_0^{L, k+1} & \leq f_{u}(var_1^{L,k} + \Delta_{c\rightarrow u}\big(z_1^{L,k}), t_0^{L, k+1} - t_1^{L, k}).
    \end{aligned}
    \label{eq:general_round_1}
\end{equation}
The variance of lane $R$ after the round $k$ is upper bounded by $var_0^{R,k+1}$, where we have
\begin{equation}
    \begin{aligned}
        var_1^{R, k} & \leq f_{u}(var_0^{R, k} + \Delta_{c\rightarrow u}\big(z_0^{R,k}), t^{k}_1-t^{k}_0),\\
        var_0^{R, k+1} & \leq f_{c}(var_1^{R, k} + \Delta_{u\rightarrow c}(z_1^{R, k}), t_0^{R, k+1} - t_1^{R, k}).
    \end{aligned}
    \label{eq:general_round_2}
\end{equation}
where $z_0^{L, k}, z_1^{L, k}, z_0^{R, k}, z_1^{R, k}$ are the states of lane $L$ and $R$ at $t_0^k$ and $t_1^k$ before AV switches lanes.

\slirebuttal{From Eq.~\eqref{eq:variance_twolane}, the stability metric for the two-lane system is given by a sum of the variances on both lanes (e.g. $var_0^{L, k} + var_0^{R, k}, var_1^{L, k} + var_1^{R, k}, var_0^{L, k+1} + var_0^{R, k+1}$);} the system becomes stable if \slirebuttal{the upper bounds on all the corresponding variance terms in Eq.~\eqref{eq:general_round_1} and~\eqref{eq:general_round_2}} remain small after some round $\bar{k} \geq 0$, but unstable \slirebuttal{otherwise}.
\end{theorem}
\begin{proof}
Eq.~\eqref{eq:general_round_1} and~\eqref{eq:general_round_2} follows by unwrapping the dynamics with the corresponding upper bounds, applying the variances upper bounds $f_c, f_u$ at the continuous controlled and uncontrolled periods and the discrete jumps $\Delta_{u\rightarrow c}, \Delta_{c\rightarrow u}$. 
\end{proof}
\vspace*{-0.1cm}
\noindent \textbf{Two-Lane Traffic Specification.} 
\sli{The variance upper bounds ($f_c, f_u, \Delta_{c\rightarrow u}, \Delta_{u\rightarrow c}$) in the general theorem can be specified to the specific two-lane mixed-autonomy system in Sec.~\ref{sec:prelim}.} The following corollary presents the resulting variance upper bounds, and we derive the corollary in the next Sec.~\ref{sec:variance_derivation}. 

\sli{At a high level, the variance upper bounds for continuous periods is derived from both a Lyapunov analysis, which captures the asymptotic behaviors of the controlled and uncontrolled periods (variance eventually decreases and increases, respectively), and a state-dependent analysis, which refines the short-term behavior of the uncontrolled period (initial variance decrease followed by eventual increase)\slired{, as illustrated in Fig.~\ref{fig:statedep}}. Meanwhile, the magnitude of the variance changes at discrete jumps is affected by the interaction between the AV and the leading/lagging HVs.} The variance upper bounds for lane $L$ and $R$ can be similarly derived due to the symmetric structure. 
\begin{corollary}
\label{cor:traffic_round}
Consider a round $k$ within a time period $[t_0^k, t_0^{k+1}] = [t_0^k, t_1^k] \cup [t_1^k, t_0^{k+1}]$ as described in Fig.~\ref{fig:round_notation} and Table~\ref{tab:notation_var}, where lane $L$ has an initial variance $var^{L, k}_0$ \slirebuttal{at $t^k_0$ before the AV enters lane $L$}. The variance of lane $L$ after the round $k$ is upper bounded by $var_0^{L,k+1}$, with
\begin{equation}
    \begin{aligned}
    var_1^{L, k} & \leq \alpha_1(var_0^{L, k} + \Delta_{u\rightarrow c}(z_0^{L, k}))\exp(-\alpha_2(t^{1,k}_1-t^{1,k}_0)), \\
    var_0^{L, k+1} & \leq \beta_1 (var_1^{L, k} + \Delta_{c \rightarrow u}(z_0^{L, k}))\exp(\beta_2(t_0^{L, k+1} - t_1^{L, k}))
    \end{aligned}
    \label{eq:traffic_1}
\end{equation}
where  $\alpha_1, \alpha_2, \beta_1, \beta_2 > 0$, and
\begin{equation}
\begin{aligned}
    \Delta_{u\rightarrow c}(z_0^{L, k}) = &  - \frac{1}{n}var_0^{L,k} - \frac{1}{n}2a_0^{L,k} b_0^{L,k} + \frac{1}{n^2(n-1)}C^2  \\
    & + \frac{1}{n^2(n-1)}(-(n-1)v_{n, 0}^{L,k} + \sum\limits_{i=1}^{n-1} v_{i, 0}^{L,k})^2,
    \label{eq:traffic_31}
\end{aligned}
\end{equation}
\begin{equation}
    \begin{aligned}
    \Delta_{c\rightarrow u}(z_1^{L, k}) = &  \frac{1}{n-1}var_1^{L,k}  + \frac{1}{n-1}2a_1^{L,k}b_1^{L,k} - \frac{1}{(n-1)^2n}C^2 \\
    & - \frac{1}{(n-1)^2n}(-(n-1)v_{n,1}^{L,k} + \sum\limits_{i=1}^{n-1} v_{i, 1}^{L,k})^2,
    \end{aligned}
    \label{eq:traffic_32}
\end{equation}
\noindent where $a_0^{L, k}, b_0^{L, k}, a_1^{L, k}, b_1^{L, k}$ are the headways between AV and the leading and lagging HVs, and $v_{i, 0}^{L, k}, v_{i, 1}^{L, k}$ are the velocities of the HVs ($1 \leq i \leq n-1$) and the AV ($i=n$), both at $t_0^k$ and $t_1^k$ before AV switches lanes (see notation in Fig.~\ref{fig:laneswitch_notation}). 

Let $\epsilon > 0$, if the initial state $\tilde{z}_1^{L, k}$ for the uncontrolled period (the state immediately after the AV exits the lane, with variance $var_1^{L, k} + \Delta_{u \rightarrow c}(z_0^{L, k})$) is close to a fixed initial state $\bar{z}_{0,u}$, i.e. $\|\tilde{z}_1^{L, k} - \bar{z}_{0, u}\|_2 \leq \delta_\epsilon$ for some $\delta_\epsilon \geq 0$, we have the tightened variance upper bound
\slired{\begin{equation}
    \begin{aligned}
        var_0^{L, k+1} \leq 
        \begin{cases}  \overline{var}_u(t_0^{k+1} - t_1^{k}) + \epsilon, \hspace{0.4cm} \text{if } t_0^{k+1} - t_1^{k} \in [0, t_\epsilon] \\ 
        \beta_1 (\overline{var}_u(t_\epsilon) + \epsilon)\exp(\beta_2(t_0^{k+1} - t_1^{k} - t_\epsilon)) , \\
        \hspace{3.7cm} \text{if } t_0^{L, k+1} - t_1^{L, k} \in [t_\epsilon, \infty]\end{cases}
    \end{aligned}
    \label{eq:traffic_2}
\end{equation}}
\begin{figure}
    \centering
    \includegraphics[width=0.3\textwidth]{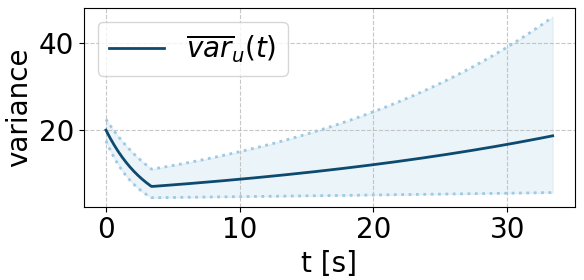}
    \caption{\sli{\textbf{Illustration of the State-Dependent Variance Upper Bound}. For a lane $l \in \{L, R\}$, consider an uncontrolled period with a trajectory \slirebuttal{$z_u^l(t)$}, whose initial state $z_{0}$ is close to the fixed initial state $\bar{z}_{0, u}$ (within the shaded blue region at $t = 0$), where $\bar{z}_{0, u}$ is the state immediately after the AV exits a stable lane at equilibrium, as depicted in Fig.~\ref{fig:control_uncontrol_singlelane} (f). Due to the continuous dependence of the system's solution on the initial state (Sec.~\ref{sec:analysis_uncontrol_var}), the variance of the trajectory $var^l(t)$ with initial state $z_{0}$ stays close to the variance $\overline{var}_u(t)$ of the nominal trajectory with initial state $\bar{z}_{0, u}$ (the solid dark blue curve); the width of the region, which upper bounds the distance between $var^l(t)$ and $\overline{var}_u(t)$, increases as $t$ increases. The state-dependent upper bound of $var^l(t)$ hence follows a similar rate of decrease as $\overline{var}_u(t)$ at some initial time period $[0, t_\epsilon]$, when the upper bound is relatively tight.\vspace*{-0.2cm}}}
    \label{fig:statedep}
\end{figure}
where $\overline{var}_u(t)$ is the variance of the trajectory $\bar{z}_{u}(t)$ with the fixed initial state $\bar{z}_{0,u}$ as depicted in Fig.~\ref{fig:control_uncontrol_singlelane} (f-j), which decreases in the initial period $[0, t_\epsilon]$. 

Moreover, let $\bar{z}_u(t)$ and $\bar{z}_c(t)$ denote trajectories under some fixed initial states $\bar{z}_{0, u}$ and $\bar{z}_{0, c}$. If the initial states for the controlled and uncontrolled periods are close to the corresponding fixed initial states $\|\tilde{z}_0^{L, k} - \bar{z}_{0, c}\|_2 \leq \delta_\epsilon$ and $\|\tilde{z}_1^{L, k} - \bar{z}_{0, u}\|_2 \leq \delta_\epsilon$, the states at the end of the controlled and uncontrolled periods $z_1^{L, k}$ and $z_0^{L, k+1}$ further satisfy 
\slired{\begin{equation}
\begin{aligned}
\|z_1^{L, k} - \bar{z}_c(t_1^{k} - t_0^{k})\|_2 \leq \epsilon, \;\;  \text{ if } t_1^{k} - t_0^{k} \in [0, t_\epsilon] \\
\|z_0^{L, k+1} - \bar{z}_u(t_0^{k+1} - t_1^{k})\|_2 \leq \epsilon \;\;\text{ if } t_0^{k+1} - t_1^{k} \in [0, t_\epsilon]  \\
\end{aligned}
\end{equation}}
\end{corollary}
\vspace{-0.4cm}
\begin{proof} The variance bounds in the Corollary are derived in Sec.~\ref{sec:variance_derivation}. The Lyapunov-based variance bound (Eq.~\eqref{eq:traffic_1}) is derived in Prop.~\ref{prop:cont_var_bd1}; the closed-form expressions for the discrete jumps $\Delta_{u\rightarrow c}$ and $\Delta_{c \rightarrow u}$ are derived in Prop.~\ref{prop:lanechange}. The state-dependent tightened variance bound for the uncontrolled period (Eq.~\eqref{eq:traffic_2}) is derived in Prop.~\ref{prop:uncontrol_var}, and the state-dependent bounds on the actual trajectory are derived in Cor.~\ref{cor:cont_init}.
\end{proof}
\vspace*{-0.3cm}
The Corollary offers a systematic explanation for emergent traffic phenomena of stabilizing a two-lane traffic system with a single AV, where we present a detailed analysis in Sec.~\ref{sec:controller_fixed_duration}. At a high level, suppose $t_1^{k}-t_0^{k} = t_0^{k+1}-t_1^{k} = T\;\forall k\in \mathbb{N}$. Traffic break occurs when the AV rapidly switches between lanes with $T \approx 0 \in [0, t_\epsilon]$, where the system essentially returns to the same state after every two  AV lane switches. The two-lane system can be stabilized to low-variance states due to minimal variance impact from discrete jumps $\Delta_{u\rightarrow c}(z_0^{L, k}) + \Delta_{c\rightarrow u}(z_1^{L,k}) \approx 0$, as well as variance reduction from both the continuous controlled and uncontrolled periods (see Eq.~\eqref{eq:traffic_1} and~\eqref{eq:traffic_2}).

When we reduce the lane-switch frequency by increasing $T$, the AV lane-switches lead to variance increases in the system, $\Delta_{u\rightarrow c}(z_0^{L, k}) + \Delta_{c\rightarrow u}(z_1^{L,k}) > 0$, due to the evolution of the system's state between lane switches. The stability of the two-lane system is compromised when $T$ is either (1) too small, such that the variance reduction from continuous period is insufficient to counterbalance the variance increase from the discrete jumps (see Eq.~\eqref{eq:traffic_1} and~\eqref{eq:traffic_2}), or (2) too large, which results in significant variance increases during the continuous uncontrolled periods when $T > t_\epsilon$ is large (see Eq.~\eqref{eq:traffic_2}). As a consequence, stabilizing the two-lane system with less frequent lane-switches require an appropriate $T$ so that the continuous period can effectively reduce variances to dissipate the impact from the discrete jumps. We refer the reader to Sec.~\ref{sec:controller_fixed_duration} for details.

\section{TRAFFIC SYSTEM VARIANCE SPECIFICATION}
\label{sec:variance_derivation}

\begin{figure*}
    \centering
    \begin{subfigure}{\textwidth}
        \label{fig:control_singlelane}
        \begin{subfigure}{0.12\textwidth}
        \includegraphics[width=\textwidth]{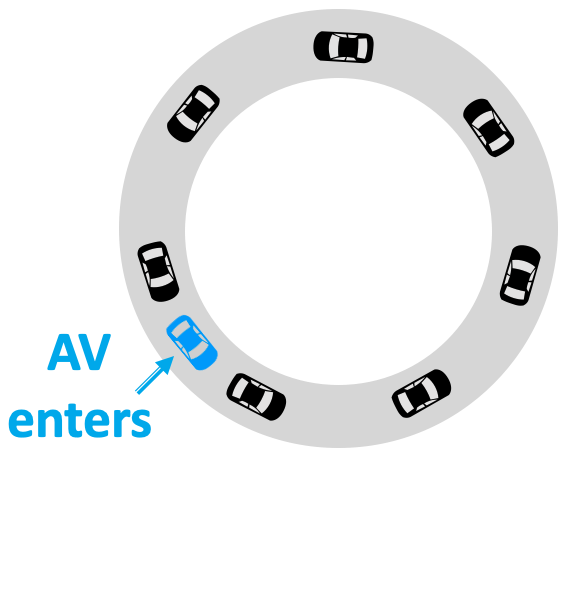}
        \caption{Controlled}
        \end{subfigure}
        \begin{subfigure}{0.21\textwidth}
        \includegraphics[width=\textwidth]{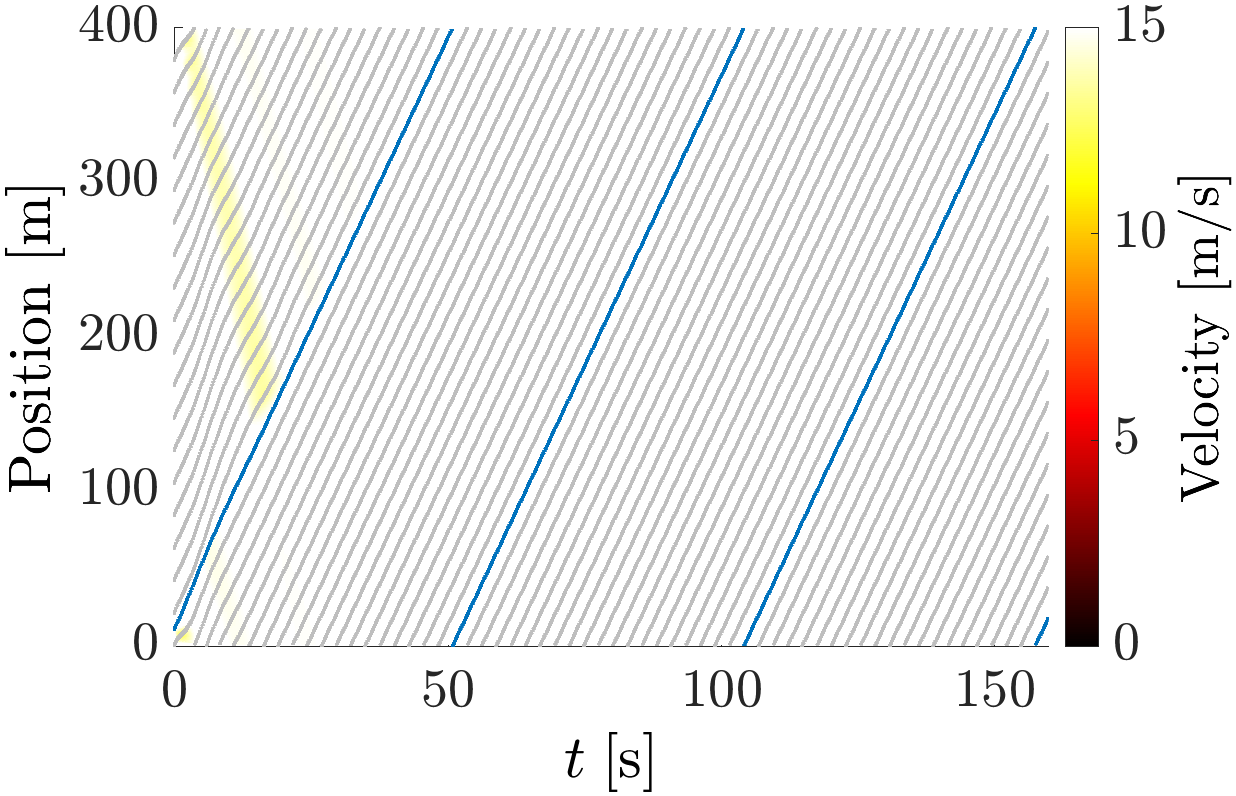}
        \caption{Time-Space Diagram}
        \end{subfigure}
        \begin{subfigure}{0.19\textwidth}
        \includegraphics[width=\textwidth]{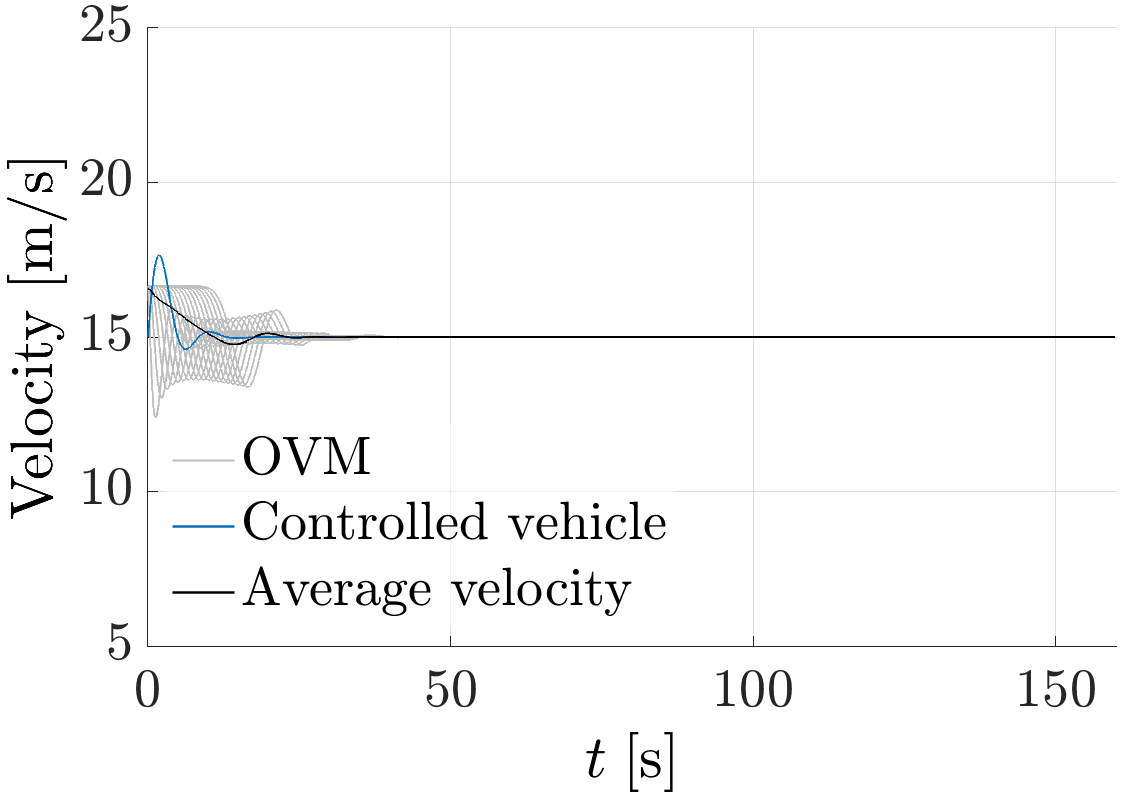}
        \caption{Time-Velocity Diagram}
        \end{subfigure}
        \begin{subfigure}{0.21\textwidth}
            \includegraphics[width=\textwidth]{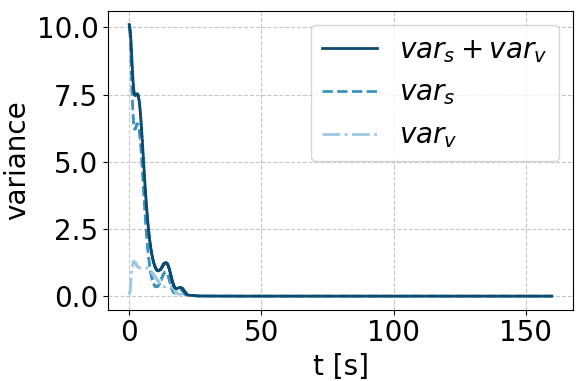}
        \caption{Variance}
        \end{subfigure}
        \begin{subfigure}{0.21\textwidth}
            \includegraphics[width=\textwidth]{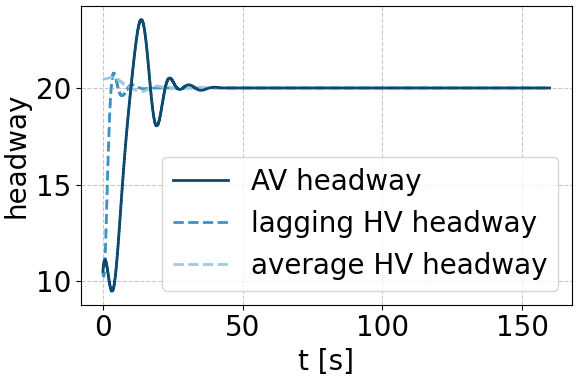}
        \caption{Headways}
        \end{subfigure}
    \end{subfigure}
    \begin{subfigure}{\textwidth}
        \label{fig:uncontrol_singlelane}
        \begin{subfigure}{0.12\textwidth}
        \includegraphics[width=\textwidth]{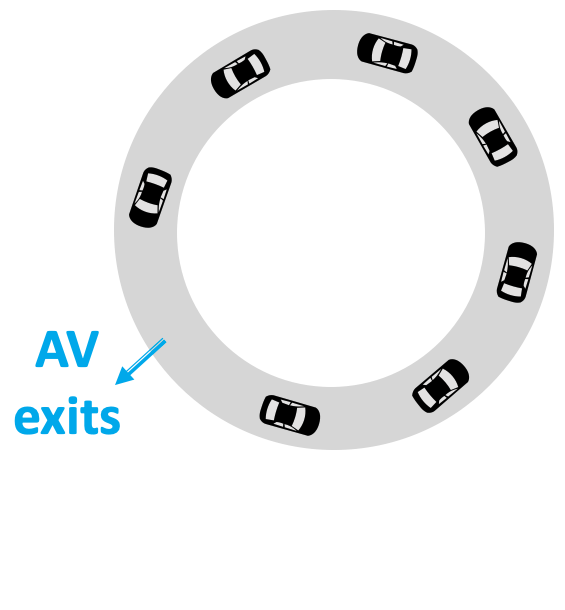}
        \caption{Uncontrolled}
        \end{subfigure}
        \begin{subfigure}{0.21\textwidth}
        \includegraphics[width=\textwidth]{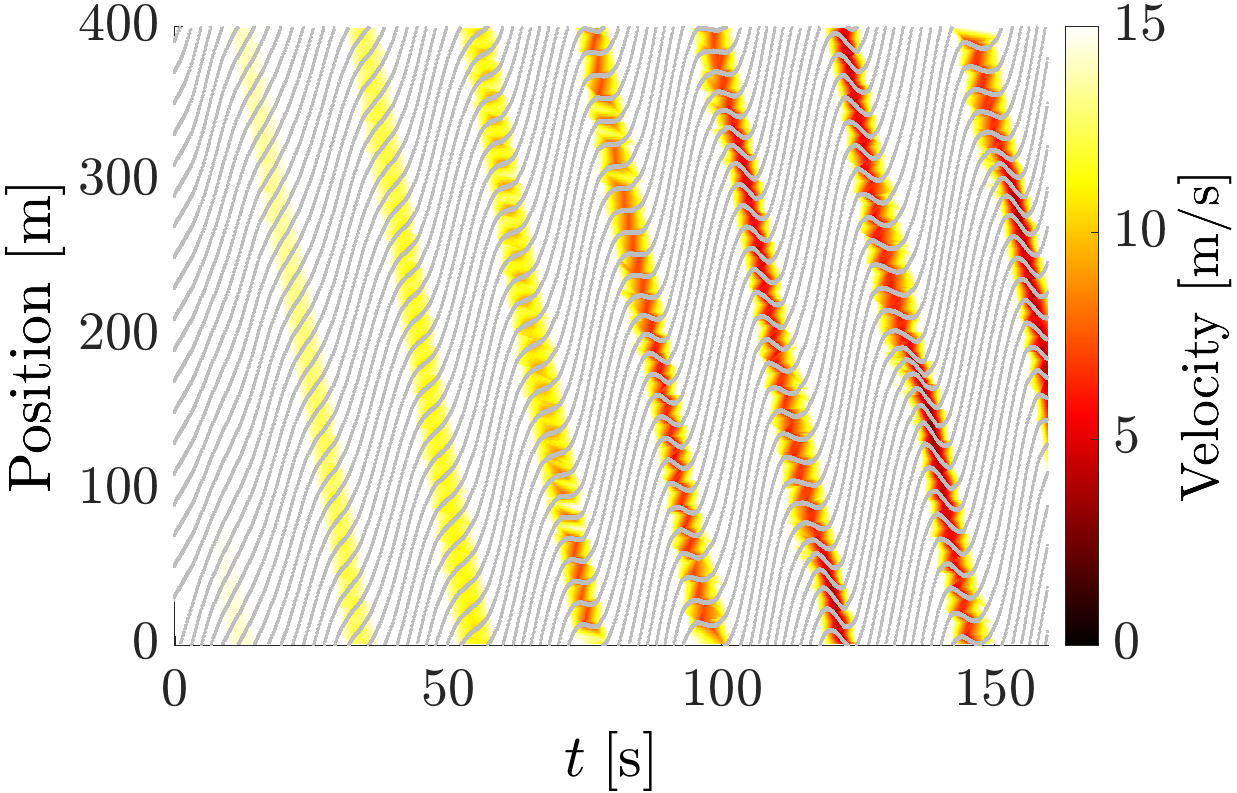}
        \caption{Time-Space Diagram}
        \end{subfigure}
        \begin{subfigure}{0.19\textwidth}
        \includegraphics[width=\textwidth]{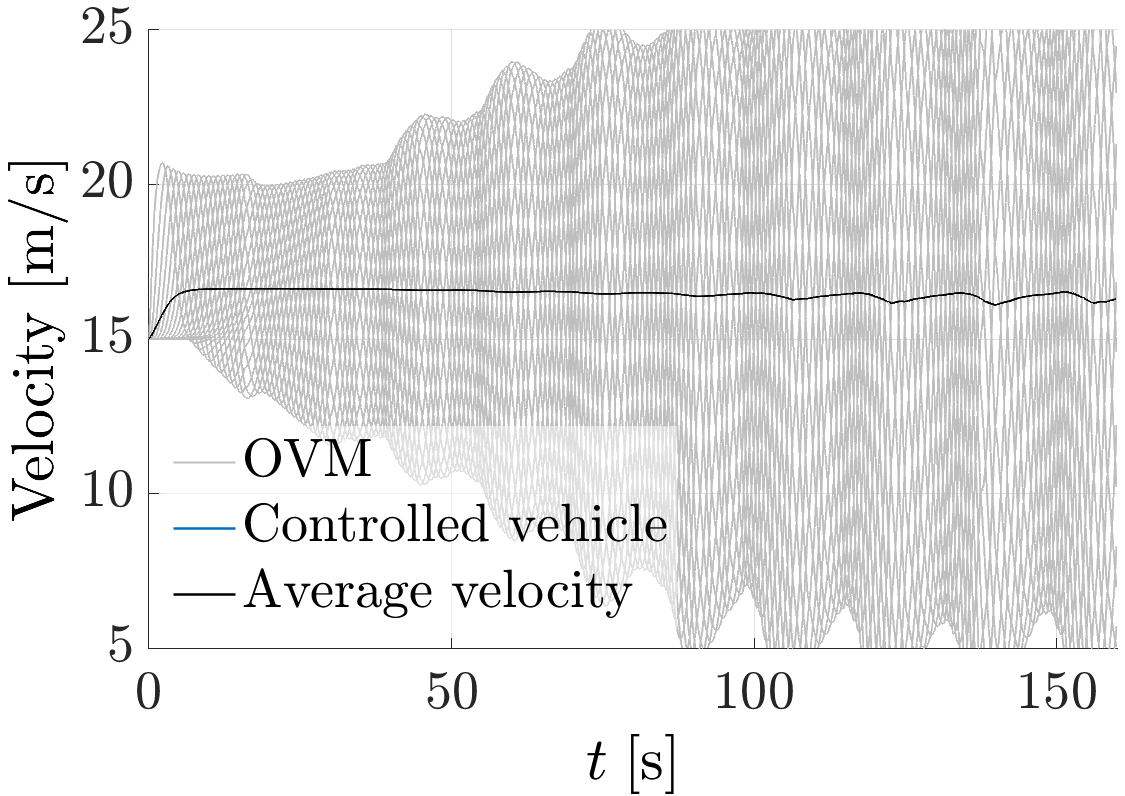}
        \caption{Time-Velocity Diagram}
        \end{subfigure}
        \begin{subfigure}{0.21\textwidth}
            \includegraphics[width=\textwidth]{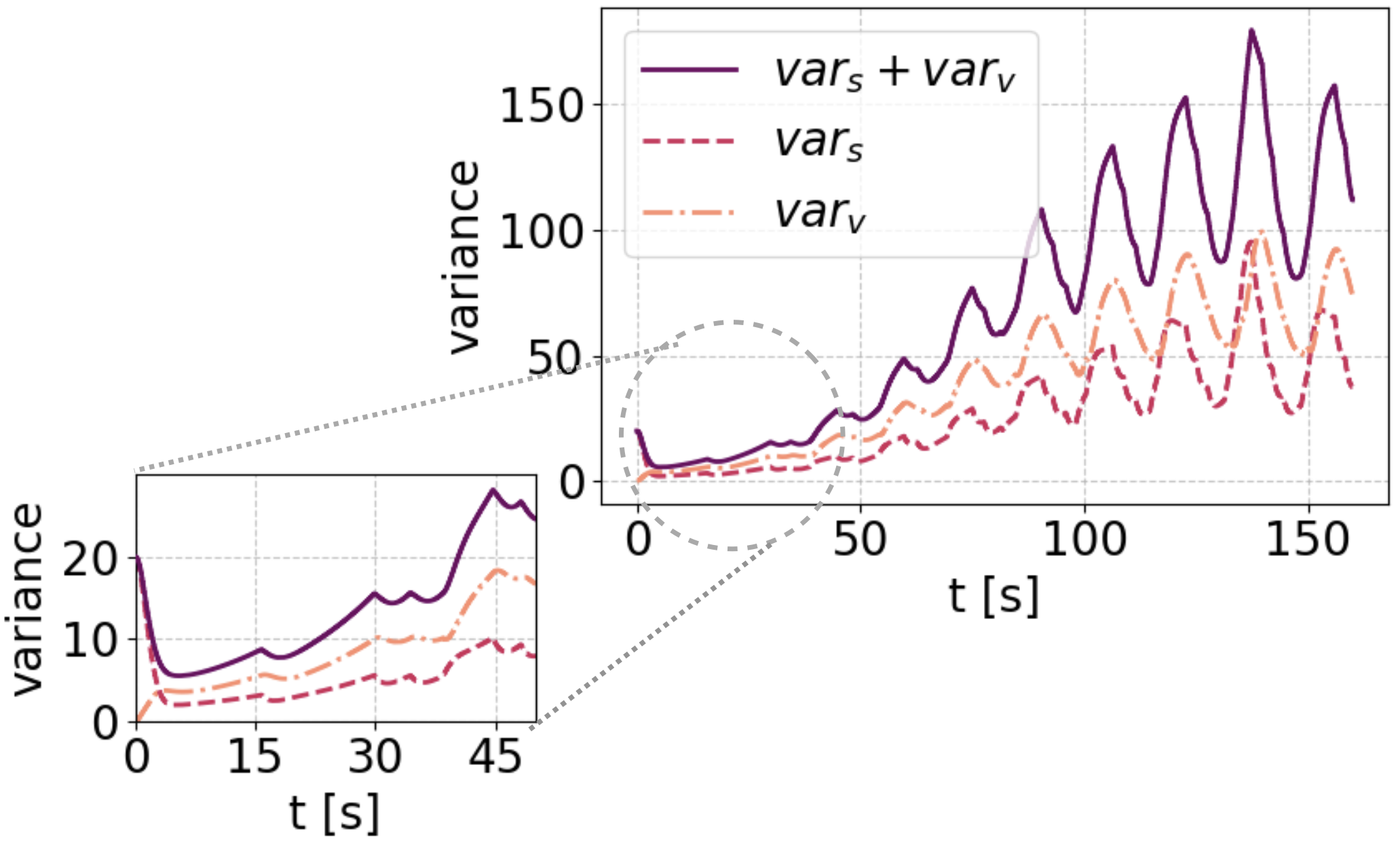}
        \caption{Variance}
        \end{subfigure}
        \begin{subfigure}{0.21\textwidth}
            \includegraphics[width=\textwidth]{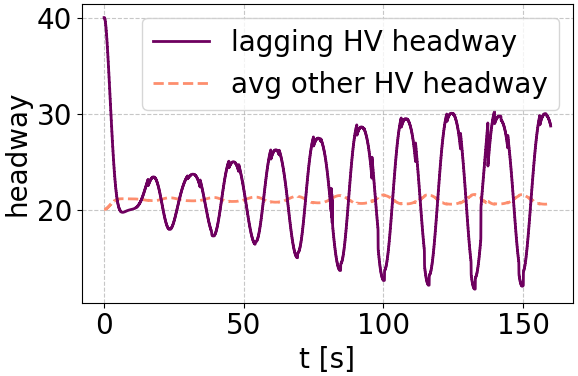}
        \caption{Headways}
        \end{subfigure}
    \end{subfigure}
    \caption{\textbf{Nominal trajectories of continuous controlled (top) and uncontrolled (bottom) periods under fixed initial states.} \\ (a, f) The initial states $\bar{z}_{0, c}, \bar{z}_{0, u}$ \sli{immediately after the AV enters or exits a stable lane}. (b, g) Time-space diagrams of the trajectories $\bar{z}_c(t), \bar{z}_u(t)$. (c, h) Time-velocity diagrams of the trajectories $\bar{z}_c(t), \bar{z}_u(t)$. (d, i) Variances throughout the trajectory $\overline{var}_c(t), \overline{var}_u(t)$ (controlled: decrease; uncontrolled: decrease followed by an increase). (e, controlled) Headway of the AV (increase followed by a decrease), the lagging HV after the AV enters, and the average headway of all HVs. (j, uncontrolled) Headway of the lagging HV prior to exiting (decrease followed by oscillation), and the average headway of all other HVs.}
    \label{fig:control_uncontrol_singlelane}\vspace{-0.2cm}
\end{figure*}

\subsection{Continuous Systems: Lyapunov Analysis}
\label{sec:analysis_control_var}
\slirebuttal{We derive the first set of variance upper bounds $f_c, f_u$ for continuous controlled and uncontrolled periods from Lyapunov analysis by converting the variance $var^l(t)$ to the error state norm $\|x^l(t)\|_2^2 = \|z^l(t) - z^*_{n_l}(t)\|_2^2$ in Lemma~\ref{lemma:varnorm}. We impose the following assumption to derive the lower bound of Lemma~\ref{lemma:varnorm}.}
\vspace{-0.4cm}
\begin{assumption}
Consider the trajectory $z^l(t)$ of lane $l$ with the headways and velocities $\{s_i^l(t)\}_{i=1}^{n_l}$ and $\{v_i^l(t)\}_{i=1}^{n_l}$. We assume that the trajectory satisfies the following condition: if $\exists k$ such that $v^l_k(t) \neq v^*_{n_l}$, then $\exists i, j$ such that $v^l_i(t) < v^*_{n_l}$ and $v^l_j(t) \geq v^*_{n_l}$.
\label{assump:varnorm}
\end{assumption}
\begin{remark}
Assump.~\ref{assump:varnorm} ensures that a low  $variance(\{v^l_i(t)\}_i)$ results in a small deviation $\|v^l(t) - v^*_{n_l}\|_2^2$, which we use to derive the lower bound of Lemma~\ref{lemma:varnorm}. 
Notably, the headway variance $variance(\{s_i^l(t)\}_i)$ always equals $\frac{1}{n_l}\|s^l(t) - s^*_{n_l}\|_2^2$, as the ring road structure ensures that the average headway (the variance ellipsoid center) always equals the equilibrium $s^*_{n_l}=\frac{C}{n_l}$; this is in contrast with the average velocity, which can deviate from $v^*_{n_l}$. Assump.~\ref{assump:varnorm} reflects typical unstable traffic systems, where stop-and-go waves are formed with vehicles moving at varied velocities (high and low) adjusted according to the varied headways, hence reflecting the implicit velocity regulation from the OVM dynamics due to the constant average headway. 

Assump.~\ref{assump:varnorm} can be further relaxed as the subspace defined by the condition is attractive: that is, the OVM dynamics will push the trajectory back to the subspace when it reaches outside. This can be shown by directly integrating the OVM formula in Eq.~\eqref{eq:ovm} with a state-dependent analysis; we leave as a future work to formally relax the assumption.
\end{remark}

We now establish a relationship between the variance and the error state norm $\|x^l(t)\|_2^2$. Notably, \slirebuttal{Lemma~\ref{lemma:varnorm} is \textit{empirically validated}} in Fig.~\ref{fig:lemma_varnorm} (see Sec.~\ref{sec:experiment_variance}  for details).
\begin{lemma} 
\label{lemma:varnorm}
At time $t$, consider a lane $l$ with an error state $x^l(t)$ and a variance $var^l(t)$. \slirebuttal{Under Assump.~\ref{assump:varnorm}}, we have
\begin{equation}
    c_1 \|x^l(t)\|_2^2 \leq var^l(t) \leq c_2 \|x^l(t)\|_2^2
\end{equation}
for some $c_1, c_2 > 0$.
\end{lemma}
\begin{proof} The variance of the headways and velocities in Eq.~\eqref{eq:variance}, which are translation-invariant functions of the state $z^l$, can be converted into a combination of two error state terms $\frac{1}{n_l}\|x^l(t)\|_2^2 - \frac{1}{n_l^2}\big(\big(\sum\limits_{i=1}^{n_l}\tilde{s}^l_i(t)\big)^2 + (\sum\limits_{i=1}^{n_l}\tilde{v}^l_i(t)\big)^2\big)$. The upper (and lower) bound can then be obtained by ignoring (and lower bounding) the second term. See Appendix~\ref{appendix:varnorm} for details.
\end{proof}
\slirebuttal{We now establish upper bounds for $\|x^l(t)\|_2^2$, and convert these bounds to into variance upper bounds.} To upper bound $\|x^l(t)\|_2^2$ for the controlled period, there exists a valid Lyapunov function $V(x^l_c(t)) = x^l_c(t)^\intercal P x^l_c(t) > 0$ with $\dot{V}(x^l_c(t)) = -x^l_c(t)^\intercal Q x^l_c(t) < 0$, where $P > 0, Q > 0$ are $n\times n$ matrices (see Sec.~\ref{sec:prelim_singlelane}); the Lyapunov function ensures that the error state norm $\|x^l_c(t)\|_2^2$ decreases with $t$, when applying the Gronwall's inequality. For the uncontrolled period, the system is unstable so a valid Lyapunov function does not exist; we hence directly bound $U(x^l_u(t)) = x^l_u(t)^\intercal x^l_u(t) = \|x^l_u(t)\|_2^2$ with Gronwall's inequality, which produces an upper bound that increases with $t$. We present the following proposition.
\begin{proposition}
\label{prop:cont_var_bd1}
Let $var_{0, c}, var_{0, u}$ be the \slirebuttal{variances} at the starting point of a continuous, controlled and uncontrolled time interval, respectively. \slirebuttal{Under Assump.~\ref{assump:varnorm},} the \slirebuttal{variances} of the linearized continuous controlled and uncontrolled systems after a control period of $t$ are upper bounded by 
\begin{equation}
\begin{aligned}
var^l_c(t) & \leq f_c(var_{0,c}, t) = \alpha_1 \exp(-\alpha_2 t) var_{0,c} \\
var^l_u(t) & \leq f_u(var_{0,u}, t) = \beta_1 \exp(\beta_2 t) var_{0,u}
\end{aligned}
\label{eq:control_var}
\end{equation}
for some $\alpha_1, \alpha_2, \beta_1, \beta_2 > 0$ and any $t \geq 0$.
\end{proposition}
\begin{proof} The proof involves applying Gronwall's inequality to $\dot{V}(x^l_c(t))$ and $U(x^l_u(t))$ defined above to obtain upper bounds on (i) $x^l_c(t)^\intercal P x^l_c(t)$, which is converted to an upper bound on $\|x^l_c(t)\|^2_2$ by $\|x^l_c(0)\|^2_2$, and (ii) $x^l_u(t)^\intercal x^l_u(t) = \|x^l_u(t)\|^2_2$ by $\|x^l_u(0)\|^2_2$. Then, we apply Lemma~\ref{lemma:varnorm} to convert the bounds into bounds on the variances. See Appendix~\ref{appendix:cont_var_bd1} for details.
\end{proof}

\vspace{-0.3cm}
\subsection{Continuous Systems: State-Dependent Analysis}

Prop.~\ref{prop:cont_var_bd1} provides a variance upper bound given \textit{any} initial state of the trajectory, based on the worst-case behavior with maximum or minimum singular value from the Lyapunov analysis. In practice, we often observe less-conservative short-term behaviors at the beginning of the trajectory: for example, the variance decreases and then increases for the uncontrolled period (see Fig.~\ref{fig:control_uncontrol_singlelane} (i) and Fig.~\ref{fig:single_lane_init} (d)). 

We can derive another state-dependent variance upper bound for trajectories whose initial states are within a region around a \textit{fixed} initial state. The following general theorem formalizes the bound, leveraging the continuous dependence of the system's solution on the initial state~\cite{khalil2001nonlinear}; as a byproduct, it also provides state information of the trajectories, which can be useful to analyze the impact of the discrete jumps in Sec.~\ref{sec:analysis_jump}: for example, a larger AV headway may lead to a large variance increase when AV exits a lane.

\begin{table*}[!htb]
\centering
\caption{\textbf{Discrete Lane-Switch.} Closed-form expressions for the headway and velocity variances change when AV exits or enters a lane, suppose the AV leaves lane $L$ and enters lane $R$, with $n_L = n, n_R = n-1$. See Fig.~\ref{fig:laneswitch_notation} in Sec.~\ref{sec:prelim_multilane} for detailed notation for the headways and velocities of individual vehicles. The closed-form expressions when the AV leaves lane $R$ and enters lane $L$ can be represented similarly by swapping $L$ and $R$ in the notation.}
\begin{tabular}{lcc}
\hline \\ [-0.7em]
         & lane $L$ (AV exits) & lane $R$ (AV enters) \\ \\ [-0.8em] \hline \\ [-0.5em]
headway  & 
{$\!\begin{aligned} \Delta_{c \rightarrow u, s}(\slirebuttal{s^L}) = \frac{1}{n_L-1}var^L_{c, s} & + \frac{1}{n_L - 1}2a^Lb^L  - \frac{1}{(n_L-1)^2n_L}L^2 \end{aligned}$}
&  {$\!\begin{aligned} \Delta_{u \rightarrow c, s}(\slirebuttal{s^R}) = -\frac{1}{n_R+1}var^R_{u, s} & - \frac{1}{n_R + 1}2a^Rb^R + \frac{1}{(n_R+1)^2n_R}L^2 \end{aligned}$} \\ \\ [-0.5em] \hline \\ [-0.5em]
velocity & {$\!\begin{aligned} \Delta_{c \rightarrow u, v}(\slirebuttal{v^L}) = & \frac{1}{n_L-1}var^L_{c, v}  - \frac{1}{(n_L-1)^2n_L}\slirebuttal{\Big(}-(n_L-1)v_n + \sum_{i: HV} v_i^L\slirebuttal{\Big)}^2 \end{aligned}$} &  {$\!\begin{aligned} \Delta_{u \rightarrow c, v}(\slirebuttal{v^R}) = & -\frac{1}{n_R+1}var^R_{u, v} + \frac{1}{(n_R+1)^2n_R}\slirebuttal{\Big(}-n_Rv_n + \sum_{i: HV} \tilde{v}_i^R\slirebuttal{\Big)}^2 \end{aligned}$}\\ \\ [-0.5em]
\hline 
\end{tabular}
\vspace*{-0.2cm}
\label{tab:analysis_jump}
\end{table*}
\begin{theorem}
\label{thm:cont_init}
For each lane $l \in \{L, R\}$, let $z(t)$ and $\bar{z}(t)$ be solutions of the continuous system dynamic with two initial conditions $z(t_0) = z_0, z(t_0) = \bar{z}_0$.
Then, 
\begin{equation}
    \|z(t) - \bar{z}(t)\|_2\leq \|z_0 - \bar{z}_0\|_2\exp[L_s(t-t_0)]
    \label{eq:cont_init_state}
\end{equation}
where $L_s$ with $s \in \{c, u\}$ represents the the Lipschitz constants of continuous controlled and uncontrolled system dynamics (Eq.~\eqref{eq:control_matrix_system} and~\eqref{eq:uncontrol_matrix_system}). We further have
\begin{equation}
var^l(t) \leq \overline{var}(t) + \frac{1}{n_l} \|z_0 + \bar{z}_0\|_2 \|z_0 - \bar{z}_0\|_2 \exp[2L_{s}(t-t_0)]
\label{eq:cont_init_var}
\end{equation}
where $var^l(t)$ and $\overline{var}(t)$ are the trajectory variances from initial states $z_0$ and $\bar{z}_0$, as in Eq.~\eqref{eq:variance}.
\end{theorem}
\begin{proof}
The proof for Eq.~\eqref{eq:cont_init_state} can be found in~\cite{khalil2001nonlinear}, Chapter 3.2. Eq.~\eqref{eq:cont_init_var} follows a similar derivation with additional algebraic manipulations, as detailed in Appendix~\ref{appendix:cont_init}.
\end{proof}
The above theorem states that there exists a tube around the nominal trajectory $\bar{z}(t)$ of a fixed initial state $\bar{z}_0$, for which the trajectory $z(t)$ of a perturbed initial state $z_0$ stays within; the theory captures the short-term behavior of the trajectory $z(t)$, guaranteeing the states and variances $z(t)$ and $var^l(t)$ to stay close to $\bar{z}(t)$ and $\overline{var}(t)$ during a short initial period\slired{, as depicted in Fig.~\ref{fig:statedep}}. As the tube width grows exponentially with time, the bound becomes loose as $t$ gets large. 

We present the following corollary of the theorem, which allows us to upper bound the state deviation and the variance within a short initial period of a trajectory. 
\slired{\begin{corollary}
\label{cor:cont_init}
For any $\epsilon > 0$, there exists a $\delta_\epsilon \geq 0$ and a time $t_\epsilon \geq 0$ such that if $\|z_0 - \bar{z}_0\|_2 \leq \delta_\epsilon$, 
\begin{equation}
\|z(t) - \bar{z}(t)\|_2\leq \epsilon \text{, and } var^l(t) \leq \overline{var}(t) + \epsilon \quad \forall t \in [0, t_\epsilon]
\label{eq:cor_cont_init}
\end{equation}
\end{corollary}}
\begin{proof}
\slired{The initial state deviation $\delta_\epsilon$ and time $t_\epsilon$ can be found by substituting the corresponding right hand side component in Eq.~\eqref{eq:cont_init_state} and Eq.~\eqref{eq:cont_init_var} with $\epsilon$ and solve for $t$. Note that $t_\epsilon \geq 0$ and $\delta \geq 0$ are typically small when $\epsilon > 0$ is small.}
\end{proof}
\noindent \textbf{Fixed initial states $\mathbf{\bar{z}_{0,c}}$ and $\mathbf{\bar{z}_{0,u}}$.} We consider the fixed initial states $\bar{z}_{0,c}$ and $\bar{z}_{0,u}$ immediately after the AV enters or exits a stable lane \textit{at equilibrium}, as depicted in Fig.~\ref{fig:control_uncontrol_singlelane} (a) and (f): for the controlled period, we consider the traffic situation where an AV enters the lane in the middle of two HVs, where the lane is previously uncontrolled and contains $n-1$ HVs drives at equilibrium; for the uncontrolled period, we consider the traffic situation immediately after AV leaves the lane, which is previously under the AV control and contains one AV and $n-1$ HVs at equilibrium. The system's solutions (nominal trajectories) $\bar{z}_u(t)$ and $\bar{z}_c(t)$ under fixed initial states can be found by explicitly integrating the system's dynamics (Eq.~\eqref{eq:control_matrix_system}), as shown in Fig.~\ref{fig:control_uncontrol_singlelane} (b) and (g).

\vspace{0.1cm}
\noindent \textbf{Uncontrolled period: variance upper bound tightening, initial variance decrease.} Fig.~\ref{fig:control_uncontrol_singlelane} (g,i) depicts the nominal trajectory $\bar{z}_u(t)$ and the variance $\overline{var}_u(t)$ under the fixed initial state $\bar{z}_{0, u}$. The trajectory's variance decreases at an initial short time period, following an increase when the lane remains uncontrolled for a long time. Intuitively, when AV exits a stable lane, it opens up space with the headway of the lagging HV immediately increases. The uncontrolled system first rapidly closes the gap, but then gradually becomes unstable and forms stop-and-go waves. The following corollary hence further tighten the variance upper bound by taking the minimum of Eq.~\eqref{eq:control_var} and~\eqref{eq:cor_cont_init} during $t \in [0, t_\epsilon^c]$\slirebuttal{.}

\label{sec:analysis_uncontrol_var}
\begin{proposition}
\label{prop:uncontrol_var}
Let $var_{0, u}$ be the initial variance of a continuous, uncontrolled time interval, whose initial state $z_{0, u}$ is close to some fixed initial state $\bar{z}_{0, u}$. Let $\epsilon > 0$. If $\|z_{0, u} - \bar{z}_{0, u}\|_2 \leq \delta_\epsilon$ for some $\delta_\epsilon \geq 0$, the variance of the linearized continuous uncontrolled system after a period of $t$ is upper bounded by 
\slired{\begin{equation}
    \begin{aligned}
        var^l(t) & \leq f_u(var_{0, u}, t) \\
        & = 
        \begin{cases}
          \overline{var}_u(t) + \epsilon,  & t \in [0, t_\epsilon] \\
          \beta_1 (\overline{var}_u(t_\epsilon) + \epsilon) \exp(\beta_2 (t - t_\epsilon)) & t \in [t_\epsilon, \infty]
        \end{cases} 
    \end{aligned}
    \label{eq:uncontrol_var}
\end{equation}}
where $\beta_1, \beta_2 > 0, \epsilon \geq 0, t_\epsilon \geq 0$, and $\overline{var}_u(t)$ is the variance function of a trajectory $\bar{z}_u(t)$ with the fixed initial state $\bar{z}_{0, u}$.
\end{proposition}
\begin{proof} The proposition is obtained by taking the minimum of the Lyapunov-based upper bound in Prop.~\ref{prop:cont_var_bd1} and the state-dependent upper bound in Cor.~\ref{cor:cont_init}.
\end{proof}

When the uncontrolled system is initialized close to $\bar{z}_{0, u}$ (AV leaves a near-stable lane), Prop.~\ref{prop:uncontrol_var} assures that the variance of the trajectory rapidly decreases initially when $t \in [0, t_\epsilon]$ (following a similar rate of decrease as $\overline{var}_u(t)$ due to the continuous dependence on the initial state), but eventually becomes unstable when $t$ is large (following the exponential upper bound from Lyapunov analysis). The combined result tightens the upper bound from Prop.~\ref{prop:cont_var_bd1}, during $t \in [0, t_\epsilon]$ and explains the initial variance decrease typically observed during the uncontrolled period.

\vspace{0.1cm}
\noindent \textbf{Controlled period: state information.} In the nominal trajectory $\bar{z}_c(t)$, the AV first increases the headway between itself and the leading HV from $s^{*u}$ to beyond $s^{*c}$ (see Fig.~\ref{fig:control_uncontrol_singlelane} (e)) while gradually equalizing the HV headways (close to $s^{*c}$); then, the AV closes its headway to $s^{*c}$, ensuring the entire system to stay near equilibrium. The above corollary assures that the perturbed trajectories $z(t)$, whose initial state $z_{0}$ in a region around the fixed initial state $\bar{z}_{0, c}$, exhibit similar behaviors during an initial short period $t \in [0, t_\epsilon^c]$ for some $t_\epsilon^c \geq 0$; we use the state information to analyze the impact of discrete jumps $\Delta_{c\rightarrow u}(z^l)$ from AV exiting a lane at the end of the controlled period (see Sec.~\ref{sec:analysis_jump} for details). If the AV exits before the lane is stabilized, a larger AV headway can lead to a large $\Delta_{c \rightarrow u}(z^l)$ and increase the variance on the exiting lane.

\subsection{Discrete AV Lane-Switch Analysis}
\label{sec:analysis_jump}

The following Proposition presents the closed-form expressions of the variance changes for both lanes $l \in \{L,R\}$ at a discrete jump when the AV switches lanes.
\begin{proposition} 
\label{prop:lanechange}
Suppose AV leaves lane $L$ and enters lane $R$. Let the states of lane $L$ and $R$ before the jump be $z^L, z^R$, whose headway and velocity components are \slirebuttal{$(s^L, v^L), (s^R, v^R)$}, whose key sub-components are summarized with special notation in Fig.~\ref{fig:laneswitch_notation}. Let the headway- and velocity- variances of lane $L$ and $R$ before the jump be $(var^L_{c, s}, var^L_{c, v})$ and $(var^R_{u, s}, var^R_{u, v})$, and after the jump be $(\tilde{var}^L_{u, s}, \tilde{var}^L_{u, v})$ and $(\tilde{var}^R_{c, s}, \tilde{var}^R_{c, v})$. 

The headway- and velocity- variance changes at the discrete jump for both lanes are $\Delta^L_{c \rightarrow u, s}(\slirebuttal{s^L}) = \tilde{var}^L_{u, s} - var^L_{c, s}; \Delta^R_{u \rightarrow c, s}(\slirebuttal{s^R}) = \tilde{var}^R_{c, s} - var^R_{u, s}; \Delta^L_{c \rightarrow u, v}(\slirebuttal{v^L}) = \tilde{var}^L_{u, v} - var^L_{c, v}; \Delta^R_{u \rightarrow c, v}(\slirebuttal{v^R}) = \tilde{var}^R_{c, v} - var^R_{u, v}$, whose closed-forms are presented in Table~\ref{tab:analysis_jump}.
\end{proposition}
\begin{proof}
The closed-form expression of the variance changes at each discrete jump can be found with standard algebra by comparing terms in the variances before and after the jump. For example, since
$var^L_{c, s} = \frac{1}{n_L}((a^L)^2 + (b^L)^2 + \sum_{\substack{i: none\\lag\;HV}} s_i^2) - \frac{1}{n_L^2} (a^L + b^L + \sum_{\substack{i: none\\lag\;HV}} s_i)^2 = \frac{1}{n_L}((a^L)^2 + (b^L)^2 + \sum_{\substack{i: none\\lag\;HV}} s_i^2) - {n_L^2} C^2$ and $\tilde{var}^L_{u,s} = \frac{1}{n_L - 1}((a^L + b^L)^2 + \sum_{\substack{i: none\\lag\;HV}} s_i^2) - \frac{1}{(n_L - 1)^2}C^2$, we have $\Delta_{c \rightarrow u,s}(\slirebuttal{s^L})= \tilde{var}^L_{u, s} - var^L_{c, s}= \frac{1}{n_L-1}var^1_{c, s} + \frac{1}{n_L - 1}2ab - \frac{1}{(n_L-1)^2n_L}C^2$. See Appendix~\ref{appendix:lanechange} for detailed derivations. 
\end{proof}

\noindent \textbf{Impact of discrete jumps on the system stability.} The closed-form expressions allow interpretations of how variance changes during AV lane-switch, which is driven by the following two factors: 1) the pre-jump variance, e.g. $var^1_{s, c}$: a high pre-jump variance, e.g. $var_{s, c}^L$, leads to increases (and decreases) in the post-jump variance, e.g. $\tilde{var}_{s, c}^L$, on the exiting (and entering) lanes, 2) the state information, as illustrated in Fig.~\ref{fig:laneswitch_goodbad}: a large headway between the AV and its two adjacent HVs leads to increases (and decreases) in the post-jump headway variances. Meanwhile, a greater deviation of the AV's velocity from the average velocity of all HVs results in decreases (and increases) the post-jump velocity variances on the exiting (and entering) lanes. 

Specifically, the closed-forms in Table~\ref{tab:analysis_jump} result in qualitative and quantitative understanding of $\Delta_{c\rightarrow u}(z^L)$ when AV exits a lane, in the following scenarios:
\begin{itemize}
\item Insufficient Control Duration (detailed in Observation~\ref{observation:fd}): due to the continuity of the system's solution with respect to the fixed trajectory $\bar{z}_c(t)$, when the AV exits a controlled lane prematurely before stabilizing the lane, it results in a large AV's headway when it exits, as illustrated by the peak in Fig.~\ref{fig:control_uncontrol_singlelane} (e). This in turn causes a large variance increase in the exiting lane due to a combination of a high pre-jump variance $var^L_{s,c}$ and a large product $a^L\cdot b^L$.
\item Equilibrium State: if the system stabilizes before the AV exits, $z^L$ is close to a known equilibrium state, which enables the computation of the closed-form $\Delta_{c\rightarrow u}(z^L)$. For instance, with Sufficient Control Duration (see Observation~\ref{observation:fd}), the system reaches the equilibrium $z^*_n$, where AV headway equals all HV headways when the AV exits, as shown at the end of trajectory in Fig~\ref{fig:control_uncontrol_singlelane} (e); this leads to a smaller variance increase in the exiting lane due to lower pre-jump variance $var^L_{s,c}$ and a smaller product $a^L\cdot b^L$.
\end{itemize}

These insights provide a systematic understanding of how AV controllers influence system stability. They also facilitate the design of efficient AV lane-switch controller, including \slired{proactively} reducing the AV headway before it exits, as discussed in the following Sec.~\ref{sec:controller}.

\begin{figure}
    \centering
    \includegraphics[width=0.48\textwidth]{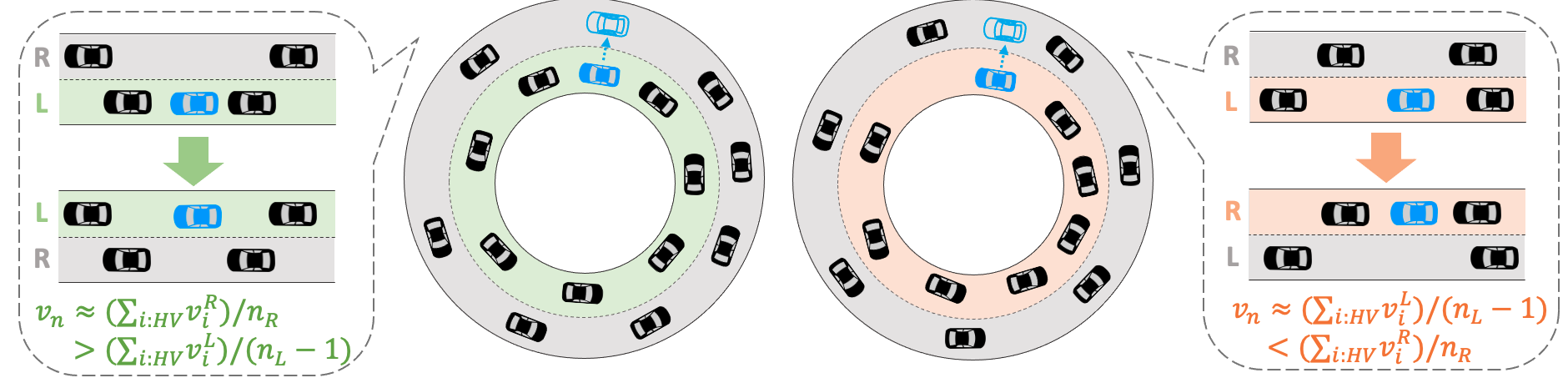}
    \caption{\textbf{Impact of AV lane switches on the system stability.} Illustrations of AV lane switches that decrease (left) or increase (right) system's variance based on the result in Table~\ref{tab:analysis_jump}.}
    \vspace*{-0.2cm}
    \label{fig:laneswitch_goodbad}
\end{figure}

\begin{figure*}[!t]
    \centering
    \includegraphics[width=0.985\textwidth]{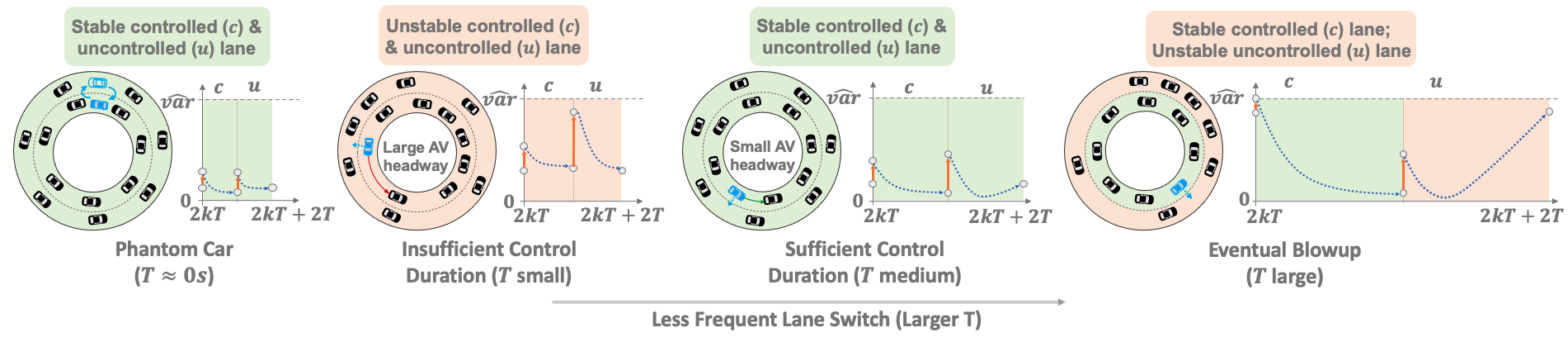}
    \caption{\textbf{Varied behaviors of the fixed-duration controller.} The theoretical analysis reveals distinct behaviors of the periodic orbit under different switch frequencies (every $Ts$). Low variance orbits occur in Phantom Car ($T \approx 0s$) and Sufficient Control Duration ($T$ medium), and high variance orbits occur in Insufficient Control Duration ($T$ small, both lanes are unstable), and Eventual Blowup ($T$ large, the uncontrolled lane is unstable). We provide a qualitative theoretical estimation of the trajectory variance and an illustration of the state at an AV lane-switch for each scenario.}
    \vspace*{-0.2cm}
    \label{fig:fixed_duration}
\end{figure*}

\section{TRAFFIC PHENOMENA ANALYSIS}
\label{sec:controller_fixed_duration}
This section applies the stability analysis in Sec.~\ref{sec:analysis} to investigate emergent traffic phenomena such as the traffic breaks and less-intrusive traffic regulation. By default, we assume when an AV is not switching lanes, it stabilizes its current lane with a full state feedback control $u^l(t) = -K x_c^l(t)$ (see Sec.~\ref{sec:experiment_setup} for details). Specifically, we examine the impact of lane-switch frequency on the system's stability by consider a fixed-duration switch strategy specified as follows.

\begin{controller}[Fixed Duration]
For the fixed-duration controller with a duration $T$, the AV controls one lane for a duration of $T$ seconds, and then switch to the other lane for another duration of $T$ seconds, and repeat.
\end{controller}
Due to the symmetry in the control duration, it is sufficient to consider the stability of lane $L$, as lane $R$ exhibits a similar behavior. Notably, for all $T$ presented in simulation (Sec.~\ref{sec:experiment}), we observe that the two-lane system converges to a periodic orbit around equilibrium. We state as an assumption below and empirically verify in Fig.~\ref{fig:assump_orbit} (see Sec.~\ref{sec:experiment_variance} for details). 
\begin{assumption}[Periodic Orbit]
\label{assump:orbit}
For any lane-switch frequency $T$, the two-lane system converges to a periodic orbit around equilibrium, that is, $var_0^{l, k} = var_0^{l, k+1}$ and $var_1^{l, k} = var_1^{l, k+1}$ for all $l \in \{L, R\}$ and all $k \geq K$ for some constant $K \geq 0$. 
\end{assumption}
\slirebuttal{We identify different stability behaviors arise from different lane-switch frequencies}, as depicted in Fig.~\ref{fig:fixed_duration} and stated below.
\begin{observation} 
\label{observation:fd}
The periodic orbit of the two-lane system exhibits the following distinct behaviors based on the lane-switch frequency, presented in the ascending order of $T$:
\begin{enumerate} 
\item \textbf{Phantom Car ($T \approx 0$, low-variance):} the system converges to a low-variance orbit, for both the controlled and uncontrolled lane, where the AV operates as if it duplicates itself to controls both lanes simultaneously.
\item \textbf{Insufficient Control Duration (small $T$, high-variance):} The system converges to a high-variance orbit, for both the controlled and uncontrolled lanes. Here, frequent AV lane-switch causes substantial variance increase, which cannot be adequately dissipated due to the limited time that the AV controls the lane before it exits.
\item \textbf{Sufficient Control Duration (medium $T$, low-variance):} The system converges to a low-variance orbit, both for the controlled and uncontrolled lane. In this scenario, the AV maintains control of each lane for a sufficient duration before exiting, effectively countering the impact of less frequent AV lane-switches.
\item \textbf{Eventual Blowup (large $T$, high-variance):} The system converges to a high-variance orbit with low variance in the controlled lane but high variance in the uncontrolled lane. \slirebuttal{Here, the AV's excessively long control over one lane leads to instability in the uncontrolled lane.}
\end{enumerate}
\end{observation}

\noindent \textbf{Formalization of Observation~\ref{observation:fd}.} We formalize the observation using the variance upper bound in Thm.~\ref{thm:general_round} (specified in Cor.~\ref{cor:traffic_round}) through examining the effects of the duration $T$ on the continuous dynamics and the discrete jumps. 

Specifically, we examine $var_1^{L, k} \leq f_c(var_0^{L, k} + \Delta_{u \rightarrow c}(z_0^{L, k}), T)$ and $var_0^{L, k+1} \leq f_{u}(var_1^{L,k} + \Delta_{c\rightarrow u}(z_1^{L,k}), T)$ for a lane $L$ at a round $k$, where the AV first controls the lane followed by an uncontrolled period. According to the specification in Cor.~\ref{cor:traffic_round}, $f_c(\cdot, T)$ decreases with $T$, while $f_u(\cdot, T)$ decreases with $T$ initially (when $t \in [0, t_\epsilon]$) but increases with $T$ in the long run (when $t > t_\epsilon$); the discrete jumps $\Delta_{u \rightarrow c}(z_0^{L, k})$ and $\Delta_{c\rightarrow u}(z_1^{L,k})$ may lead to variance increases, the extent of which depends on the state $z_0^{L, k}, z_1^{L,k}$ after the continuous controlled or uncontrolled duration $T$ (see Sec.~\ref{sec:analysis_control_var}).

\begin{enumerate}
\item \textbf{Phantom Car}: \textit{the primary reason for the low-variance orbit is the minimal stability impact from the discrete jumps.} Specifically, $\Delta_{u \rightarrow c}(z_0^{L, k}) + \Delta_{c \rightarrow u}(z_1^{L,k}) \approx 0$ due to frequent AV lane-switches; this sum equals zero in the limit when $T = 0s$, as system recovers the same state ($z_0^{L,k} = z_0^{L, k+1}$) when the AV re-enters instantaneously. Moreover, the variance upper bounds $var_1^{L, k} \leq f_c(var_0^{L, k} + \Delta_{u \rightarrow c}(z_0^{L, k}), T)$ and $var_0^{L,k+1} \leq f_{u}(var_1^{L,k} + \Delta_{c\rightarrow u}(z_1^{L,k}), T)$ decrease during both the continuous controlled and uncontrolled periods, as $T \in [0, t_\epsilon]$ is small (see Cor.~\ref{cor:traffic_round}); the system hence gradually converges to a low-variance orbit when initialized from a high-variance state.
\item \textbf{Insufficient Control Duration}: \textit{the primary reason for the high-variance orbit is insufficient time to decrease the variance during the controlled period, which fails to dissipate the substantial variance increase from the frequent discrete jumps.} Specifically, an increased $T > 0$ leads to a large $\Delta_{u \rightarrow c}(z_0^{L, k}) + \Delta_{c \rightarrow u}(z_1^{L,k}) > 0$ since the traffic state changes substantially between two AV lane-switches. Starting from a high initial variance $var_0^{L, k}$ at a round $k$, (i) the variance at the end of the controlled period $var_1^{L, k} \leq f_c(var_0^{L, k} + \Delta_{u \rightarrow c}(z_0^{L, k}), T)$ remains high when $T$ is small (see Eq.~\eqref{eq:traffic_1}). (ii) The system hence incurs a high $\Delta_{c \rightarrow u}$ when AV exits the lane, due to a large $var_1^{L, k}$ and a large headway product $a^L\cdot b^L$ (see the discussion in Sec.~\ref{sec:analysis_jump}). (iii) This leads to a high variance at the end of the uncontrolled period $var_0^{L, k+1} \leq f_{u}(var_1^{L,k} + \Delta_{c\rightarrow u}(z_1^{L,k}), T)$ when the initial variance (first component) is high and $T$ is small (see Eq.~\eqref{eq:traffic_2}), resulting in a high initial variance for the next round $k+1$.  
\item \textbf{Sufficient Control Duration}: \textit{the primary reason for the low-variance orbit is sufficient variance decrease in the controlled period, which offsets the variance increase caused by the infrequent AV lane-switch relative to the control duration $T$.} Starting from an initial state with a high variance $var_0^{L, k}$, (i) the variance $var_1^{L, k} \leq f_c(var_0^{L, k} + \Delta_{u \rightarrow c}(z_0^{L, k}), T)$ can be adequately reduced when $T$ is sufficiently long (see Eq.~\eqref{eq:traffic_1}). (ii) This results in a low $\Delta_{c\rightarrow u}(z_1^{L, k})$ as the lane is stabilized near equilibrium (see the discussion in Sec.~\ref{sec:analysis_jump}). (iii) This leads to a low $var_0^{L, k+1} \leq f_{u}(var_1^{L,k} + \Delta_{c\rightarrow u}(z_1^{L,k}), T)$ as the initial variance (first component) is low and the duration $T$ (second component) is not excessively long (see Eq.~\eqref{eq:traffic_2}). Hence, the system converges to a low-variance orbit with an appropriate $T$ (neither too short nor too long).
\item \textbf{Eventual Blowup}: \textit{the primary reason for the high-variance orbit is the high variance at the end of the uncontrolled period.} When $T$ is excessively long, the system has a high $var_0^{L, k+1} \leq f_{u}(var_1^{L,k} + \Delta_{c\rightarrow u}\big(z_1^{L,k}), T)$ that makes the uncontrolled lane unstable (see Eq.~\eqref{eq:traffic_2}). In contrast, the system's variance is low at the end of the control period (low $var_1^{L, k} \leq f_c(var_0^{L, k} + \Delta_{u \rightarrow c}(z_0^{L, k}), T)$), as the AV controls the lane for a long time (see Eq.~\eqref{eq:traffic_1}).
\end{enumerate}

\noindent \textbf{Explanation of periodic orbital behaviors.} When $T$ is close to zero (Phantom Car) or sufficiently large (Sufficient Control Duration, Eventual Blowup), the system converges to a periodic orbit (Assump.~\ref{assump:orbit}) as the AV stabilizes the lane to equilibrium before it exits, i.e. $var_1^{l, k} \approx 0$ with $z_1^{l, k}$ near equilibrium ($z^*_{n-1}$ or $z^*_n$, as further discussed in Sec.~\ref{sec:experiment_fd}); hence, the trajectory visits similar states in different rounds. we leave as a future work to explain the periodic orbital behavior for other scenarios (Insufficient Control Duration, where $var_1^{l, k} > 0$) possibly with Floquet theory~\cite{barone1977floquet}.

\noindent \textbf{Emergent Traffic Phenomena.} The observation reveals emergent behaviors in multi-lane traffic control. When $T \approx 0$ (Phantom Car), the AV with rapid lane-switches resembles duplicated phantom vehicles to stabilize both lanes. As previewed in Fig.~\ref{fig:illustration}, this mimics the traffic break (rolling slowdown) typically implemented by a patrolling vehicle weaving across multiple lanes to guide the traffic~\cite{fhwa2019, californiahighway}. This paper hence provides a theoretical justification for a single AV to stabilize multi-lane traffic by imposing traffic breaks.\\
\indent \sli{While theoretically clean, the $T\approx 0$ setting is impractical due to physical limitations of moving actual vehicles,} and may lead to disruption on the natural traffic. Extending to a control longer duration $T$ offers a more practical and less intrusive alternative. \sli{This theory indicates the importance of selecting an appropriate duration $T$,} which should be neither too short (insufficient variance reduction) nor too long (instability during excessive uncontrolled period) \slirebuttal{to balance the variance changes} from the controlled and uncontrolled periods \slirebuttal{and mitigate impacts} from the discrete jumps. 

Notably, Observation~\ref{observation:fd} uncovers the aforementioned emergent traffic phenomena through a \textit{qualitative} interpretation of the theoretical analysis. Experiments in Sec.~\ref{sec:experiment_fd} further validates the close alignment of the theory and simulation, both \textit{qualitatively} and \textit{quantitatively}. Combining the theory with simulation, we provide additional insights such as 1) the low impact of the controller on the discrete jump when AV enters a lane $\Delta_{u \rightarrow c}$ (compared to $\Delta_{c \rightarrow u}$), and 2) state-dependent properties of the periodic orbits. We refer the readers to Sec.~\ref{sec:experiment_theory} and~\ref{sec:experiment_fd} for details.

\section{CONTROLLER DESIGN}
\label{sec:controller}
Besides interpreting traffic phenomena emerging from the fixed-duration controller, the theoretical framework allows us to devise AV lane-switch strategies to enhance system stability, by reducing specific terms from the variance upper bound in Thm.~\ref{thm:general_round}. We present two controller extensions to reduce the impact from the discrete jumps $\Delta_{c \rightarrow u}$ and $\Delta_{u \rightarrow u}$, respectively.

\subsubsection{Anticipatory Control}
\label{sec:controller_pe}

\begin{figure}[!t]
    \centering
    \includegraphics[width=0.4\textwidth]{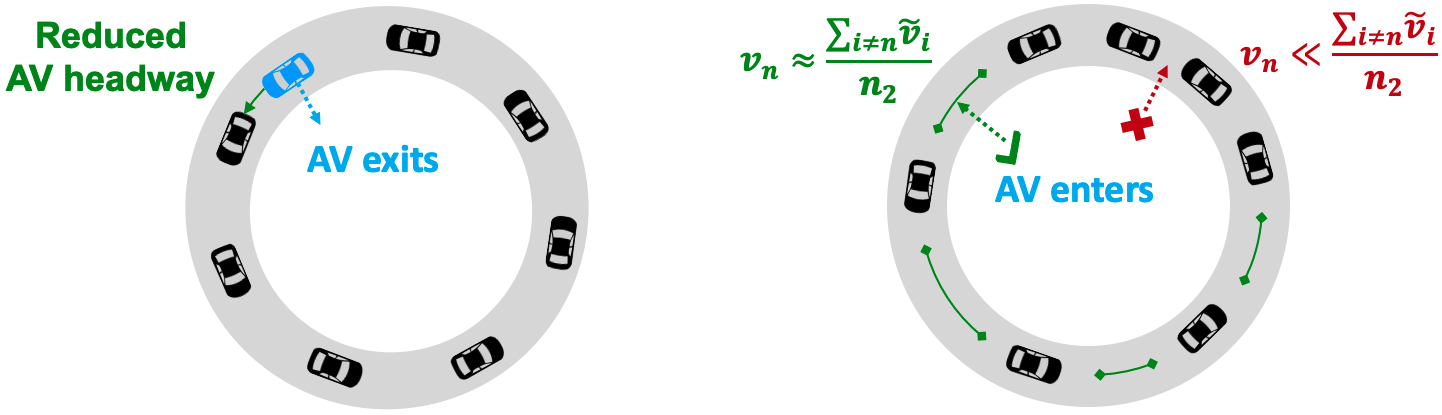}
    \caption{\textbf{Anticipatory (left) and Traffic-aware (right) controllers.} The anticipatory controller reduces the variance when AV exits a lane by decreasing the AV headway and transitioning the traffic to a new equilibrium $\hat{s}^{*l}, \hat{v}^{*l}$ prior to the exit. The traffic-aware controller reduces the variance when AV enters a lane by executing a lane-switch within a time window when the imposed criteria are satisfied (sufficient AV headway and alignment of the AV velocity to the average traffic velocity).}
    \vspace*{-0.3cm}
    \label{fig:ta_pe_controller}
\end{figure}

To reduce the stability impact when AV exits a lane ($\Delta_{c \rightarrow u}$), we apply a single-lane theoretical result derived in~\cite{zheng2020smoothing}, which states that the AV can reduce its headway and stabilize the system to a higher velocity equilibrium with the headways $\hat{s}^{*}_n = [\hat{s}_{HV}^{*}, ..., \hat{s}_{HV}^{*}, \hat{s}_{AV}^{*}]$, where $\hat{s}_{HV}^{*}$ and $\hat{s}_{AV}^{*}$ are the equilibrium headways for HVs and the AV, with $\hat{s}_{AV}^{*} \in [0, s^{*}_n]$ and $\hat{s}_{HV}^{*} = \frac{C - s_{AV}^{*}}{n-1}$, with $s^*_n = C / n$ the original equilibrium headway described in Sec.~\ref{sec:prelim_singlelane} with the original equilibrium state $z^*_n$. Similarly, the velocities $\hat{v}^{*} = (\hat{v}_{HV}^{*}, ... , \hat{v}_{HV}^{*}, \hat{v}_{AV}^{*})$, where $\hat{v}_{HV}^{*} = \hat{v}_{AV}^{*}$ is greater than $v^{*}_n$ when $\hat{s}^{*}_{AV} < s^{*}_n$. This can be achieved by a linear feedback control $\hat{u}(t) = - \hat{K} \hat{x}(t)$ with the error state $\hat{x}(t) = z(t) - \hat{z}^*_n$, computed under the linearized system around the higher-velocity equilibrium $\hat{z}^*_n = [\hat{s}_{HV}^{*}, \hat{v}_{HV}^{*}, ..., \hat{s}_{AV}^{*}, \hat{v}_{AV}^{*}]$. We modify the fixed-duration controller as follows:

\begin{controller}[\sli{Anticipatory} Control]
\label{statement:pe} \sli{This controller anticipates the traffic after the AV exits and brings the lane towards the equilibrium of the uncontrolled lane before the AV exits the lane.} Specifically, consider a control period $[0, T]$. During $t \in [0, p_{ex}\cdot T]$ where $p_{ex} \in [0, 1]$, the AV adopts the fixed-duration controller $u(t)$ to stabilize the lane to the equilibrium state $z^{*}_n$. During $t \in [p_{ex}\cdot T, T]$, the AV updates to the new controller $\hat{u}(t)$ for the equilibrium state $\hat{z}^{*}_n$ \sli{that reduces the AV headway, before it exits the lane at $T$}. 
\end{controller}

From Table~\ref{tab:analysis_jump} (the AV exit column), the anticipatory control strategy reduces $\Delta_{c\rightarrow u, s}$ by reducing the product $a^l\cdot b^l$, as the AV headway is reduced with a reduced $\hat{s}_{AV}^{*l} < s_{AV}^{*l}$. In the limit, we have $a = \hat{s}_{AV}^{*l} = 0$, and hence the new $a^l\cdot b^l = 0$ is significantly smaller than the previous equilibrium $a^l\cdot b^l = s_{HV}^{*l} = (C/n)^2$. Meanwhile, $\Delta_{c\rightarrow u, v}$ remains small as the AV velocity remains similar to the HV's average velocity. The strategy hence reduces $\Delta_{c\rightarrow u}$ when AV exits a lane.

Notably, during the controlled period, the new equilibrium $\hat{z}^*_n$ has an increased the headway variance as the AV headway is smaller than the HV headway, while the original equilibrium $z^{*}_n$ has a zero headway variance. The reduced AV headway $\hat{s}^{*}_{AV}$ may be less ideal due to safety or comfort concern for the HVs. The constant $p_{ex}$ in the statement balances the usage of the controller $u(t)$ and $\hat{u}(t)$ to stabilize the system to the original $z^{*l}$ or new $\hat{z}^{*}$, with a trade-off between lower variance during controlled period, or lower discrete jump $\Delta_{c\rightarrow u}$.

\begin{figure*}
    \centering
    \begin{subfigure}{0.25\textwidth}
    \includegraphics[width=\textwidth]{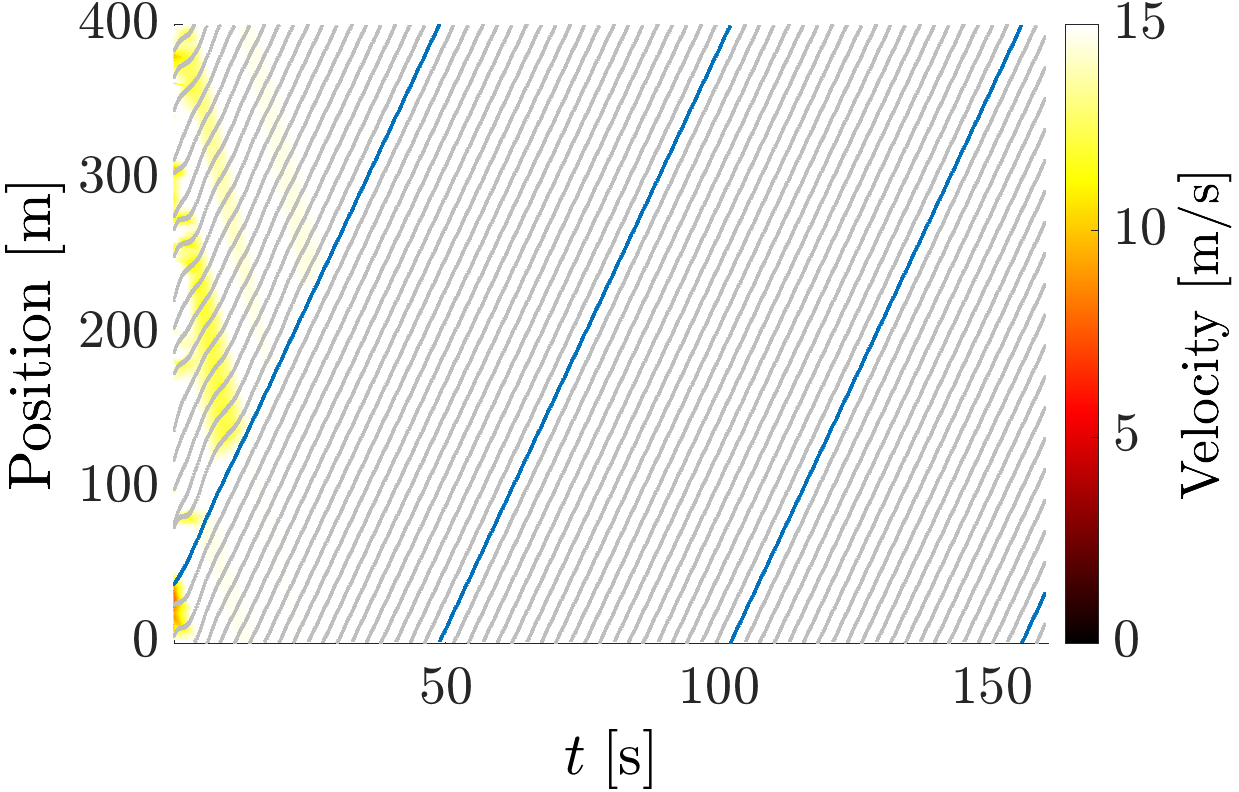}
    \caption{Controlled Time-Space Diag.}
    \end{subfigure}
    \begin{subfigure}{0.24\textwidth}
    \includegraphics[width=\textwidth]{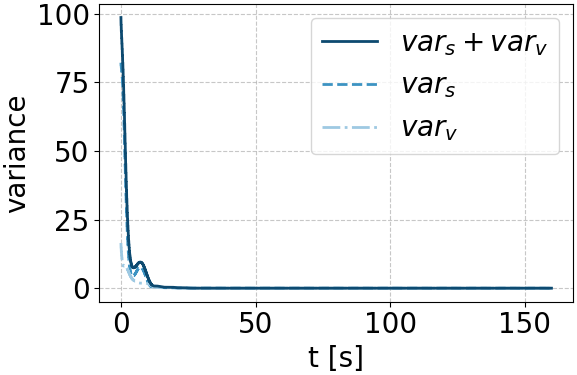}
    \caption{Controlled Variance}
    \end{subfigure}
    \begin{subfigure}{0.25\textwidth}
    \includegraphics[width=\textwidth]{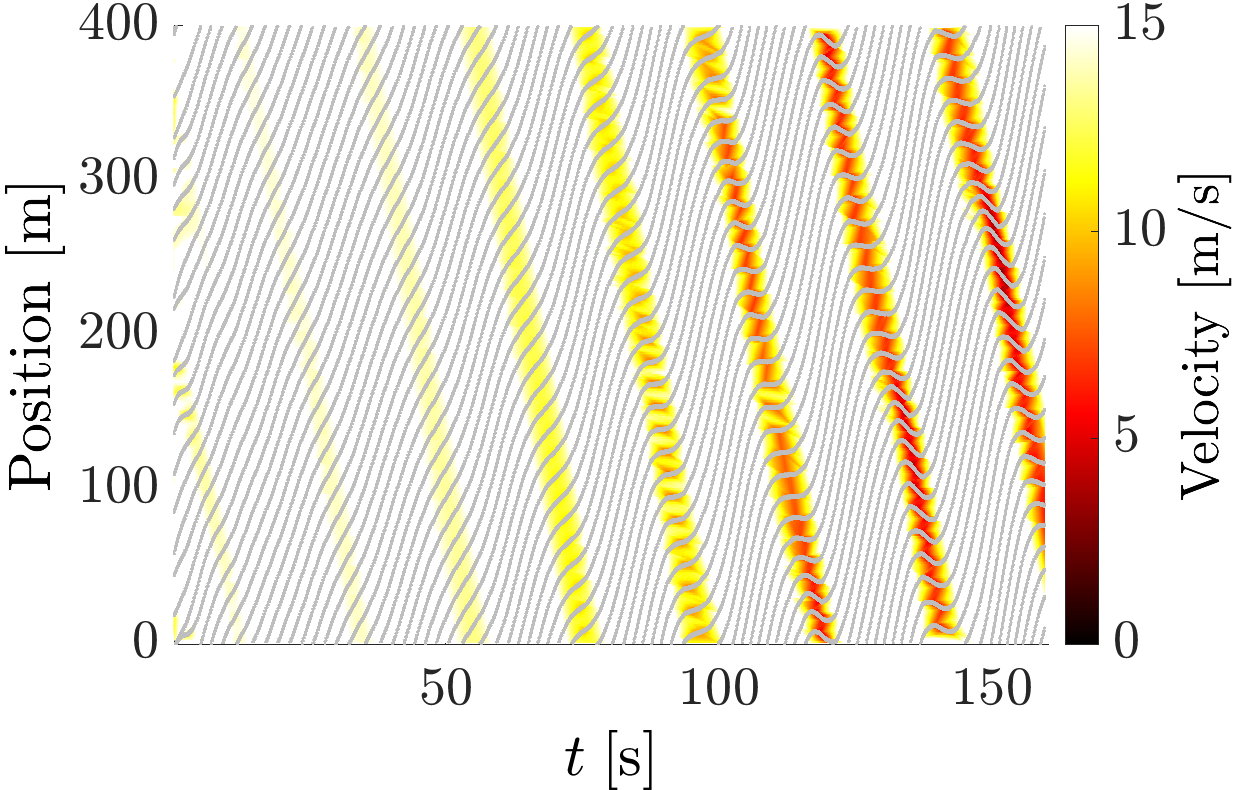}
    \caption{Uncontrolled Time-Space Diag.}
    \end{subfigure}
    \begin{subfigure}{0.24\textwidth}
    \includegraphics[width=\textwidth]{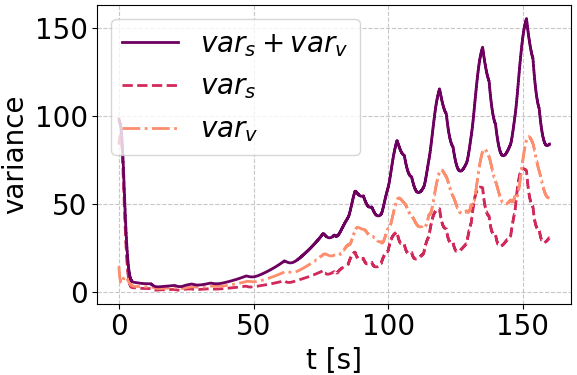}
    \caption{Uncontrolled Variance}
    \end{subfigure}
    \caption{\textbf{Controlled and uncontrolled single-lane trajectories with random initialization.} The time-space diagram and simulation variance across the trajectories for the controlled (a-b) and uncontrolled (c-d) lanes.}
    \vspace*{-0.2cm}
    \label{fig:single_lane_init}
\end{figure*}
\begin{table*}[!ht]
\centering
 \caption{System parameters for the two-lane traffic system based on the Optimal Velocity Model (OVM).}
\begin{tabular}{ |c|c|c| } 
 \hline
 \textbf{Symbol} & \textbf{Value} & \textbf{Description} \\ 
 \hline
 $C$ & $400m$ & Circumference of each ring-road $l \in \{L,R\}$, where the equilibrium spacing $s^*_{n_l} = C / n_l$ \\ 
 \hline
 $n_l$ & $20$ (c) / $19$ (u) & Number of vehicles in the controlled (c) or uncontrolled (u) ring-road, where the equilibrium spacing $s^*_{n_l} = C / n_l$ \\
 \hline
 $s_{st}$ & $5m$ & Small spacing threshold such that the optimal velocity $= 0$ below the threshold, see Eq.~\eqref{eq:ovm_f} \\
 \hline
 $s_{go}$ & $35m$ & Large spacing threshold such that the optimal velocity $= v_{max}$ above the threshold, see Eq.~\eqref{eq:ovm_f} \\
 \hline
  $v_{max}$ & $30m/s$ & Maximum optimal velocity, see Eq.~\eqref{eq:ovm_f} \\
 \hline
 $\alpha$ & $0.6$ & Driver's sensitivity to the difference between the current velocity and the desired spacing-dependent optimal velocity, see Eq.~\eqref{eq:ovm}\\
 \hline
 $\beta$ & $0.9$ & Driver's sensitivity to the difference between the velocities of the ego vehicle and the preceding vehicle, see Eq.~\eqref{eq:ovm} \\
 \hline
\end{tabular}
 \label{tab:params}
 \vspace*{-0.3cm}
\end{table*}
\vspace*{0.2cm}
\subsubsection{Traffic-aware Lane-switch}
\label{sec:controller_ta}
To reduce the variance impact when AV enters a lane ($\Delta_{u \rightarrow c}$), we can augment the fixed-duration controller with the following lane-switch strategy.

\begin{controller} [Traffic-Aware Lane-switch]
\label{statement:ta}
\sli{This controller allows the AV to enter a lane if the state after the AV enters is near the equilibrium of a controlled lane.} Specifically, consider a lane that is uncontrolled from $t = 0$. The AV enters the lane during $t \in [T - \Delta T_{ta}, T + \Delta T_{ta}]$ if the following hold:
\begin{enumerate}
\item $a^l\cdot b^l$ is large, i.e. the distances between the AV and the two adjacent HVs on the enter lane are at least $p_{en, s}\cdot s^{*}_{n_l}$ for some $p_{en, s} \in [0, 1]$.
\item $|v_n - \frac{1}{n-1}\sum_{i: HV} v^l_i|$ is small, i.e. the difference in the AV's velocity and the average HVs' velocity is at most $p_{en, v} \cdot v^*_{n}$ for some $p_{en, v} \in [0, 1]$,
\end{enumerate}
If the criteria are not met within the time window, controller mandates the AV to execute a lane-switch at $T + \Delta T_{ta}$.
\end{controller}
Based on Table~\ref{tab:analysis_jump} (the AV enters column), the two additional criteria for the traffic-aware lane-switch strategy reduces $\Delta_{u \rightarrow c, s}$ and $\Delta_{u \rightarrow c, v}$ from the fixed-duration controller, through an increase in the product $a^l\cdot b^l$ and a reduction in the velocity difference $|v_n - \frac{1}{n-1}\sum_{i: HV} v^l_i|$. Hence, the controller can reduce $\Delta_{c \rightarrow u}$ when AV enters a lane.

Notably, this strategy may change the lane-switch duration, with the magnitude of the change depending on the choice of $\Delta T_{ta}$). As a consequence, this may lead to additional variance change in the continuous controlled and uncontrolled periods. For example, a longer lane-switch duration reduces (and in general increases) the variance in the controlled (and the uncontrolled) period. Empirically, we choose the parameters $\Delta T_{ta}$, $p_{en, s}$ and $p_{en, v}$) to balance the variance change $\Delta_{u \rightarrow c}$ at the discrete jumps and during the continuous periods.

\section{NUMERICAL ANALYSIS}
\label{sec:experiment}
In this section, we compare the theoretical analysis with numerical simulation. We consider the following questions:
\begin{enumerate}
\item Is the variance function in Eq.~\eqref{eq:variance} an appropriate stability metric for the multi-lane mixed autonomy system?
\item To what extent does the variance upper bound in Cor.~\ref{cor:traffic_round} align with the variances observed in \slirebuttal{simulations}?
\item Can the theory reliably explain emergent traffic phenomena, such as traffic break in Obs.~\ref{observation:fd}, when employing a fixed-duration controller with various duration $Ts$?
\item Do the anticipatory and traffic-aware controllers effectively reduce the variance when the AV switches lanes, comparing with the fixed-duration controller? 
\end{enumerate}

\subsection{Experimental Setup}
\label{sec:experiment_setup}
We extend the single-lane implementation from Zheng et al.~\cite{zheng2020smoothing} in Python to the two-lane system with an equal circumference of $C=400\slirebuttal{m}$. As detailed in Tab.~\ref{tab:params}, we use the same set of OVM parameters, with $\alpha = 0.6, \beta = 0.9, s_{st} = 5\slirebuttal{m}, s_{go} = 35\slirebuttal{m}, v_{max} = 30\slirebuttal{m/s}$. Each lane has $19$ HVs; a single AV switches between the two lanes, resulting in $n_l = 20$ or $n_l = 19$ total number of vehicles when the lane is controlled or uncontrolled, respectively, at any  given time.
Vehicles are initialized by a uniform perturbation around the equilibrium, with the $i^{th}$ vehicle's position and velocity $(x_{i, 0}^{l}, v^l_{i, 0}) = (s^{*l} + \delta_s, v^{*l} + \delta_v)$ where $\delta_s \sim Unif[-12\slirebuttal{m}, 12\slirebuttal{m}], \delta_v \sim Unif[-7.5\slirebuttal{m / s}, 7.5\slirebuttal{m / s}]$, and $s^{*l} = C / n_l, v^{*l} = V(s^{*l})$ are the equilibrium headway and velocity (see Sec.~\ref{sec:prelim_singlelane}). The AV applies a $\mathcal{H}_2$ optimal full state-feedback controller within the controlled lane $u^l(t) = -K x^l_c(t)$, where $K \in \mathbb{R}^{1\times 2n_l}$ can be obtained by the following convex program with $K = ZX^{-1}$:
\begin{equation}
\begin{aligned}
   \min\limits_{X, Y, Z} \quad &  \text{Trace}(QX) + \text{Trace}(RY) \label{eq:H2_obj}\\
   \text{subject to} \quad & (A_c X - B_c Z) + (A_c X - B_c Z)^\intercal +H H ^\intercal \preccurlyeq 0,\\
   & \begin{bmatrix}
        Y & Z\\
        Z^\intercal & X
     \end{bmatrix} \succcurlyeq 0, X \succ 0.
\end{aligned}
\end{equation} 
\begin{equation}
    \text{where } \quad Q^{\frac{1}{2}} = \text{diag}(\gamma_s, \gamma_v, ..., \gamma_s, \gamma_v), \; R^{\frac{1}{2}} = \gamma_u, \; H = I \label{eq:H2_param}
\end{equation}
with $\gamma_s = 0.03, \gamma_v = 0.15, \gamma_u = 1$ and the performance state $ z(t) = \begin{bmatrix}Q^{\frac{1}{2}}\\ 0\end{bmatrix}x^l_c(t) + \begin{bmatrix}0 \\ R^{\frac{1}{2}}\end{bmatrix}u^l(t)$. 

We simulate the system by integrating the ordinary differential equations (Eq.~\ref{eq:control_matrix_system},~\ref{eq:uncontrol_matrix_system}) using the forward Euler method, with a discretization of $T_{step}=0.01s$. \slirebuttal{We pre-filter initial states $(x^l_{i, 0}, v^l_{i, 0})$ to exclude those leading to immediate collision (within $0.1s$), and we} equip all vehicles with a standard automatic emergency braking system $\dot{v}(t) = a_{min}, \text{ if } \frac{v_i^2(t) - v_{i-1}^2(t)}{2(s_i(t) - s_d)} \geq |a_{min}|$, where $a_{min} = -5m/s^2$ is the maximum deceleration, and $s_d = 0.5m$ is the safe distance~\cite{zheng2020smoothing}\slirebuttal{; we find the emergency braking to have negligible effect on the simulated trajectories}. 

As shown in Fig.~\ref{fig:single_lane_init}, the default uncontrolled single-lane OVM system is unstable, gradually forming stop-and-go waves in the system, whereas introducing an AV to the single-lane system effectively stabilizes the lane. The figure also depicts the simulation variances for the single-lane trajectories. The variance for the controlled OVM system gradually decreases to zero as the system reaches equilibrium (uniform headway and velocity), while the variance for the uncontrolled OVM system initially decreases (due to the short-term headway alignment), and subsequently increases (due to the asymptotic instability of the uncontrolled OVM dynamics). This observation validates using variance as a stability-metric for the single-lane system. We examine the multi-lane system next. 

\begin{figure}
    \centering
    \begin{subfigure}[b]{0.24\textwidth}
         \centering
         \includegraphics[width=\textwidth]{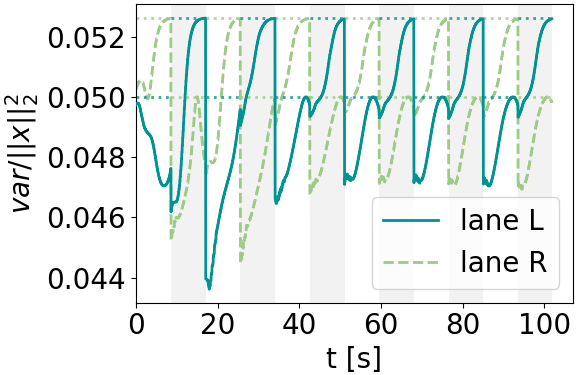}
         \caption{\quad $T = 8.5s$}
    \end{subfigure}
    \begin{subfigure}[b]{0.24\textwidth}
         \centering
         \includegraphics[width=\textwidth]{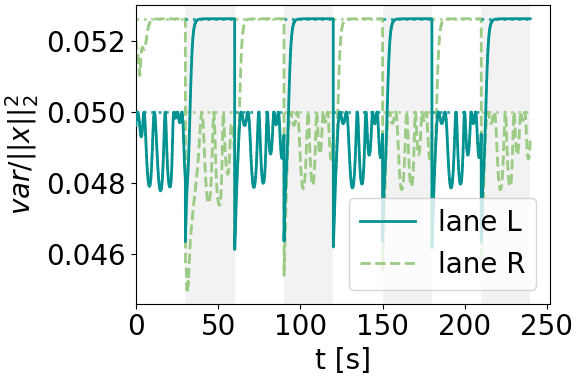}
         \caption{\quad $T = 30s$}
    \end{subfigure}
    \caption{\textbf{Variance-to-state norm ratio.} The ratio is calculated across the trajectory for various AV lane-switch frequencies. \sli{The white and gray shading represents the AV presence in lane $L$ and $R$, respectively. The horizontal dotted lines at $1/n_l = $ $1 / 19$ and $1 / 20$ represents the ratio upper bounds when the lanes are controlled and uncontrolled with $n_l = 20$ and $19$, respectively (see proof in Appendix~\ref{appendix:varnorm}).} The proximity of the ratio to the upper bound validates Lemma~\ref{lemma:varnorm} empirically. }
    \vspace*{-0.2cm}
    \label{fig:lemma_varnorm}
\end{figure}

\begin{figure}
    \centering
    \begin{subfigure}[b]{0.24\textwidth}
         \centering
         \includegraphics[width=\textwidth]{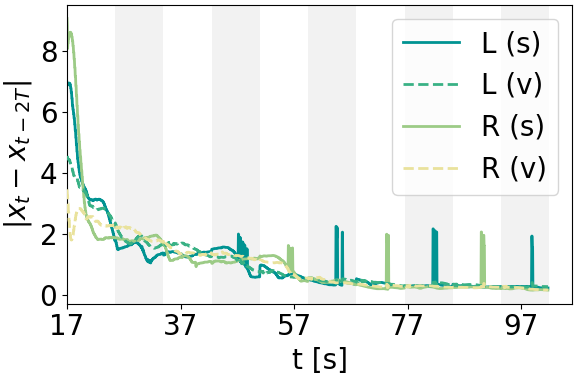}
         \caption{\quad $T = 8.5s$}
    \end{subfigure}
    \begin{subfigure}[b]{0.24\textwidth}
         \centering
         \includegraphics[width=\textwidth]{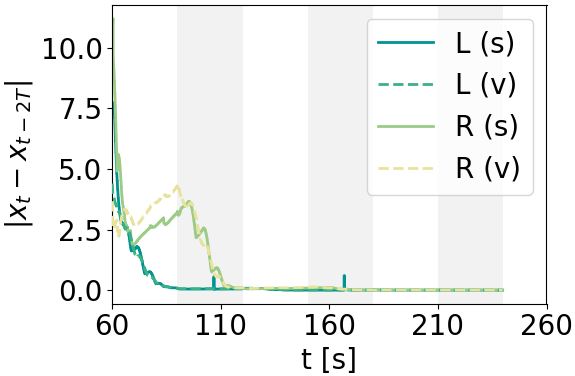}
         \caption{\quad $T = 30s$}
    \end{subfigure}
    \caption{\textbf{Periodic orbit.} We plot the distance between each state $x_t$ at time $t$ and the state $x_{t - 2T}$ in the previous round under various AV lane-switch frequencies\sli{, calculated as $\|x_t - x_{t - 2T}\|_1 / dim(x_t)$}. The convergence of the corresponding states validates Assumption ~\ref{assump:orbit} empirically. \sli{As the AV may enter a different pair of HVs at each round, we reorder the HV indices for each state $x_t$, so that the AV always appear the last vehicle in the vector representation of the state. Notably, our analysis uses the variance-based stability metric, which is invariant to the ordering of the vehicles in a specific vector representation.} }
    \vspace*{-0.2cm}
    \label{fig:assump_orbit}
\end{figure}

\vspace{-0.1cm}
\subsection{Variance as the stability metric} 
\label{sec:experiment_variance}
Fig.~\ref{fig:experiment_fd} shows the simulation trajectories and the corresponding variances of the two-lane system, where the AV employs the fixed-duration strategy under different duration $Ts$. The simulation variances appropriately captures the stability (or instability) of both lanes, establishing a clear connection between low variance and proximity to equilibrium, as well as high variance and formation of stop-and-go waves. This relationship holds for both continuous dynamics (controlled or uncontrolled) and the discrete jumps (AV exits or enters). Additionally, Fig.~\ref{fig:lemma_varnorm} plots the ratio of variance to error state norm throughout the trajectory; this ratio consistently maintains a positive value, confirming the validity of Lemma 1 (the alignment between variance and the closeness of the corresponding state to the equilibrium).

The time-space diagram reveals that the trajectories converge to periodic orbits for various duration $Ts$. Taking the $T=30s$ plot as an example, the trajectory gradually converges to a low variance orbit, with all vehicles' headways and velocities remaining roughly constant. Fig.~\ref{fig:assump_orbit} confirms the distance between the entering state at round $k$ and the subsequent rounds remains small, validating our Assump.~\ref{assump:orbit} regarding convergence to periodic orbits. 

\begin{figure}
    \centering
    \includegraphics[width=0.49\textwidth]{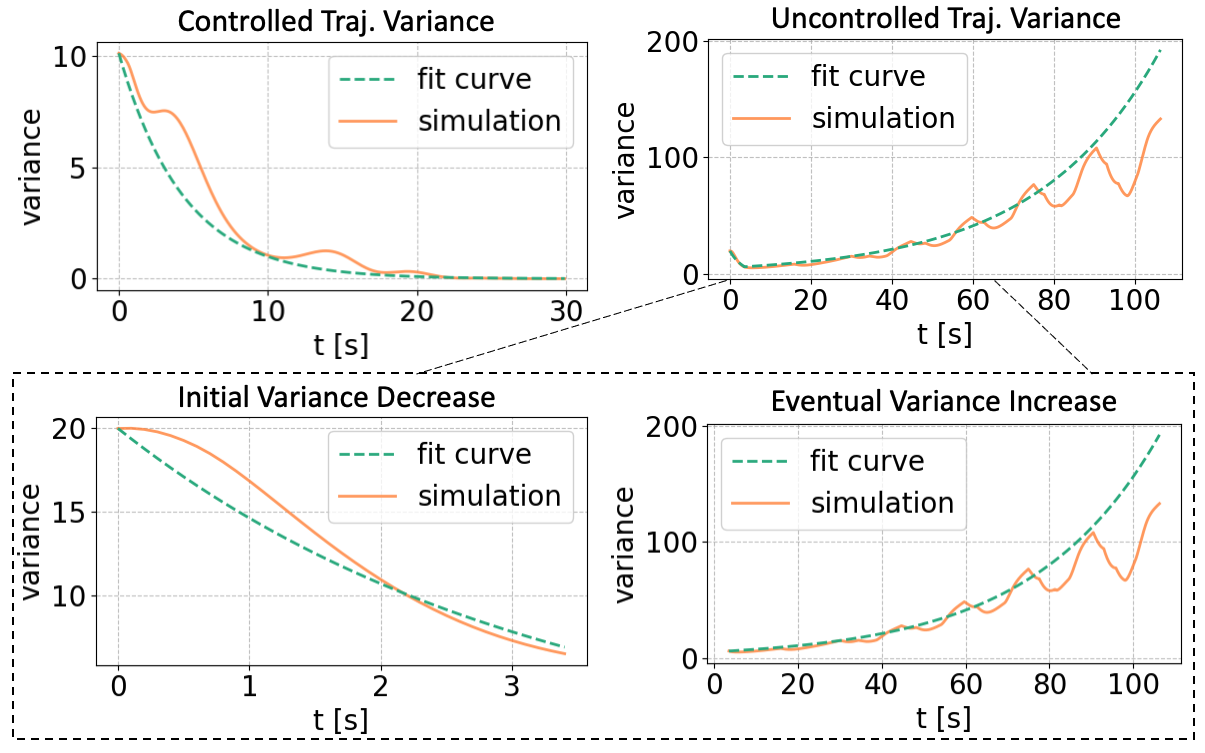}
    \caption{\textbf{Variance function fitting for the continuous dynamics}. Top: Fitted variance function for the continuous controlled and uncontrolled trajectories under fixed initial states $\bar{z}_{0, c}, \bar{z}_{0, u}$ depicted in Fig.~\ref{fig:control_uncontrol_singlelane} (a) and (f). Bottom: The fitted variance of the uncontrolled trajectory consists of an initial decrease segment followed by an eventual increase segment.}
    \vspace*{-0.2cm}
    \label{fig:single_lane_fit}
\end{figure}

\begin{figure*}
    \centering
    \begin{subfigure}[b]{0.49\textwidth}
        \centering
        \includegraphics[width=0.5\textwidth]{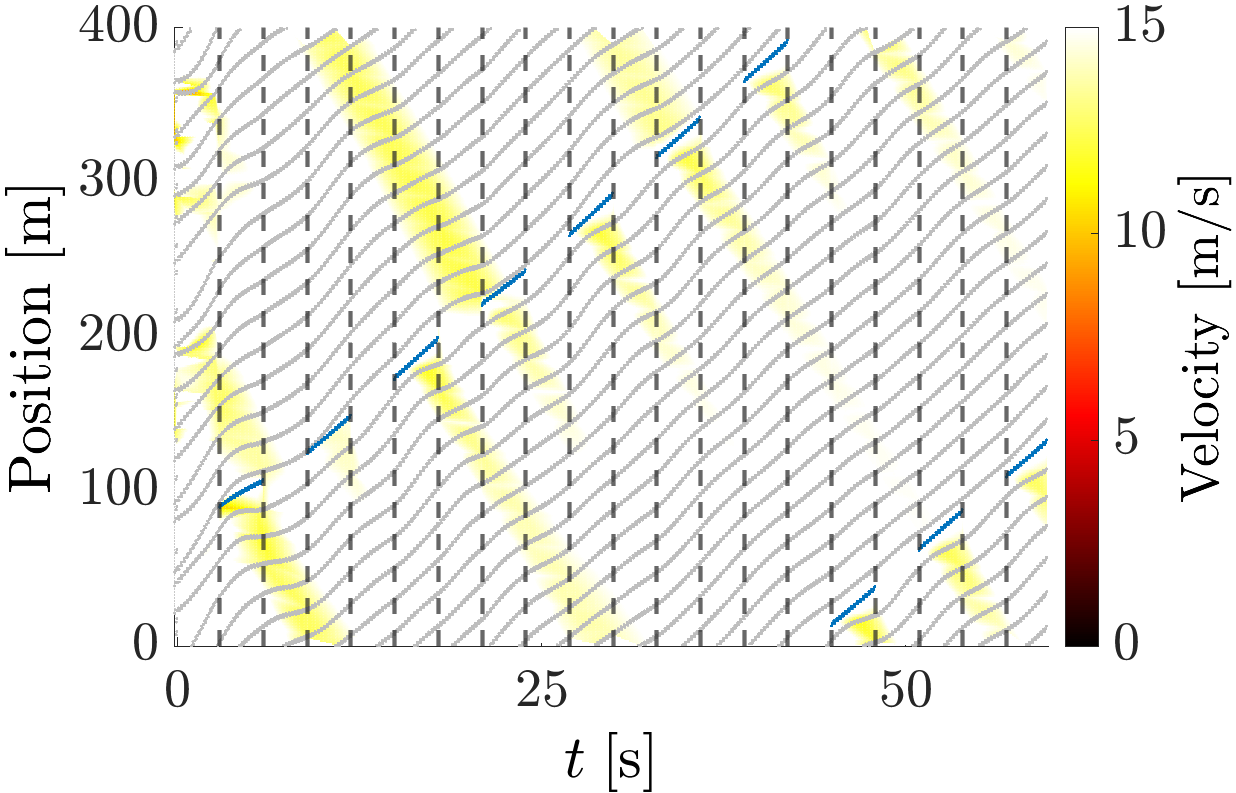}
        \includegraphics[width=0.48\textwidth]{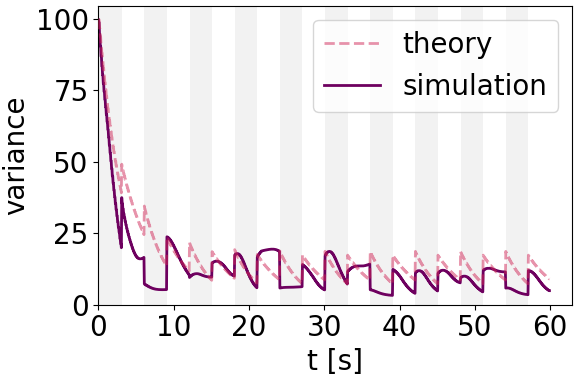}
        \caption{Phantom Car $T = 3s$}
    \end{subfigure}
    \begin{subfigure}[b]{0.49\textwidth}
        \centering
        \includegraphics[width=0.5\textwidth]{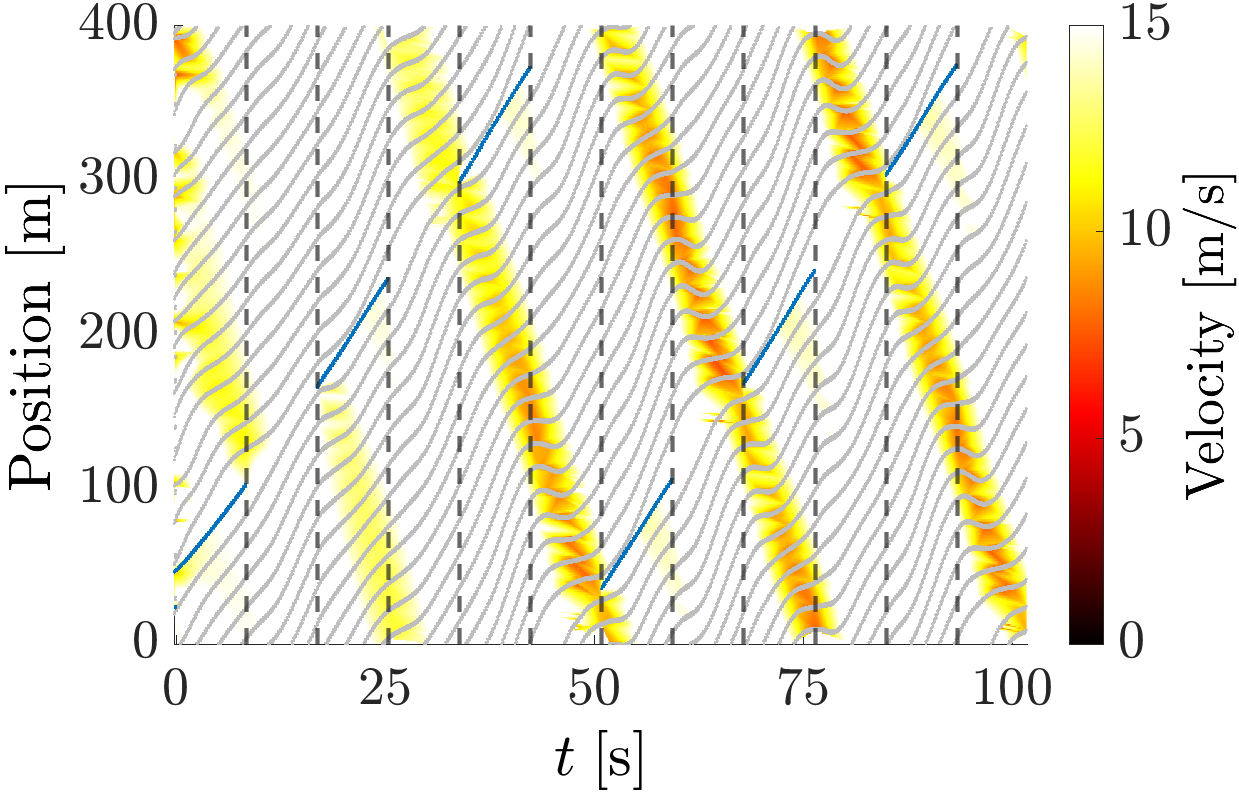}
        \includegraphics[width=0.48\textwidth]{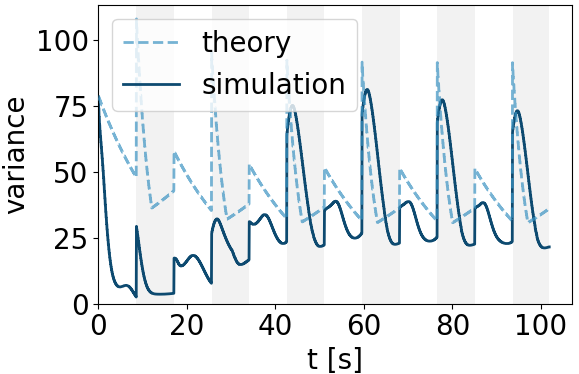}
        \caption{Insufficient Control Duration $T = 8.5s$ }
    \end{subfigure} \\ 
    \begin{subfigure}[b]{0.49\textwidth}
        \centering
        \vspace{0.3cm}
        \includegraphics[width=0.5\textwidth]{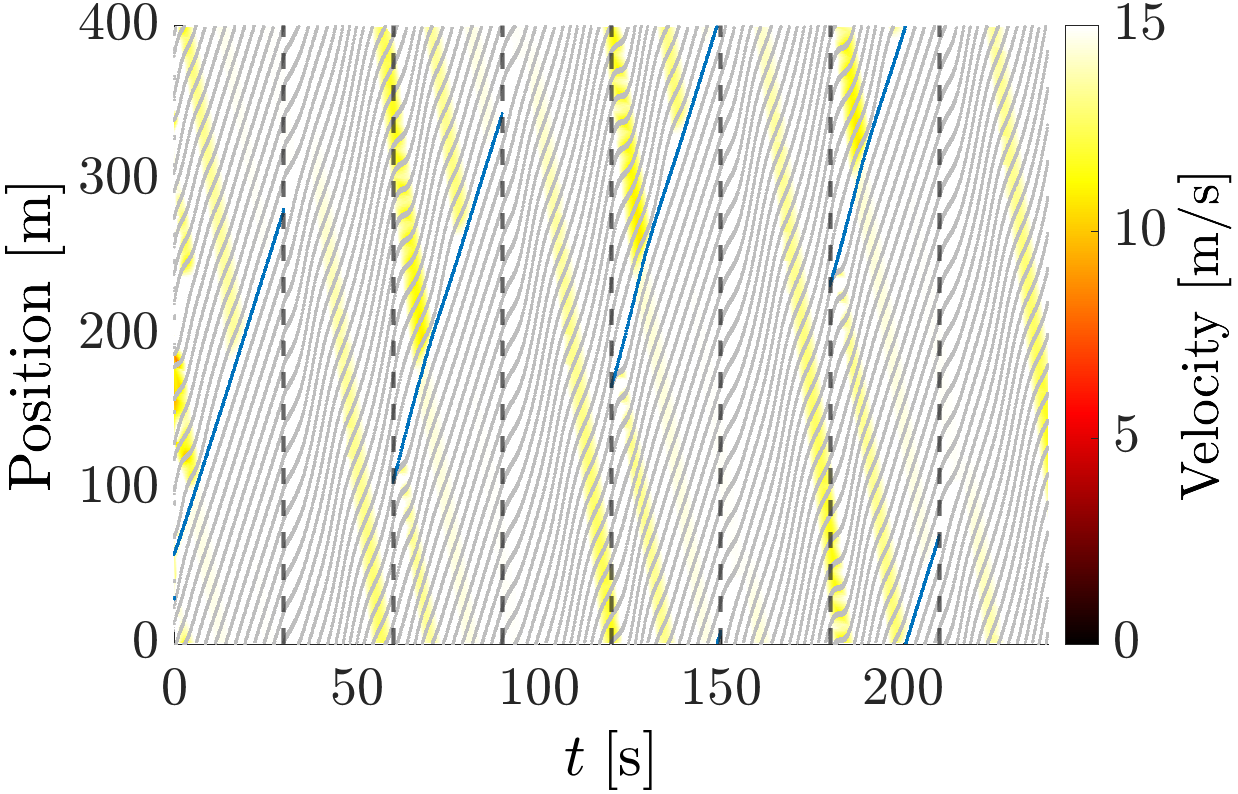}
        \includegraphics[width=0.48\textwidth]{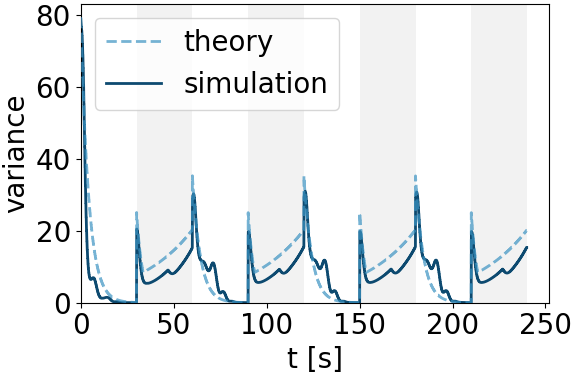}
        \caption{Sufficient Control Duration $T = 30s$}
    \end{subfigure}
    \begin{subfigure}[b]{0.49\textwidth}
        \centering
        \vspace{0.3cm}
        \includegraphics[width=0.5\textwidth]{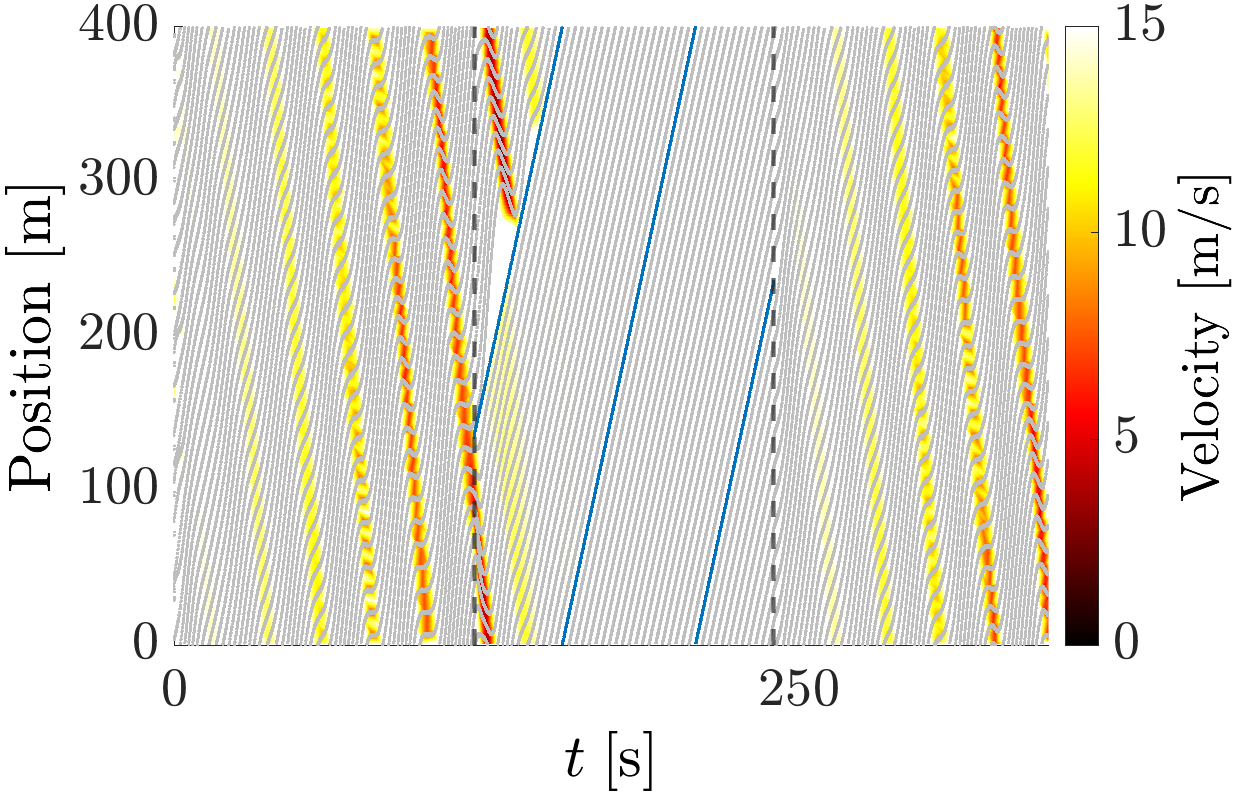}
        \includegraphics[width=0.48\textwidth]{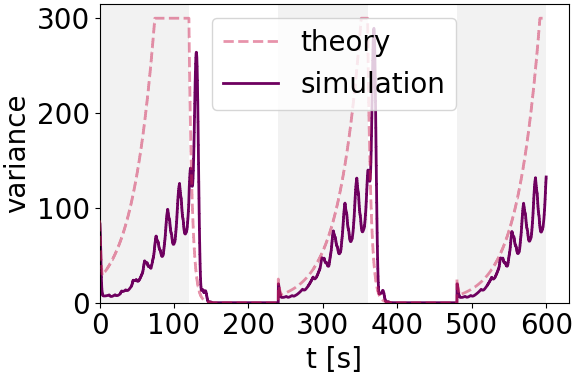}
        \caption{Eventual Blowup $T = 120s$}
    \end{subfigure}
    \caption{\sli{\textbf{Simulated trajectories and variance plots for the fixed-duration controller.} The time-space diagram, simulated and theoretically estimated variances under various AV lane-switch frequencies. The blue curves in the time-space diagram represent the AV, and the dashed lines indicate AV lane-switches. The white and gray shading in the variance plots represent the presence and absence of the AV in the lane, respectively. We show lane $L$ in (b) and (c), and lane $R$ in (a) and (d).}}
    \vspace*{-0.2cm}
    \label{fig:experiment_fd}
\end{figure*}

\vspace*{-0.2cm}
\subsection{Theoretical Variance Upper Bound}
\label{sec:experiment_theory}
We compare the theoretical variance upper bound in Thm.~\ref{thm:general_round} with the simulation variance to validate the theoretical variance upper bound as an appropriate stability metric, given the alignment between the simulation variance and system's stability.

\textbf{Parameter Fitting.} Cor.~\ref{cor:traffic_round} requires parameter values for the continuous dynamics $(\alpha_1, \alpha_2, \beta_1, \beta_2, t_\epsilon)$. We fit the parameter values from the nominal single-lane trajectories under fixed initial states $\bar{z}_{0, c}, \bar{z}_{0, u}$ in Fig.~\ref{fig:control_uncontrol_singlelane}, where we also fit the $\overline{var}_u(t)$ function with exponential functions. The fitted curves are displayed in Fig.~\ref{fig:single_lane_fit}. The resulting variance function for the continuous controlled and uncontrolled dynamics are 
\begin{equation}
    f_c(var_{0,c}, t) = \exp(-0.2319 t) var_{0,c}
    \label{eq:cont_var_fit}
\end{equation}
\begin{equation}
    \begin{aligned}
    & f_u(var_{0,u}, t) = \begin{cases} \exp(-0.3115 t) var_{0,u}, \hspace{-0.3cm} & \text{ if } t \in [0, 3.5s] \\ \exp\left(0.03313 (t - 3.5)\right) var'_{0,u}, \hspace{-0.3cm} & \text{ if } t\geq 3.5s\\  \end{cases}, \\
    \end{aligned}
    \label{eq:uncont_var_fit}
\end{equation}
\noindent where $var'_{0,u}  = \exp(-0.03115 \cdot 3.5) var_{0,u} = f_u(var_{0, u}, 3.5)$ is the variance evaluated at $t = 3.5s$, which serves as the initial variance for the increasing period $t \geq 3.5s$.
We set the constant multiplier to each exponential as $\alpha_1 = \beta_1 = 1$, so that the fitted curve starts with the same variance as the simulation curve (the values at $t = 0$ match). \slired{For Insufficient Control Duration ($T = 8.5s$), we empirically observe slower variance decrease during the controlled period $f_c$, due to the challenges imposed by the unstable system to the controller. We hence adjust the fitted controlled dynamics as $f_c(var_{0,c}, t) = \exp(-0.2319 / 4 t) var_{0,c} = \exp(-0.0580 t) var_{0,c}$ to account for the slower variance decrease.}

While the proof of our theory provides parameter estimates using min/max singular values, fitting the parameters based on the single-lane trajectories lead to tighter estimates of the variance upper bounds. We leave as a future work to find tight value estimates solely from the theory. 

\begin{table}[!t]
\centering
\caption{\textbf{Estimated and simulation variance increases at the discrete jumps.} The estimated variance increases, divided into high, middle, and low values based on the theoretical analysis, are consistent with the simulation variance (displayed with both the mean and standard deviation).\vspace*{-0.05cm}}
\scalebox{0.95}{
\begin{tabular}{ccccc}
\hline \\[-0.85em]
& \begin{tabular}[c]{@{}c@{}}Phantom Car\\ $T = 3s$\end{tabular}  & \begin{tabular}[c]{@{}c@{}}Insufficient \\ Control Dur.\\ $T = 8.5s$\end{tabular}  & \begin{tabular}[c]{@{}c@{}}Sufficient \\ Control Dur.\\ $T = 30s$\end{tabular}   & \begin{tabular}[c]{@{}c@{}}Eventual\\ Blowup\\ $T = 120s$\end{tabular}    \\ \\[-0.85em]
\hline \\[-0.85em]
\begin{tabular}[c]{@{}c@{}}estimated\\ $\Delta_{c\rightarrow u}$\end{tabular}   & 10 & 60 & 20& 20  \\ \\[-0.85em]
\begin{tabular}[c]{@{}c@{}}estimated\\ $\Delta_{u\rightarrow c}$\end{tabular}   & 10   & 15    & 15   & 5    \\ \\[-0.85em]
\hline \\[-0.85em]
\begin{tabular}[c]{@{}c@{}}simulation\\ $\Delta_{c\rightarrow u}$\end{tabular} & \begin{tabular}[c]{@{}c@{}}1.45 \\ (10.45)\end{tabular} & \begin{tabular}[c]{@{}c@{}}53.54\\ (1.61)\end{tabular} & \begin{tabular}[c]{@{}c@{}}19.32\\ (0.01)\end{tabular} & \begin{tabular}[c]{@{}c@{}}18.89\\ (1.18)\end{tabular} \\ \\[-0.85em]
\begin{tabular}[c]{@{}c@{}}simulation\\ $\Delta_{u\rightarrow c}$\end{tabular} & \begin{tabular}[c]{@{}c@{}}7.15 \\ (2.11)\end{tabular} & \begin{tabular}[c]{@{}c@{}}13.06\\ (0.28)\end{tabular} & \begin{tabular}[c]{@{}c@{}}10.42\\ (2.12)\end{tabular} & \begin{tabular}[c]{@{}c@{}}6.86\\ (1.82)\end{tabular} \\ \\[-0.85em]
\hline
\end{tabular}
}
\vspace*{-0.5cm}
\label{tab:discrete_jump_value}
\end{table}
We classify the parameter values for discrete jumps as high, middle and low based on lane-switch duration, using qualitative interpretations from the theoretical analysis in Sec.~\ref{sec:analysis_jump} and~\ref{sec:controller_fixed_duration}. The estimated values are presented in Table~\ref{tab:discrete_jump_value}. Notably, the estimated jump values closely align with simulation, whose mean and standard deviation are shown in the same table. We have the following two specific cases:
\begin{itemize}
\item $\boldsymbol{\Delta_{c\rightarrow u}}$ \textbf{(AV exits)}: Insufficient Control Duration has a high estimated jump value ($60$) due to insufficient variance decrease during the control period (high $var_1^{l, k}$ and a large headway product $a \cdot b$ that make $\Delta_{c\rightarrow u}$ high, see Table~\ref{tab:analysis_jump}). Sufficient Control Duration and Eventual Blow up has a lower estimated jump value ($20$), as the systems are near stable when AV exits (low $var_1^{l, k}$ and a smaller headway product). Phantom Car has the lowest estimated jump value ($10$) as high frequent switching results in a near stable multi-lane system with $\Delta_{c \rightarrow u} + \Delta_{c \rightarrow u} \approx 0$. 
\item $\boldsymbol{\Delta_{u \rightarrow c}}$ \textbf{(AV enters)}: we estimate a jump value ($15$) for Insufficient and Sufficient Control, a lower jump value ($10$) for Phantom Car due to the small net effects $\Delta_{c\rightarrow u} + \Delta_{u\rightarrow c}$ from consecutive AV exit and re-enter, and a lowest jump value ($5$) for Eventual Blowup due to a large $var^2_{u}$ leading to a lower variance increase at the jump (see Table~\ref{tab:analysis_jump}). 
\end{itemize}
Empirically, we observe $\Delta_{u \rightarrow c}$ has less variability and is smaller than $\Delta_{c \rightarrow u}$ for various AV control duration $Ts$. As time passes, the state of the uncontrolled (and controlled) lane become less (and more) correlated with the AV action. Hence, the variance jump $\Delta_{u \rightarrow c}(z^l_u)$, when AV enters a lane that is uncontrolled for $Ts$, is less affected by the AV's control strategy than $\Delta_{c \rightarrow u}(z^l_c)$, when AV exits a lane after controlling it for $Ts$. This hence results in more homogeneous jump values for $\Delta_{u \rightarrow c}(z^l_u)$ under different control duration $Ts$. \slired{Furthermore, $\Delta_{u \rightarrow c}$ is smaller than $\Delta_{c \rightarrow u}$ due to the opposite signs of the variance and the state information in the closed-form expressions (for example, a large pre-jump variance leads to a large $\Delta_{c\rightarrow u}$ but small $\Delta_{u\rightarrow c}$, as seen in  Table~\ref{tab:analysis_jump}).}

Alternatively, we can estimate the discrete jumps by imposing assumptions on state (e.g. assume $a^R, b^R$ and the pre-jump variance are within a close neighborhood from the equilibrium), and compute the closed-form jump values of the estimated states using Table~\ref{tab:analysis_jump}. We leave this as a future work.

\begin{figure}
    \centering
    \includegraphics[width=0.49\textwidth]{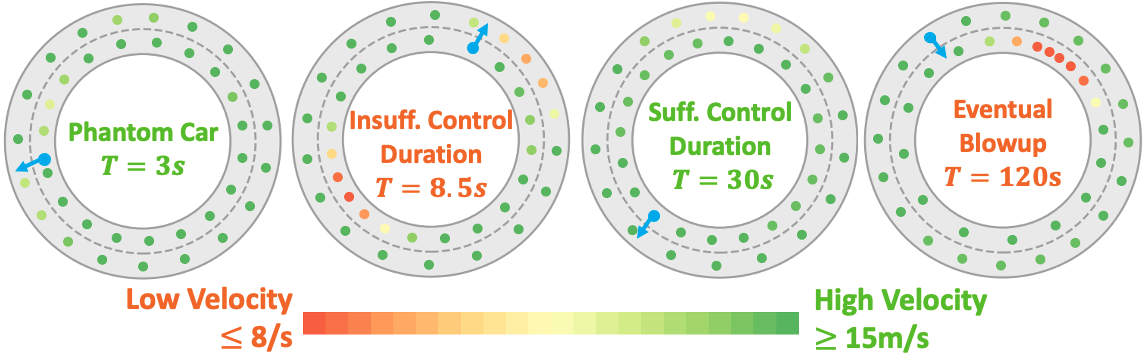}
    \caption{\textbf{Simulated traffic states of the fixed-duration controllers.} Depiction of the simulated traffic states under various lane-switch frequencies, when the AV (blue circle) switches lanes (the arrow indicates the AV movement). The traffic flow moves counter-clockwise. The color of each HV represents its velocity (red indicates traffic jams). The simulated states matches the theoretical illustration in Fig.~\ref{fig:fixed_duration}.}
    \vspace*{-0.2cm}
    \label{fig:experiment_fd_illustration}
\end{figure}

\begin{figure*}[!t]
    \centering
    \begin{subfigure}[b]{0.49\textwidth}
        \centering
        \includegraphics[width=0.5\textwidth]{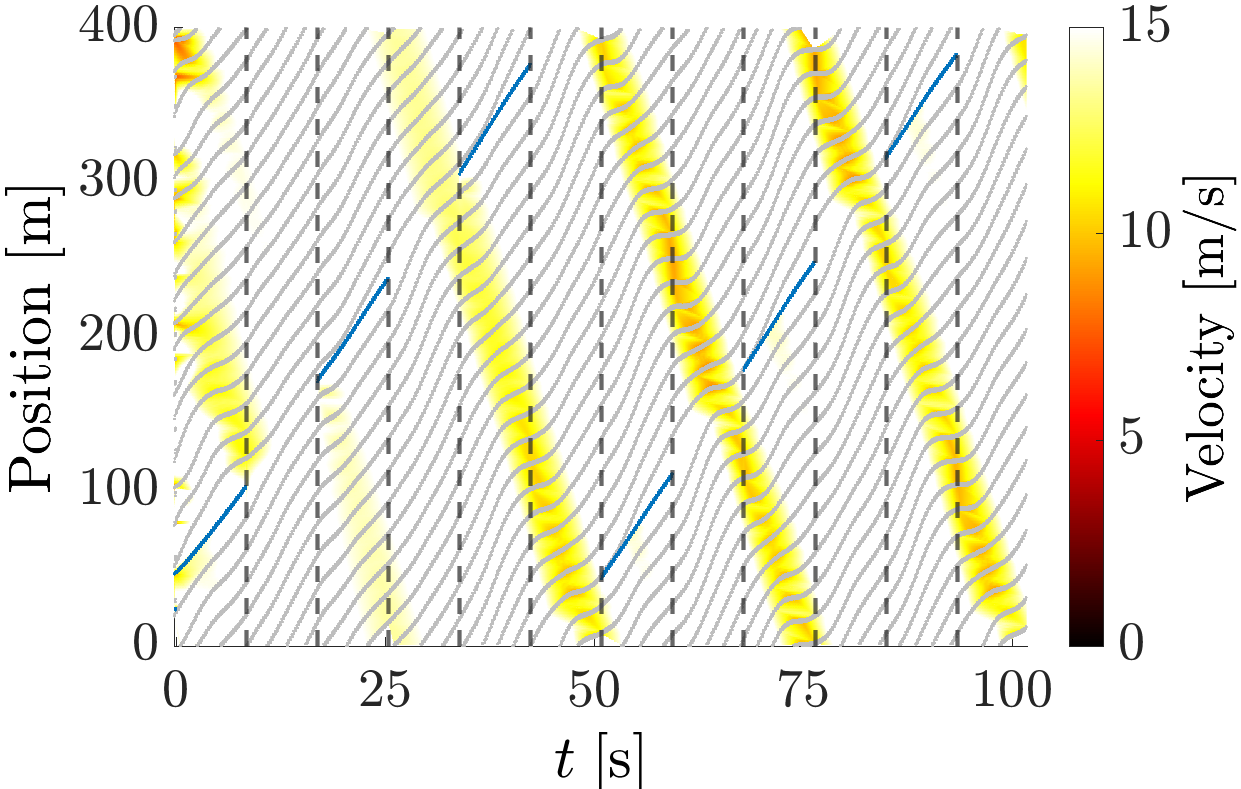}
        \includegraphics[width=0.48\textwidth]{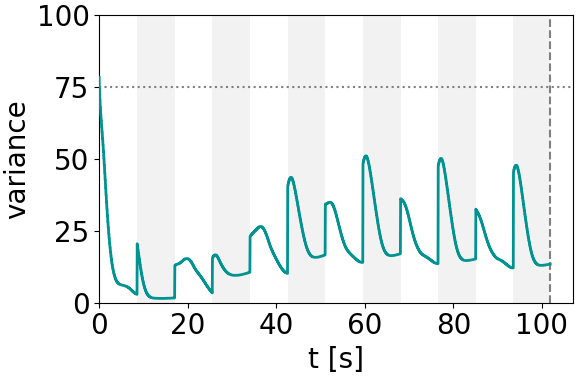}
        \caption{$T = 8.5s$, Anticipatory Control}
    \end{subfigure}
    \begin{subfigure}[b]{0.49\textwidth}
        \centering
        \includegraphics[width=0.49\textwidth]{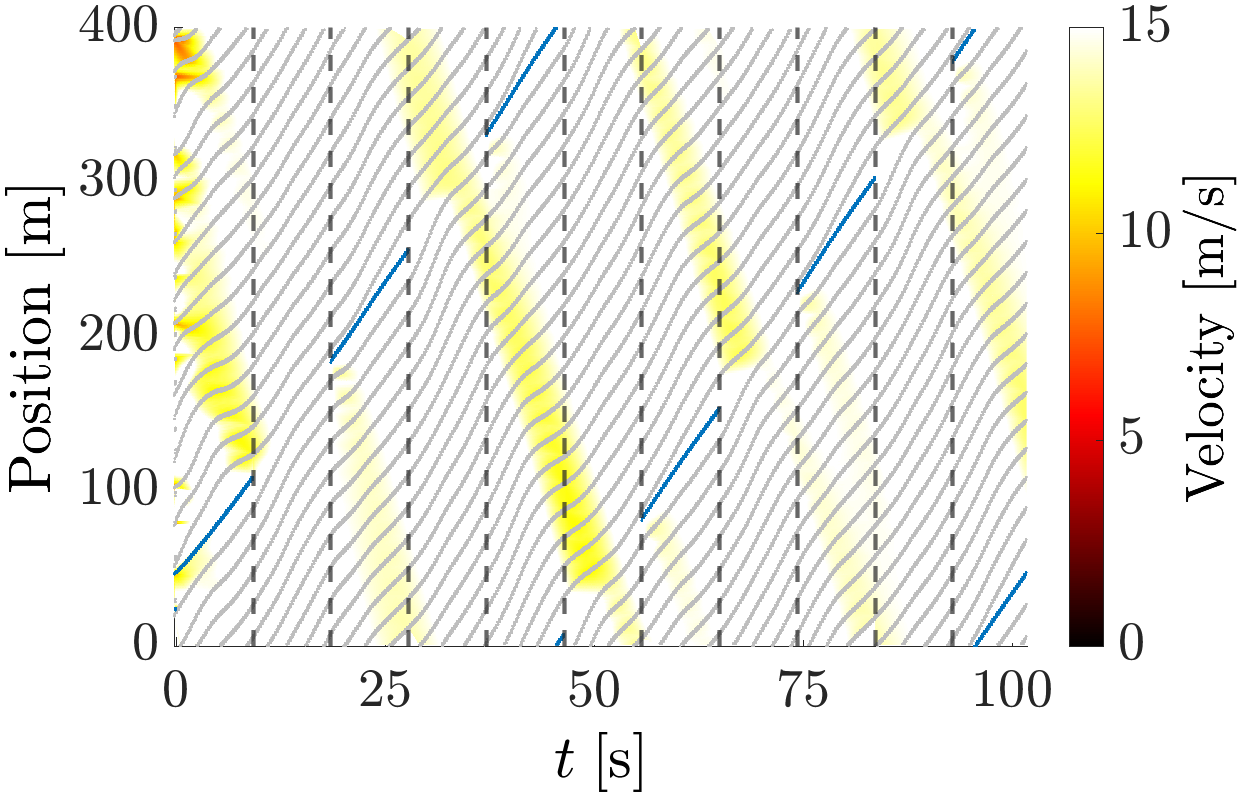}
        \includegraphics[width=0.48\textwidth]{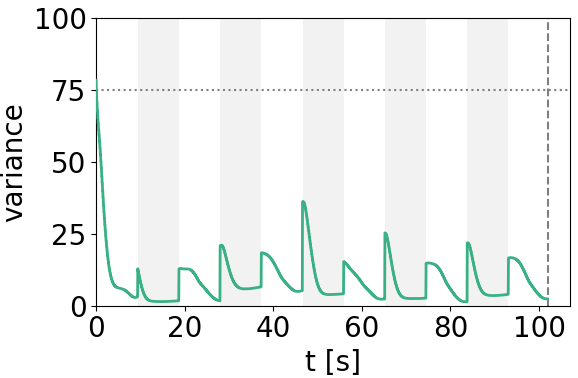}
        \caption{$T = 8.5s$, Integrated Anticipatory and Traffic-Aware Control}
    \end{subfigure} \\
    \begin{subfigure}[b]{0.49\textwidth}
        \centering
        \vspace{0.3cm}
        \includegraphics[width=0.5\textwidth]{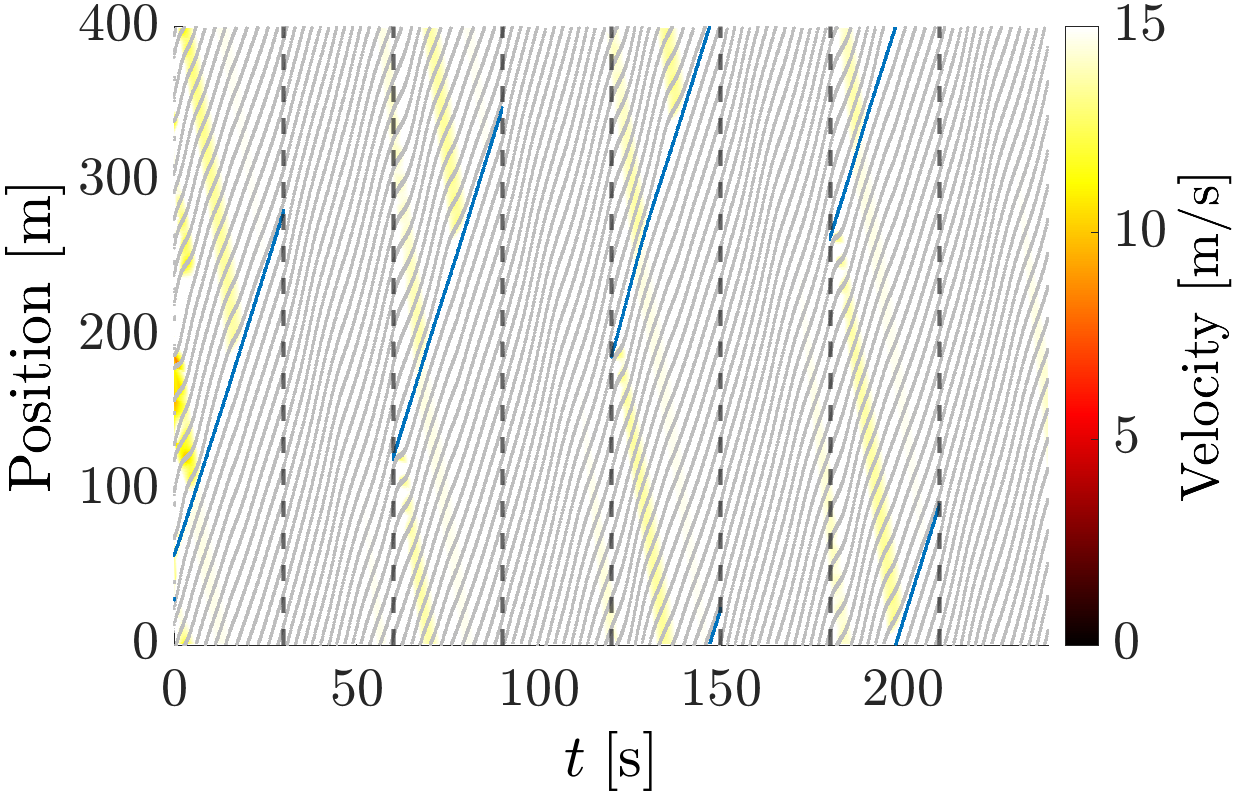}
        \includegraphics[width=0.48\textwidth]{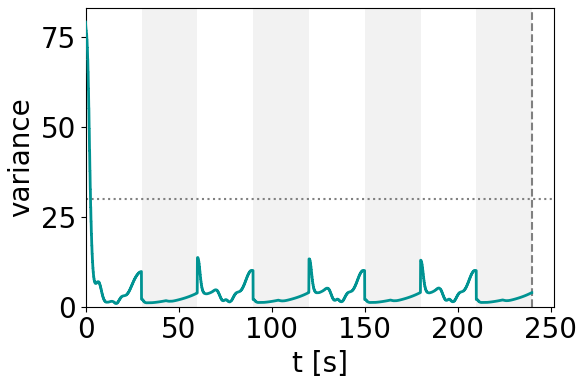}
        \caption{$T = 30s$, Anticipatory Control}
    \end{subfigure}
    \begin{subfigure}[b]{0.49\textwidth}
        \centering
        \vspace{0.3cm}
        \includegraphics[width=0.5\textwidth]{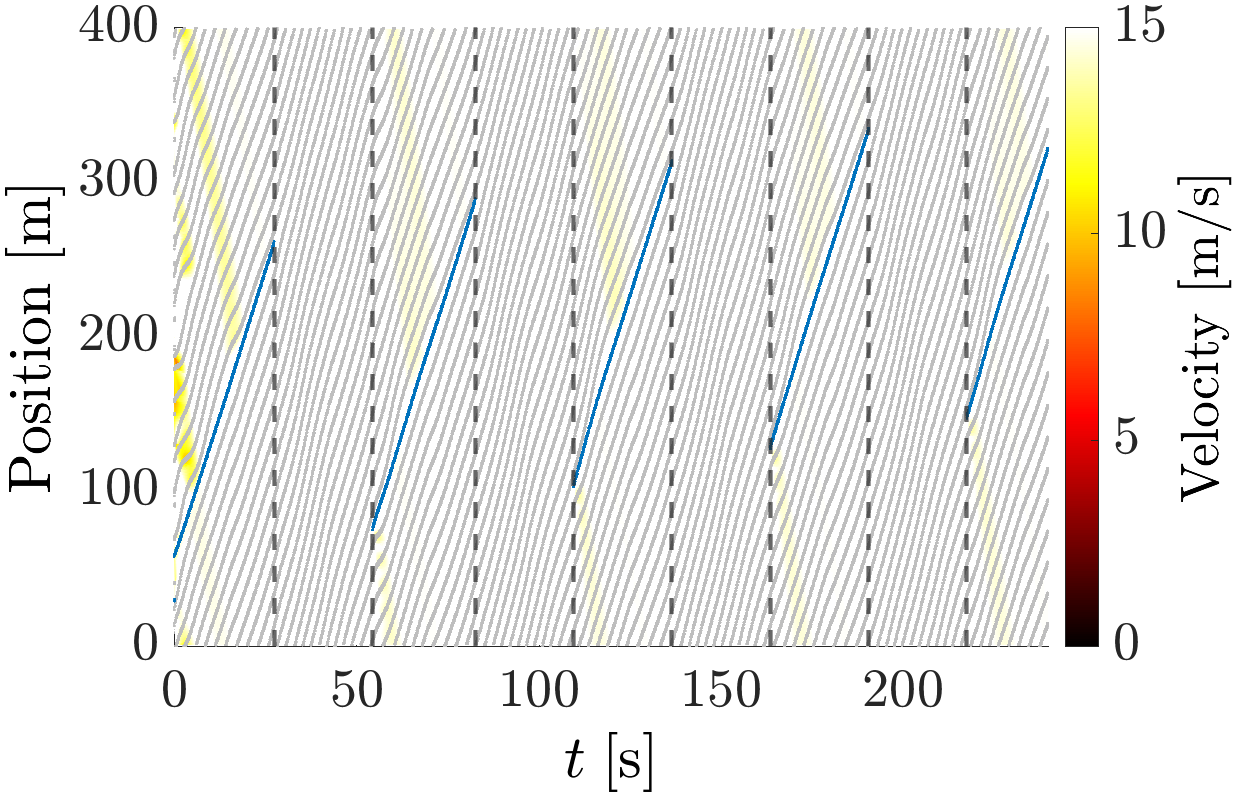}
        \includegraphics[width=0.48\textwidth]{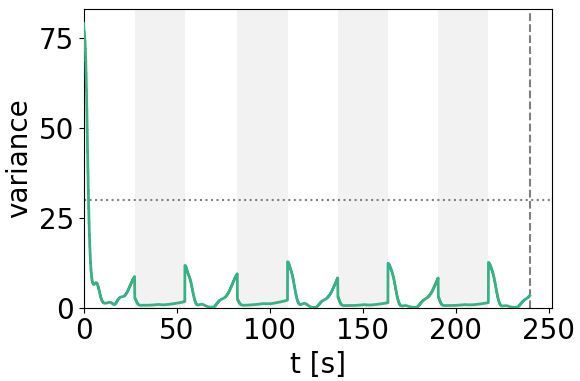}
        \caption{$T = 30s$, Integrated Anticipatory and Traffic-Aware Control}
    \end{subfigure} 
    \caption{\textbf{Simulated trajectories of the anticipatory and traffic-aware controllers.} The time-space diagram and the simulated variance for the anticipatory controller (left two columns) and the integrated anticipatory and traffic-aware controller (right two columns), with lane-switch duration $T=8.5s$ (top) and $T=30s$ (bottom). The horizontal dotted lines are the simulation variances for the corresponding fixed-duration controller with the same $Ts$. The augmented controllers effectively reduce the systems' variances from the fixed-duration controller.}
    \vspace*{-0.2cm}
    \label{fig:experiment_pe_ta}
\end{figure*}

\textbf{Theory and Simulation Variance Comparison.} 
Fig~\ref{fig:experiment_fd} displays the simulated and theoretically estimated  variance upper bound for the fixed-duration controller. Notably, we observe a strong correspondence between the theory upper bound and simulation variance across various control duration $Ts$. The theory accurately captures the variance trends, both qualitatively and quantitatively, for both the continuous dynamics and the discrete jumps. This alignment indicates the efficacy of the proposed theoretical stability analysis for the multi-lane mixed-autonomy system. 

\vspace*{-0.3cm}
\subsection{Traffic Phenomena Analysis and Controller Design}
This section evaluates the effect of the different AV controllers in Sec.~\ref{sec:controller_fixed_duration} and~\ref{sec:controller}, aiming to provide insights into AV lane-switch strategies to stabilize the two-lane system.

\vspace{0.2cm}
\subsubsection{Fixed-Duration Control for Emergent Traffic Phenomena}
\label{sec:experiment_fd}
Both simulated and theoretically estimated variance plots in Fig.~\ref{fig:experiment_fd} reveal emergent traffic phenomena under various fixed-control duration $Ts$; we further visualize the traffic state at an AV lane switch in Fig.~\ref{fig:experiment_fd_illustration}, after the trajectory reaches the periodic orbit. These figures are consistent with the four scenarios identified by Observation~\ref{observation:fd} from the proposed stability theory.

Interestingly, we observe different low-variance orbits for Phantom Car and Sufficient Control Duration. Let $z^*_{n}$ denote the equilibrium state for the controlled period (the initial state in Fig.~\ref{fig:control_uncontrol_singlelane} (a)) and $z^*_{n-1}$ denote the equilibrium state for the uncontrolled period (the initial state in Fig.~\ref{fig:control_uncontrol_singlelane} (f)). In the Phantom Car scenario ($T=3s$), the system converges to $z^*_{n-1}$ during both the continuous controlled and uncontrolled periods, whereas in the Sufficient Control Duration scenario ($T=30s$), the system converges to $z^*_n$ and $z^*_{n-1}$ for the respective controlled and uncontrolled period. 

Theoretically speaking, in the Phantom Car scenario ($T\approx 0$), the high-frequency lane-switching primarily produces a first-order effect between the AV and its two neighboring HVs, as the limited duration prevents propagation to higher-order effects which involve other vehicles. The first-order local effect consists of a competing force between increasing and decreasing the neighboring HVs' headways, as illustrated in Fig.~\ref{fig:control_uncontrol_singlelane} (e), when AV starts controlling the lane, and Fig.~\ref{fig:control_uncontrol_singlelane} (j), after the AV exits the lane. Due to the uncontrolled dynamics' faster initial response to close AV's headway gap, the system converges towards $z^*_{n-1}$, where HVs maintain nearly equal headways (ignoring the AV), with the presence of the AV ensuring the stability of the system. 

In the Sufficient Control Duration scenario ($T=30s$), in contrast, the controlled dynamics have ample time to guide the system towards $z^*_n$. The AV exits the lane afterwards, and the uncontrolled dynamics shifts the system towards $z^*_{n-1}$. The AV subsequently re-enters lane before the uncontrolled dynamics become unstable. The convergence to different equilibriums $z^*_n$ and $z^*_{n-1}$ for the controlled and uncontrolled periods hence leads to a higher variance compared to the Phantom Car scenario. Nonetheless, this scenario is easier to implement and less disruptive to the traffic (e.g. the AV isn't required keep all HVs behind at all times). 

\vspace{0.15cm}
\subsubsection{Anticipatory Control}
\label{sec:experiment_pe}

Sec.~\ref{sec:controller_pe} proposes a theory-informed anticipatory control strategy to reduce the variance when AV exits the lane. Fig.~\ref{fig:experiment_pe_ta} (left two columns) illustrates the relevant time-space diagram and simulation variance plot for $T = 8.5s$ and $T = 30s$, setting $p_{ex, s} = 0.5$, when we impose the anticipatory control. We observe a significant reduction in the variance increase when AV exits a lane, due to the effective reduction of the AV's headway that reduces the headway gap between the neighboring HVs when the AV exits the lane.

\vspace{0.15cm}
\subsubsection{Integrated Anticipatory and Traffic-Aware Control} 
\label{sec:experiment_ta}

Finally, we combine the traffic-aware control with the anticipatory control, and visualize the corresponding plots in Fig.~\ref{fig:experiment_pe_ta} (right two columns). 
Sec.~\ref{sec:controller_ta} proposes an augmented AV lane-switch strategy based on insights from the theoretical analysis (Prop.~\ref{prop:lanechange}), with the aim to reduce the variance when AV enters a lane. Fig.~\ref{fig:experiment_pe_ta} displays the time-space diagram and the simulation variance curve for the proposed strategy under various $Ts$, setting $p_{en, s} = 0.2, p_{en, v} = 0.8$. We further set $\Delta T_{ta} = 0.1T$ and mandate a lane switch at $1.1T$ if the imposed criteria are not met within the $[0.9T, 1.1T]$ period.

Compared with the anticipatory control, the combined strategy effectively \slirebuttal{reduces} the variance increase when the AV both exits and enters, further validating the practicality of the proposed theory in designing integrated AV control strategy for multi-lane mixed-autonomy systems.
\vspace{-0.1cm}
\section{CONCLUSION}
\label{sec:conclusion}
This work presents a theoretical framework to analyze the stability of the multi-lane mixed-autonomy system. Casting into the hybrid system framework, the proposed variance-based analysis combines theoretical bounds for the continuous dynamics and the discrete jumps (induced by AV lane-switch). The analysis provides principled understanding to emergent traffic phenomena such as traffic breaks and less-intrusive regulation, and can inform AV controller design for anticipatory and traffic-aware controllers to mitigate traffic congestion. 

The analysis in this work opens up several interesting avenues for future research: the proposed theoretical analysis \slirebuttal{can be enhanced} by (1) finding tight\slirebuttal{er} theoretical \slirebuttal{bounds of} the variance changes for both the continuous dynamics and discrete jumps (see Sec.~\ref{sec:analysis})\slirebuttal{, potentially extending from the single-lane reachability analysis~\cite{liu2023reachability} to consider different shapes of the ellipsoidal sets}, and (2) considering more general traffic scenarios including more lanes and varying number of vehicles per-lane. We would like to \slirebuttal{further} expand beyond the proposed analysis by (1) incorporating lane changes from human drivers, \slirebuttal{potentially modeling as exogenous and endogenous random processes and extending the Hybrid System formulation to a combinatorial number of modes,} and (2) applying the theoretical analysis to design more complicated AV control strategies, leveraging advanced control techniques and potentially reinforcement learning. \sli{We plan to further conduct field studies to assess the effectiveness of the theoretically-informed AV control strategies.} We hope this work can seed more future work on providing principled theoretical justification of emergent traffic behaviors in autonomous vehicle control, and for future development of efficient, safe, and sustainable multi-lane mixed-autonomy traffic systems.

\vspace{-0.15cm}
\appendix[]
\subsection{Proof of Lemma~\ref{lemma:varnorm}}
\label{appendix:varnorm}
\begin{proof}
Due to the translation invariance of the variance operator, we can subtract the equilibrium headway and velocity $s^{*}_{n_l}, v^{*}_{n_l}$ from the actual headways and velocities, and obtain
\begin{equation}
    \begin{aligned}
        & var^l(t) = variance(\{s_i^l(t) - s^{*}_{n_l}\}_i) + variance(\{v_i^l(t)  - v^{*}_{n_l}\}_i) \\
        = \; & variance(\{\tilde{s}_i^l(t)\}_i) + variance(\{\tilde{v}_i^l(t)\}_i) \\
        = \; &\Big(\frac{1}{n_l}\|\tilde{s}^l(t)\|_2^2 - \frac{1}{n_l^2}\big(\sum\limits_{i=1}^{n_l}\tilde{s}^l_i(t)\big)^2\Big) + \Big(\frac{1}{n_l}\|\tilde{v}^l(t)\|_2^2 - \frac{1}{n_l^2}\big(\sum\limits_{i=1}^{n_l}\tilde{v}^l_i(t)\big)^2\Big)
    \end{aligned}
\end{equation}
where the error state $x^l(t) = [\tilde{s}_1^l(t), \tilde{v}_1^l(t), ..., \tilde{s}_{n_l}^l(t), \tilde{v}_{n_l}^l(t)]$ with the headway errors $\tilde{s}^l(t) = [\tilde{s}_1^l(t), ..., \tilde{s}_{n_l}^l(t)]$ and velocity errors $\tilde{v}
^l(t) = [\tilde{v}_1^l(t), ..., \tilde{v}_{n_l}^l(t)]$.

\vspace{0.2cm}
\noindent \textbf{Upper Bound.} Ignoring the negative terms, we arrive at the following upper bound with $c_2 = \frac{1}{n_l} > 0$: 
\begin{equation}
\begin{aligned}
    var^l(t) \leq \frac{1}{n_l}\|\tilde{s}^l(t)\|_2^2 + \frac{1}{n_l}\|\tilde{v}^l(t)\|_2^2 =  c_2 \|x^l(t)\|_2^2 
    \label{eq:proof_varnorm_ub}
\end{aligned}
\end{equation}
\noindent \textbf{Lower Bound.} By the property of the ring road and the OVM dynamics (where $s^*_{n_l} = C / n_l$), we have $\sum\limits_{i=1}^{n_l} \tilde{s}_i^l(t) = \sum\limits_{i=1}^{n_l} s_i^l(t) - n_l \cdot s^*_{n_l} = C - C = 0$, from which we obtain the exact equality (and hence the same lower bound) for the headway $variance(\{\tilde{s}_i^l(t)\}_i) = \frac{1}{n_l}\|\tilde{s}(t)\|_2^2$.
While the sum of the headways stays constant, this does not always hold for the velocity. \slirebuttal{Under Assump.~\ref{assump:varnorm},} $\{\tilde{v}^l_i(t)\}_i = \{v^l_i(t) - v^{l*}_{n_l}\}_i$ consists of negative and \slirebuttal{nonnegative} terms which cancel out each other, making $\frac{1}{n_l^2}\big(\sum\limits_{i=1}^{n_l} \tilde{v}^l(t)\big)^2$ small and $variance(\{\tilde{v}_i^l(t)\}_i)$ large. 
\slirebuttal{Specifically,}
\vspace{-0.1cm}
\begin{equation*}
   \begin{aligned}
        & \frac{1}{n_l^2}\Big(\sum\limits_{i=1}^{n_l} \tilde{v}_i^l(t)\Big)^2 = \frac{1}{n_l^2}\Big(\sum\limits_{i=1}^{n_l^{-}} \tilde{v}_i^{l-}(t) + \sum\limits_{i=1}^{n_l^{+}}\tilde{v}_i^{l+}(t)\Big)^2  \\
        \leq \; &\frac{1}{n_l^2} \Big(\sum\limits_{i=1}^{n_l^{-}} \tilde{v}_i^{l-}(t)\Big)^2 + \frac{1}{n_l^2} \Big(\sum\limits_{i=1}^{n_l^{+}} \tilde{v}_i^{l+}(t)\Big)^2 = \frac{1}{n_l^2} \|\tilde{v}^{l-}(t)\|_1^2 + \frac{1}{n_l^2} \|\tilde{v}^{l+}(t)\|_1^2 \\
        \leq \; & \frac{n_l^{-}}{n_l^2} \|\tilde{v}_i^{l-}(t)\|_2^2 + \frac{n_l^{+}}{n_l^2}  \|\tilde{v}_i^{l+}(t)\|_2^2 \leq \frac{\max\{n_l^{-}, n_l^{+}\}}{n_l^2}\|\tilde{v}_i^l(t)\|_2^2
    \end{aligned}
\end{equation*}
where the first line separates $\{\tilde{v}^l_i(t)\}_{i=1}^{n_l}$ into \slirebuttal{negative} and \slirebuttal{nonnegative} terms ($\{\tilde{v}^{l-}_i(t)\}_{i=1}^{n_l-}, \{\tilde{v}^{l+}_i(t)\}_{i=1}^{n_l+}$); the second line splits the two sets of terms by ignoring the negative cross terms, ~\slirebuttal{and uses the} definition of the $l_1$-norm as terms in each set have the same sign; the last line upper bounds the $l_1$ norm by $n_2$ norm through the inequality $\|\delta\|_1^2 \leq m \|\delta\|_2^2$ where $\delta \in \mathbb{R}^m$, \slirebuttal{and also} aggregates and upper bounds the two $l_2$-norm terms into one single term. 
\slirebuttal{From Assump.~\ref{assump:varnorm}, we have either the velocity component all at equilibrium, i.e. $\tilde{v}_i^l(t) = 0 \; \forall i$, or $1 < n^{-}_l, n^{+}_l < n_l$ and $n_l - \max\{n^{-}_l, n^{+}_l\} > 0$, which leads to }
\vspace*{-0.1cm}
\begin{equation}
\begin{aligned}
    & variance(\{\tilde{v}_i^l(t)\}_i)  = \frac{1}{n_l}\|\tilde{v}^l(t)\|_2^2 - \frac{1}{n_l^2}\Big(\sum\limits_{i=1}^{n_l} \tilde{v}_i^l(t)\Big)^2 \\
   \geq  \; &  \frac{\left(n_l - \max\{n^{+}_l, n^{-}_l\}\right)}{n_l^2}\|\tilde{v}_i^l(t)\|_2^2 = c_1 \|\tilde{v}^l(t)\|_2^2 
\end{aligned}
\label{eq:proof_varnorm_lb}
\end{equation}
for some $c_1 > 0$.
\end{proof}
\vspace{-0.4cm}
\subsection{Proof of Proposition~\ref{prop:cont_var_bd1}}
\label{appendix:cont_var_bd1}
\begin{proof}
\slirebuttal{\textit{Controlled period.}} Given the Lyapunov function $V(x^l_c(t)) = x^l_c(t)^\intercal P x^l_c(t)$ with $\lambda_{min}(P)\|x^l_c(t)\|_2^2 \leq x^l_c(t)^\intercal P x^l_c(t) \leq \lambda_{max}(P) \|x^l_c(t)\|_2^2$, we have
\begin{equation}
\begin{aligned}
 \dot{V}(t) & = \frac{d}{dt}x^l_c(t)^\intercal P x^l_c(t) = -{x^l_c}^\intercal(t) Q x^l_c(t) \\
& \leq -\lambda_{min}(Q)\|x^l_c(t)\|_2^2 \leq -\frac{\lambda_{min}(Q)}{\lambda_{max}(P)}{x^l_c}^\intercal(t) P x^l_c(t)
\end{aligned}
\end{equation}
where $P, Q > 0$. By Gronwall's inequality, we obtain
\begin{equation*}
\begin{aligned}
x^l_c(t)^\intercal P x^l_c(t) & \leq \exp\left(-\frac{\lambda_{min}(Q)}{\lambda_{max}(P)}t\right)x^l_c(0)^\intercal Px(0), \text{ hence }
\end{aligned}
\end{equation*}
\begin{equation}
    \begin{aligned}
    \|x^l_c(t)\|_2^2 & \leq \frac{\lambda_{max}(P_1)}{\lambda_{min}(P_1)} \exp\left(-\frac{\lambda_{min}(Q_1)}{\lambda_{max}(P_1)}t\right) \|x^l_c(0)\|_2^2\\
    & \leq \frac{1}{c_1}\frac{\lambda_{max}(P_1)}{\lambda_{min}(P_1)} \exp\left(-\frac{\lambda_{min}(Q_1)}{\lambda_{max}(P_1)}t\right) var_0,
    \end{aligned}
\end{equation}
for some $c_1 > 0$, from Lemma~\ref{lemma:varnorm}. Applying the Lemma again, we obtain $var^l(t) \leq c_2\|x^l_{c}(t)\|_2^2 \leq \alpha_1 \exp(-\alpha_2 t) var_0$ with $\alpha_1 = \big(c_2\lambda_{max}(P_1)\big) / \big(c_1 \lambda_{min}(P_1)\big), \alpha_2 = \lambda_{min}(Q_1) / \lambda_{max}(P_1)$.
\slirebuttal{\textit{Uncontrolled period.}} We can invoke a similar analysis $\dot{x}^l_u(t) = A_u x^l_u(t)$ by considering $U(x) = x^l_u(t)^\intercal x^l_u(t) = \|x^l_u(t)\|_2^2$. We have
$\dot{U}(t) = \frac{d}{dt}x^l_u(t)^\intercal x^l_u(t) = x^l_u(t)^\intercal W x^l_u(t)$, where $W = A_u + A_u^\intercal$ with $\lambda_{max}(W) > 0$, as the uncontrolled OVM system is unstable. Therefore, $\frac{d}{dt} x^l_u(t)^\intercal x^l_u(t) = x^l_u(t)^\intercal W x^l_u(t) \leq \lambda_{max}(W)\|x^l_u(t)\|^2$.
Applying Gronwall's inequality, we obtain
\begin{equation}
\begin{aligned}
\|x^l_u(t)\|_2^2 = x^l_u(t)^\intercal x^l_u(t) & \leq \exp\left(\lambda_{max}(W) t\right)x^l_u(0)^\intercal x^l_u(0)\\
& = \exp\left(\lambda_{max}(W) t\right) \|x^l_u(0)\|_2^2
\end{aligned}
\end{equation}
Apply Lemma~\ref{lemma:varnorm}, we obtain $ var^l(t) \leq c_2\|x^l_u(t)\|_2^2 \leq \beta_1 \exp(\beta_2 t) var_0$ with $\beta_1 = c_2, \beta_2 = \lambda_{max}(Q)$ for all $t\geq 0$.
\end{proof}

\vspace*{-0.4cm}
\subsection{Proof of Theorem~\ref{thm:cont_init}}
\label{appendix:cont_init}
\begin{proof}
The variance operator on a vector $\delta = \{\delta_i\}_{i=1}^m \in \mathbb{R}^d$ can be written as a quadratic function 
$variance(\delta) = \delta^\intercal P_{var} \delta$ where $P_{var} = \frac{1}{m}\mathit{I}_{m\times m} - \frac{1}{m^2}\mathit{1}_{m\times 1}\mathit{1}_{1\times m} \geq 0$ with $\lambda_{max}(P_{var}) \leq \frac{1}{m}$. Consier $\delta = \{\delta_i\}_{i=1}^m, \bar{\delta} = \{\bar{\delta}_i\}_{i=1}^m$, we have 
\begin{equation}
    \begin{aligned}
    & \; |variance(\delta) - variance(\bar{\delta})| \\
    = & \; |\delta^\intercal P_{var} \delta - \bar{\delta}^\intercal P_{var}\bar{\delta}| 
    = |(\delta + \bar{\delta})^\intercal P_{var} (\delta - \bar{\delta})|\\
    \leq & \;\lambda_{max}(P_{var}) |(\delta + \bar{\delta})^\intercal(\delta - \bar{\delta})| \\
    \leq & \; \frac{1}{m} |(\delta + \bar{\delta})^\intercal(\delta - \bar{\delta})| 
    \leq \frac{1}{m} \|\delta + \bar{\delta}\|_2 \|\delta - \bar{\delta} \|_2
    \end{aligned}
\end{equation}
Denote the state $z(t) = [s_{z, 1}^l(t), v_{z, 1}^l(t), ..., s_{z, n_l}^l(t), v_{z, n_l}^l(t)]$ with the headways $s_{z}(t) = \{s_{z, i}^l(t)\}_{i=1}^{n_l}$ and velocities $v_{z}(t) = \{v_{z, i}^l(t)\}_{i=1}^{n_l}$, and the corresponding variance $var^l_{z}(t) = variance(s_{z}(t)) + variance(v_{z}(t))$ (and similar notation for $\bar{z}$), we can follow a similar derivation as above and get
\begin{equation}
    \begin{aligned}
        & \left|var^l_{z}(t) - var^l_{\bar{z}}(t)\right| = \big|(variance(s_z(t)) - variance(s_{\bar{z}}(t))) + \\
        & \hspace*{2.95cm} (variance(v_z(t)) - variance(v_{\bar{z}}(t)))\big|\\
        \leq\; &  \frac{1}{n_l} |(s_{z}(t) + s_{\bar{z}}(t))^\intercal(s_{z}(t) - s_{\bar{z}}(t)) +  (v_{z}(t) + v_{\bar{z}}(t))^\intercal(v_{z}(t) - v_{\bar{z}}(t))| \\
        =\; & \frac{1}{n_l} |(z(t) + \bar{z}(t))^\intercal (z(t) - \bar{z}(t))| \leq  \frac{1}{n_l} \|z(t) + \bar{z}(t)\|_2 \|z(t) - \bar{z}(t)\|_2
    \end{aligned}
    \label{eq:appendix_var}
\end{equation}
Following the same derivation as for Eq.~\eqref{eq:cont_init_state} (with the minus sign replaced with a plus sign), we arrive at
\begin{equation}
\|z(t) + \bar{z}(t)\|_2\leq \|z_0 + \bar{z}_0\|_2\exp[L(t-t_0)]
\end{equation}
Combining with Eq.~\eqref{eq:appendix_var} and~\eqref{eq:cont_init_state}, we hence obtain
\begin{equation}
    \begin{aligned}
        \left|var^l_{z}(t) - var^l_{\bar{z}}(t)\right| \leq \frac{1}{n_l}\|z_0 + \bar{z}_0\|_2 \|z_0 - \bar{z}_0\|_2 \exp[2L(t-t_0)]
    \end{aligned}
\end{equation}
Moving $var^l_{z_0}(t)$ to the right hand side, we arrive at the final upper bound.
\end{proof}
\vspace*{-0.6cm}
\subsection{Proof of Proposition~\ref{prop:lanechange}}
\label{appendix:lanechange}
\begin{proof}
Similar to $\Delta^1_{c \rightarrow u,s}$ derived in the main paper, we have the other cases
\begin{equation*}
    \begin{aligned}
        var^R_{u, s} & = \frac{1}{n_R}\slirebuttal{\Big(}(a^R + b^R)^2 + \sum_{\substack{i: none\\lag\;HV}} (s^R_i)^2\slirebuttal{\Big)} - \frac{1}{n_R^2} C^2\\
        \tilde{var}^R_{c, s} & = \frac{1}{n_R + 1}\slirebuttal{\Big(}(a^R)^2 + (b^R)^2 + \sum_{\substack{i: none \\ lag\;HV}} (s^R_i)^2\slirebuttal{\Big)} - \frac{1}{(n_R - 1)^2}C^2 
        \end{aligned}
\end{equation*}
\begin{equation*}
\begin{aligned}
        var^L_{c, v} &= \frac{1}{n_L}\slirebuttal{\Big(}v_n^2 + \sum_{i: HV} (v_i^L)^2\slirebuttal{\Big)} - \frac{1}{n_L^2}\slirebuttal{\Big(}v_n + \sum\limits_{i: HV} v_i^L\slirebuttal{\Big)}^2\\
        \tilde{var}^L_{u, v} &= \frac{1}{n_L - 1}\slirebuttal{\Big(}\sum\limits_{i: HV} (v_i^L)^2\slirebuttal{\Big)} - \frac{1}{(n_L - 1)^2}\slirebuttal{\Big(}\sum\limits_{i: HV} v_i^L\slirebuttal{\Big)}^2\\
        var^R_{u, v} &= \frac{1}{n_R}\slirebuttal{\Big(}\sum\limits_{i: HV} (v^R_i)^2\slirebuttal{\Big)} - \frac{1}{n_R^2}\slirebuttal{\Big(}\sum\limits_{i: HV} v^R_i\slirebuttal{\Big)})^2\\
        \tilde{var}^R_{c, v} &= \frac{1}{n_R + 1}\slirebuttal{\Big(}\sum\limits_{i: HV} (v^R_i)^2 + v_n^2\slirebuttal{\Big)} - \frac{1}{(n_R + 1)^2}\slirebuttal{\Big(}\sum\limits_{i: HV} v^R_i + v_n\slirebuttal{\Big)}^2\\
    \end{aligned}
\end{equation*}
\slirebuttal{We hence obtain the closed-form expressions for $\Delta_{u \rightarrow c,s}(s^R) = \tilde{var}^R_{c, s} - var^R_{u, s}, \Delta_{c \rightarrow v, u}(v^L) = \tilde{var}^L_{u, v} - var^L_{c, v}$ and $\Delta_{u \rightarrow v, c}(v^R) = \tilde{var}^R_{c, v} - var^R_{u, v}$ as in Table~\ref{tab:analysis_jump} by standard algebra.}
\end{proof}


\vspace*{-0.5cm}
\bibliography{references}

\begin{thebibliography}{10}
\providecommand{\url}[1]{#1}
\csname url@samestyle\endcsname
\providecommand{\newblock}{\relax}
\providecommand{\bibinfo}[2]{#2}
\providecommand{\BIBentrySTDinterwordspacing}{\spaceskip=0pt\relax}
\providecommand{\BIBentryALTinterwordstretchfactor}{4}
\providecommand{\BIBentryALTinterwordspacing}{\spaceskip=\fontdimen2\font plus
\BIBentryALTinterwordstretchfactor\fontdimen3\font minus \fontdimen4\font\relax}
\providecommand{\BIBforeignlanguage}[2]{{%
\expandafter\ifx\csname l@#1\endcsname\relax
\typeout{** WARNING: IEEEtran.bst: No hyphenation pattern has been}%
\typeout{** loaded for the language `#1'. Using the pattern for}%
\typeout{** the default language instead.}%
\else
\language=\csname l@#1\endcsname
\fi
#2}}
\providecommand{\BIBdecl}{\relax}
\BIBdecl

\bibitem{national2019traffic}
N.~H. T.~S. Administration \emph{et~al.}, ``Traffic safety facts 2017: A compilation of motor vehicle crash data,'' \emph{DOT HS}, vol. 812806, 2019.

\bibitem{epaghgemission}
\BIBentryALTinterwordspacing
``Fast facts on transportation greenhouse gas emissions.'' [Online]. Available: \url{https://www.epa.gov/greenvehicles/fast-facts-transportation-greenhouse-gas-emissions}
\BIBentrySTDinterwordspacing

\bibitem{winston2015transportation}
C.~Winston, ``Transportation and the united states economy: Implications for governance,'' \emph{Brookings Institution, Washington}, 2015.

\bibitem{fhwa2019}
\BIBentryALTinterwordspacing
``Safely implementing rolling roadblocks for short-term highway construction, maintenance, and utility work zones.'' [Online]. Available: \url{https://ops.fhwa.dot.gov/publications/fhwahop19031/index.htm}
\BIBentrySTDinterwordspacing

\bibitem{californiahighway}
\BIBentryALTinterwordspacing
``California highway patrol terminology.'' [Online]. Available: \url{http://americanindian.net/traffica.html}
\BIBentrySTDinterwordspacing

\bibitem{saha2021application}
P.~Saha and A.~Kobryn, ``Application of rolling slowdown versus roadblock on high profile roadways: A case study,'' in \emph{International Conference on Transportation and Development}, 2021, pp. 204--217.

\bibitem{nafakh2022safety}
A.~J. Nafakh, F.~V. Davila, Y.~Zhang, J.~D. Fricker, and D.~M. Abraham, ``Safety and mobility analysis of rolling slowdown for work zones: Comparison with full closure,'' 2022.

\bibitem{stern2018dissipation}
R.~E. Stern, S.~Cui, M.~L. Delle~Monache, R.~Bhadani, M.~Bunting \emph{et~al.}, ``Dissipation of stop-and-go waves via control of autonomous vehicles: Field experiments,'' \emph{Transportation Research Part C: Emerging Technologies}, vol.~89, pp. 205--221, 2018.

\bibitem{wu2021flow}
C.~Wu, A.~R. Kreidieh, K.~Parvate, E.~Vinitsky, and A.~M. Bayen, ``Flow: A modular learning framework for mixed autonomy traffic,'' \emph{IEEE Transactions on Robotics}, 2021.

\bibitem{yan2022unified}
Z.~Yan, A.~R. Kreidieh, E.~Vinitsky, A.~M. Bayen, and C.~Wu, ``Unified automatic control of vehicular systems with reinforcement learning,'' \emph{IEEE Transactions on Automation Science and Engineering}, 2022.

\bibitem{cui2017stabilizing}
S.~Cui, B.~Seibold, R.~Stern, and D.~B. Work, ``Stabilizing traffic flow via a single autonomous vehicle: Possibilities and limitations,'' in \emph{2017 IEEE Intelligent Vehicles Symposium (IV)}.\hskip 1em plus 0.5em minus 0.4em\relax IEEE, 2017, pp. 1336--1341.

\bibitem{zheng2020smoothing}
Y.~Zheng, J.~Wang, and K.~Li, ``Smoothing traffic flow via control of autonomous vehicles,'' \emph{IEEE Internet of Things Journal}, vol.~7, no.~5, pp. 3882--3896, 2020.

\bibitem{zheng2015stability}
Y.~Zheng, S.~E. Li, J.~Wang, D.~Cao, and K.~Li, ``Stability and scalability of homogeneous vehicular platoon: Study on the influence of information flow topologies,'' \emph{IEEE Transactions on intelligent transportation systems}, vol.~17, no.~1, pp. 14--26, 2015.

\bibitem{zhu2018analysis}
W.-X. Zhu and H.~M. Zhang, ``Analysis of mixed traffic flow with human-driving and autonomous cars based on car-following model,'' \emph{Physica A: Statistical Mechanics and its Applications}, vol. 496, 2018.

\bibitem{wang2020controllability}
J.~Wang, Y.~Zheng, Q.~Xu, J.~Wang, and K.~Li, ``Controllability analysis and optimal control of mixed traffic flow with human-driven and autonomous vehicles,'' \emph{IEEE Transactions on Intelligent Transportation Systems}, vol.~22, no.~12, pp. 7445--7459, 2020.

\bibitem{mousavi2022synthesis}
S.~S. Mousavi, S.~Bahrami, and A.~Kouvelas, ``Synthesis of output-feedback controllers for mixed traffic systems in presence of disturbances and uncertainties,'' \emph{IEEE Transactions on Intelligent Transportation Systems}, 2022.

\bibitem{swaroop1994string}
D.~Swaroop, \emph{String stability of interconnected systems: An application to platooning in automated highway systems}.\hskip 1em plus 0.5em minus 0.4em\relax University of California, Berkeley, 1994.

\bibitem{bose2003analysis}
A.~Bose and P.~A. Ioannou, ``Analysis of traffic flow with mixed manual and semiautomated vehicles,'' \emph{IEEE Transactions on Intelligent Transportation Systems}, vol.~4, no.~4, pp. 173--188, 2003.

\bibitem{rogge2008vehicle}
J.~A. Rogge and D.~Aeyels, ``Vehicle platoons through ring coupling,'' \emph{IEEE Transactions on Automatic Control}, vol.~53, no.~6, 2008.

\bibitem{wu2018stabilizing}
C.~Wu, A.~M. Bayen, and A.~Mehta, ``Stabilizing traffic with autonomous vehicles,'' in \emph{2018 IEEE International Conference on Robotics and Automation (ICRA)}.\hskip 1em plus 0.5em minus 0.4em\relax IEEE, 2018, pp. 6012--6018.

\bibitem{giammarino2020traffic}
V.~Giammarino, S.~Baldi, P.~Frasca, and M.~L. Delle~Monache, ``Traffic flow on a ring with a single autonomous vehicle: An interconnected stability perspective,'' \emph{IEEE Transactions on Intelligent Transportation Systems}, vol.~22, no.~8, pp. 4998--5008, 2020.

\bibitem{liu2022structural}
D.~Liu, B.~Besselink, S.~Baldi, W.~Yu, and H.~L. Trentelman, ``On structural and safety properties of head-to-tail string stability in mixed platoons,'' \emph{IEEE Transactions on Intelligent Transportation Systems}, 2022.

\bibitem{jayawardana2022learning}
V.~Jayawardana and C.~Wu, ``Learning eco-driving strategies at signalized intersections,'' \emph{arXiv preprint arXiv:2204.12561}, 2022.

\bibitem{bando1995dynamical}
M.~Bando, K.~Hasebe, A.~Nakayama, A.~Shibata, and Y.~Sugiyama, ``Dynamical model of traffic congestion and numerical simulation,'' \emph{Physical review E}, vol.~51, no.~2, p. 1035, 1995.

\bibitem{jin2016optimal}
I.~G. Jin and G.~Orosz, ``Optimal control of connected vehicle systems with communication delay and driver reaction time,'' \emph{IEEE Transactions on Intelligent Transportation Systems}, vol.~18, no.~8, 2016.

\bibitem{liu2023reachability}
D.~Liu, B.~Besselink, S.~Baldi, W.~Yu, and H.~L. Trentelman, ``A reachability approach to disturbance and safety propagation in mixed platoons,'' \emph{IEEE Transactions on Automatic Control}, 2023.

\bibitem{gisolo2022nonlinear}
C.~M. Gisolo, M.~L. Delle~Monache, F.~Ferrante, and P.~Frasca, ``Nonlinear analysis of stability and safety of optimal velocity model vehicle groups on ring roads,'' \emph{IEEE Transactions on Intelligent Transportation Systems}, vol.~23, no.~11, pp. 20\,628--20\,635, 2022.

\bibitem{li2023stabilization}
S.~Li, R.~Dong, and C.~Wu, ``Integrated analysis of coarse-grained guidance for traffic flow stability,'' \emph{IEEE Transactions on Control of Network Systems}, 2023.

\bibitem{liberzon2003switching}
D.~Liberzon, \emph{Switching in systems and control}.\hskip 1em plus 0.5em minus 0.4em\relax Springer, 2003, vol. 190.

\bibitem{shorten2007stability}
R.~Shorten, F.~Wirth, O.~Mason, K.~Wulff, and C.~King, ``Stability criteria for switched and hybrid systems,'' \emph{SIAM review}, vol.~49, no.~4, 2007.

\bibitem{zhu2015optimal}
F.~Zhu and P.~J. Antsaklis, ``Optimal control of hybrid switched systems: A brief survey,'' \emph{Discrete Event Dynamic Systems}, vol.~25, 2015.

\bibitem{frazzoli2000robust}
E.~Frazzoli, M.~A. Dahleh, and E.~Feron, ``Robust hybrid control for autonomous vehicle motion planning,'' in \emph{Proceedings of the 39th IEEE Conference on Decision and Control (Cat. No. 00CH37187)}, vol.~1.\hskip 1em plus 0.5em minus 0.4em\relax IEEE, 2000, pp. 821--826.

\bibitem{gans2007stable}
N.~R. Gans and S.~A. Hutchinson, ``Stable visual servoing through hybrid switched-system control,'' \emph{IEEE Transactions on Robotics}, vol.~23, no.~3, pp. 530--540, 2007.

\bibitem{schill2018robust}
M.~M. Schill and M.~Buss, ``Robust ballistic catching: A hybrid system stabilization problem,'' \emph{IEEE Transactions on Robotics}, vol.~34, no.~6, pp. 1502--1517, 2018.

\bibitem{shorten1998stability}
R.~Shorten and K.~Narendra, ``On the stability and existence of common lyapunov functions for stable linear switching systems,'' in \emph{Proceedings of the 37th IEEE Conference on Decision and Control (Cat. No. 98CH36171)}, vol.~4.\hskip 1em plus 0.5em minus 0.4em\relax IEEE, 1998, pp. 3723--3724.

\bibitem{liberzon2004common}
D.~Liberzon and R.~Tempo, ``Common lyapunov functions and gradient algorithms,'' \emph{IEEE Transactions on Automatic Control}, vol.~49, no.~6, pp. 990--994, 2004.

\bibitem{williams2013constrained}
R.~K. Williams and G.~S. Sukhatme, ``Constrained interaction and coordination in proximity-limited multiagent systems,'' \emph{IEEE Transactions on Robotics}, vol.~29, no.~4, pp. 930--944, 2013.

\bibitem{branicky1998multiple}
M.~S. Branicky, ``Multiple lyapunov functions and other analysis tools for switched and hybrid systems,'' \emph{IEEE Transactions on automatic control}, vol.~43, no.~4, pp. 475--482, 1998.

\bibitem{long2017multiple}
L.~Long, ``Multiple lyapunov functions-based small-gain theorems for switched interconnected nonlinear systems,'' \emph{IEEE Transactions on Automatic Control}, vol.~62, no.~8, pp. 3943--3958, 2017.

\bibitem{yuan2014hybrid}
C.~Yuan and F.~Wu, ``Hybrid control for switched linear systems with average dwell time,'' \emph{IEEE Transactions on Automatic Control}, vol.~60, no.~1, pp. 240--245, 2014.

\bibitem{zhao2014switching}
X.~Zhao, S.~Yin, H.~Li, and B.~Niu, ``Switching stabilization for a class of slowly switched systems,'' \emph{IEEE Transactions on Automatic Control}, vol.~60, no.~1, pp. 221--226, 2014.

\bibitem{moridpour2010lane}
S.~Moridpour, M.~Sarvi, and G.~Rose, ``Lane changing models: a critical review,'' \emph{Transportation letters}, vol.~2, no.~3, pp. 157--173, 2010.

\bibitem{zheng2014recent}
Z.~Zheng, ``Recent developments and research needs in modeling lane changing,'' \emph{Transportation research part B: methodological}, 2014.

\bibitem{tang2009macroscopic}
T.-Q. Tang, S.~Wong, H.-J. Huang, and P.~Zhang, ``Macroscopic modeling of lane-changing for two-lane traffic flow,'' \emph{Journal of Advanced Transportation}, vol.~43, no.~3, pp. 245--273, 2009.

\bibitem{jin2010macroscopic}
W.-L. Jin, ``Macroscopic characteristics of lane-changing traffic,'' \emph{Transportation research record}, vol. 2188, no.~1, pp. 55--63, 2010.

\bibitem{rahman2013review}
M.~Rahman, M.~Chowdhury, Y.~Xie, and Y.~He, ``Review of microscopic lane-changing models and future research opportunities,'' \emph{IEEE transactions on intelligent transportation systems}, vol.~14, no.~4, 2013.

\bibitem{halati1997corsim}
A.~Halati, H.~Lieu, and S.~Walker, ``Corsim-corridor traffic simulation model,'' in \emph{Traffic Congestion and Traffic Safety in the 21st Century: Challenges, Innovations, and Opportunities Urban Transportation Division; Highway Division, ASCE; Federal Highway Administration,; and National Highway Traffic Safety Administration, USDOT.}, 1997.

\bibitem{ahmed1996models}
K.~Ahmed, M.~Ben-Akiva, H.~Koutsopoulos, and R.~Mishalani, ``Models of freeway lane changing and gap acceptance behavior,'' \emph{Transportation and traffic theory}, vol.~13, 1996.

\bibitem{knoop2014calibration}
V.~L. Knoop and C.~Buisson, ``Calibration and validation of probabilistic discretionary lane-change models,'' \emph{IEEE Transactions on Intelligent Transportation Systems}, vol.~16, no.~2, 2014.

\bibitem{zhang2019game}
Q.~Zhang, R.~Langari, H.~E. Tseng, D.~Filev, S.~Szwabowski, and S.~Coskun, ``A game theoretic model predictive controller with aggressiveness estimation for mandatory lane change,'' \emph{IEEE Transactions on Intelligent Vehicles}, vol.~5, no.~1, pp. 75--89, 2019.

\bibitem{kamal2015efficient}
M.~A.~S. Kamal, S.~Taguchi, and T.~Yoshimura, ``Efficient vehicle driving on multi-lane roads using model predictive control under a connected vehicle environment,'' in \emph{2015 IEEE Intelligent Vehicles Symposium (IV)}.\hskip 1em plus 0.5em minus 0.4em\relax IEEE, 2015, pp. 736--741.

\bibitem{dixit2019trajectory}
S.~Dixit, U.~Montanaro, M.~Dianati, D.~Oxtoby, T.~Mizutani, A.~Mouzakitis, and S.~Fallah, ``Trajectory planning for autonomous high-speed overtaking in structured environments using robust mpc,'' \emph{IEEE Transactions on Intelligent Transportation Systems}, vol.~21, no.~6, 2019.

\bibitem{lenz2016tactical}
D.~Lenz, T.~Kessler, and A.~Knoll, ``Tactical cooperative planning for autonomous highway driving using monte-carlo tree search,'' in \emph{2016 IEEE Intelligent Vehicles Symposium (IV)}.\hskip 1em plus 0.5em minus 0.4em\relax IEEE, 2016, pp. 447--453.

\bibitem{lauri2022partially}
M.~Lauri, D.~Hsu, and J.~Pajarinen, ``Partially observable markov decision processes in robotics: A survey,'' \emph{IEEE Transactions on Robotics}, vol.~39, no.~1, pp. 21--40, 2022.

\bibitem{ding2021epsilon}
W.~Ding, L.~Zhang, J.~Chen, and S.~Shen, ``Epsilon: An efficient planning system for automated vehicles in highly interactive environments,'' \emph{IEEE Transactions on Robotics}, vol.~38, no.~2, pp. 1118--1138, 2021.

\bibitem{ulbrich2013probabilistic}
S.~Ulbrich and M.~Maurer, ``Probabilistic online pomdp decision making for lane changes in fully automated driving,'' in \emph{16th International IEEE Conference on Intelligent Transportation Systems}.\hskip 1em plus 0.5em minus 0.4em\relax IEEE, 2013.

\bibitem{piu2022stability}
\BIBentryALTinterwordspacing
M.~Piu and G.~Puppo, ``Stability analysis of microscopic models for traffic flow with lane changing,'' pp. 495--518, 2022. [Online]. Available: \url{https://www.aimsciences.org/article/id/623853e62d80b7560acc4a96}
\BIBentrySTDinterwordspacing

\bibitem{wu2017multi}
C.~Wu, E.~Vinitsky, A.~Kreidieh, and A.~Bayen, ``Multi-lane reduction: A stochastic single-lane model for lane changing,'' in \emph{2017 IEEE 20th International Conference on Intelligent Transportation Systems (ITSC)}.\hskip 1em plus 0.5em minus 0.4em\relax IEEE, 2017, pp. 1--8.

\bibitem{dubins1957curves}
L.~E. Dubins, ``On curves of minimal length with a constraint on average curvature, and with prescribed initial and terminal positions and tangents,'' \emph{American Journal of mathematics}, vol.~79, no.~3, 1957.

\bibitem{kong2015kinematic}
J.~Kong, M.~Pfeiffer, G.~Schildbach, and F.~Borrelli, ``Kinematic and dynamic vehicle models for autonomous driving control design,'' in \emph{2015 IEEE intelligent vehicles symposium (IV)}.\hskip 1em plus 0.5em minus 0.4em\relax IEEE, 2015, pp. 1094--1099.

\bibitem{daganzo2002behavioral}
C.~F. Daganzo, ``A behavioral theory of multi-lane traffic flow. part ii: Merges and the onset of congestion,'' \emph{Transportation Research Part B: Methodological}, vol.~36, no.~2, pp. 159--169, 2002.

\bibitem{jin2010kinematic}
W.-L. Jin, ``A kinematic wave theory of lane-changing traffic flow,'' \emph{Transportation research part B: methodological}, 2010.

\bibitem{oh2015impact}
S.~Oh and H.~Yeo, ``Impact of stop-and-go waves and lane changes on discharge rate in recovery flow,'' \emph{Transportation Research Part B: Methodological}, vol.~77, pp. 88--102, 2015.

\bibitem{michalopoulos1984multilane}
P.~G. Michalopoulos, D.~E. Beskos, and Y.~Yamauchi, ``Multilane traffic flow dynamics: some macroscopic considerations,'' \emph{Transportation Research Part B: Methodological}, vol.~18, no. 4-5, pp. 377--395, 1984.

\bibitem{orosz2010traffic}
G.~Orosz, R.~E. Wilson, and G.~St{\'e}p{\'a}n, ``Traffic jams: dynamics and control,'' pp. 4455--4479, 2010.

\bibitem{lazar2016review}
H.~Lazar, K.~Rhoulami, and D.~Rahmani, ``A review analysis of optimal velocity models,'' \emph{Periodica Polytechnica Transportation Engineering}, vol.~44, no.~2, pp. 123--131, 2016.

\bibitem{khalil2001nonlinear}
K.~K. Hassan, ``Nonlinear systems,'' \emph{Pearson, 3rd edition}, 2001.

\bibitem{goebel2009hybrid}
R.~Goebel, R.~G. Sanfelice, and A.~R. Teel, ``Hybrid dynamical systems,'' \emph{IEEE control systems magazine}, vol.~29, no.~2, pp. 28--93, 2009.

\bibitem{barone1977floquet}
S.~Barone, M.~Narcowich, and F.~Narcowich, ``Floquet theory and applications,'' \emph{Physical Review A}, vol.~15, no.~3, p. 1109, 1977.

\end{thebibliography}

\end{document}